\documentclass[11pt,a4paper,reqno]{amsart}
\usepackage[utf8]{inputenc}
\usepackage[english]{babel}
\usepackage[margin=1.0in,tmargin=0.9in,bmargin=0.80in]{geometry}
\usepackage{amsmath}
\usepackage{graphicx}
\usepackage{amsfonts}
\usepackage{amssymb}
\usepackage{adjustbox}
\usepackage{float}
\usepackage{tikz}
\usetikzlibrary{matrix,chains,positioning,decorations.pathreplacing,arrows}
\usepackage{dsfont}
\usepackage{booktabs}
\usepackage{pgfpages}
\usepackage{enumerate}
\usepackage{subfigure}
\usepackage{todonotes}

\usepackage{scalerel,stackengine}
\stackMath

\newcommand\reallywidehat[1]{%
\savestack{\tmpbox}{\stretchto{%
  \scaleto{%
    \scalerel*[\widthof{\ensuremath{#1}}]{\kern-.6pt\bigwedge\kern-.6pt}%
    {\rule[-\textheight/2]{1ex}{\textheight}}
  }{\textheight}%
}{0.5ex}}%
\stackon[1pt]{#1}{\tmpbox}%
}

\usepackage{hyperref}
\usepackage{xcolor}
\hypersetup{
    colorlinks,
    linkcolor={red!50!black},
    citecolor={blue!50!black},
    urlcolor={blue!80!black}
}

\usepackage{amsthm}

\usepackage{times}

\numberwithin{equation}{section}

\numberwithin{table}{section}

\theoremstyle{plain}

\theoremstyle{definition}

\numberwithin{equation}{section}  
\newtheorem{defn}{Definition}[section]
\newtheorem{exa}[defn]{Example}
\newtheorem{rem}[defn]{Remark}

\newtheorem{thm}[defn]{Theorem}
\newtheorem{prop}[defn]{Proposition}

\newtheorem{lem}[defn]{Lemma}

\newcommand{\G}{\mathbb{G}}
\newcommand{\h}{\mathbb{K}}
\newcommand{\K}{\mathbb{K}}
\newcommand{\R}{\mathbb{R}}
\newcommand{\Q}{\mathbb{Q}}

\newcommand{\N}{\mathbb{N}}

\newcommand{\PP}{\mathbb{P}}

\newcommand{\E}{\mathbb{E}}
\newcommand{\cB}{\mathcal{B}}
\newcommand{\cC}{\mathcal{C}}
\newcommand{\cD}{\mathcal{D}}
\newcommand{\cF}{\mathcal{F}}
\newcommand{\cH}{\mathcal{H}}
\newcommand{\cI}{\mathcal{I}}
\newcommand{\cK}{\mathcal{K}}
\newcommand{\cM}{\mathcal{M}}

\newcommand{\cQ}{\mathcal{Q}}
\newcommand{\cS}{\mathcal{S}}

\newcommand{\one}{ 1 \hspace{-3pt} \mathrm{l}} 
\newcommand{\eqd}{\stackrel{\mathrm{d}}=}
\newcommand{\de}{\mathsf{\,d}}
\newcommand{\diag}{\mathop{\mathrm{diag}}\nolimits}

\DeclareFontFamily{U}{mathx}{\hyphenchar\font45}
\DeclareFontShape{U}{mathx}{m}{n}{
      <5> <6> <7> <8> <9> <10>
      <10.95> <12> <14.4> <17.28> <20.74> <24.88>
      mathx10
      }{}
\DeclareSymbolFont{mathx}{U}{mathx}{m}{n}
\DeclareMathSymbol{\bigtimes}{1}{mathx}{"91}
\title[Improved Multi-Asset Price Bounds]{Improved Robust Price Bounds for Multi-Asset Derivatives under Market-Implied Dependence Information}

\author[Jonathan Ansari, Eva L\"{u}tkebohmert, Ariel Neufeld, Julian Sester]{Jonathan Ansari$^{1}$, Eva L\"{u}tkebohmert$^{1}$
, Ariel Neufeld$^{2}$, Julian Sester$^{3}$
}

\begin{document}
\maketitle

\begin{center}
\normalsize{\today} \\ \vspace{0.2cm}
\small\textit{$^{1}$Department of Quantitative Finance,\\ Institute for Economic Research, University of Freiburg,\\ Rempartstr. 16,
79098 Freiburg, Germany.\\[1mm]
$^{2}$NTU Singapore, Division of Mathematical Sciences,\\ 21 Nanyang Link, Singapore 637371. \\ [1mm]
$^{3}$ National University of Singapore, Department of Mathematics,\\ 21 Lower Kent Ridge Road, 119077.                                                                                                                             }
\end{center}

\begin{abstract}
We show how inter-asset dependence information derived from market prices of options can lead to improved model-free price bounds for multi-asset derivatives.  Depending on the type of the traded option,  we either extract correlation information or we derive restrictions on the set of admissible copulas that capture the inter-asset dependencies. To compute the resultant price bounds for some multi-asset options of interest, we apply a modified martingale optimal transport approach. Several examples based on simulated and real market data illustrate the improvement of the obtained price bounds and thus provide evidence for the relevance and tractability of our approach.\\[0mm]


\noindent
{\bf Keywords:} 
multi-asset options, model-free pricing, quasi-copulas, correlation, dependence information
\end{abstract}

\section{Introduction}

%

In recent years model-free valuation approaches for exotic derivatives attracted enormous attention. In such approaches the aim is to determine arbitrage-free price bounds for an exotic, and therefore not liquidly traded, option $\Phi$ while imposing no assumptions on the dynamics or probability distributions of a potential underlying stochastic model of the financial market. Put differently, to price $\Phi$ one allows for all arbitrage-free pricing models and associated pricing measures $\mathbb{Q}$, and computes the extreme prices for $\Phi$ as minimial and maximal expectations $\E_\Q[\Phi]$ among these models, resulting in a range of arbitrage-free prices (see e.g. \cite{acciaio2016model, beiglbock2013model,burzoni2017model, cheridito2021martingale, cheridito2017duality,davishobson2007}).  
In this way, the model-free pricing approach respects the not-quantifiable \emph{Knightian uncertainty} \cite{knight1921risk} of having chosen a wrong financial model for option valuation, which is particularly important in periods in which financial models calibrated to historical data do not depict the real behaviour of the market appropriately, for instance due to unforeseen financial crises.

However,
a major drawback of the model-free approach evidently is that the resultant range of possible arbitrage-free prices for $\Phi$ turns out to be too large and therefore the usefulness of the original model-free pricing approach in practice is limited, see also \cite{dolinsky2018super, neufeld2018buy}.
To decrease the range of possible arbitrage-free prices, one follows an inverse approach by inferring information from the market and then reducing the set of admissible models to those models that are consistent with the considered information. This market information usually is related to market beliefs (\cite{bartl2019duality,cheridito2017duality,hou2018robust,neufeld2021model}), or to the prices of liquidly traded options either directly (\cite{ acciaio2016model, brown2001robust, burzoni2017model, cheridito2017duality, hobson1998robust, hobson2012model,neufeld2020model}) or indirectly through the marginal distributions that are derived from liquidly traded options by using the Breeden--Litzenberger result (\cite{breeden1978prices, neufeld2022numerical, talponen2014note}). The latter case relates to the so called martingale optimal transport problem (\cite{beiglbock2013model,beiglbock2016problem, beiglbock2017complete,cheridito2021martingale, dolinsky2014martingale, guo2019computational,guo2021path,guo2019local, henry2017model,liu2019compactness,tan2013optimal} to name but a few).

In this paper, we combine copula theory and martingale optimal transport to construct improved price bounds for multi-asset derivatives. To this end, we first utilize the well-known relationship between prices of certain derivatives depending on multiple underlying assets at a single time point and expectation operators defined in dependence of copulas and quasi-copulas (see \cite{Ansari-2021, Bernard-2019, Lux-2017,papapantoleon2020detection, Tankov-2011}). If the payoff function of the derivative fulfils certain monotonicity properties ($\Delta$-monotonicity, $\Delta$-antitonicity, or supermodularity), then the extreme expectations can be associated to so-called Fr\'{e}chet-bounds (see, for instance, \cite{Nelsen-2006,Rueschendorf-1980,Rueschendorf-1981}) which then represent model-free price bounds for these derivatives. By following the approaches pursued in  \cite{bartl2017marginal,Lux-2017,papapantoleon2020detection,Tankov-2011} one can further restrict the class of admissible (quasi-) copulas through the inclusion of additional market information which then leads to improved price bounds for the financial derivative of interest in a single period framework.  We extend this single-period model-free pricing approach relying on copula theory to a multi-period setting by connecting it with model-free pricing approaches relying on martingale optimal transport theory.

Our paper contributes to the literature on model-independent pricing in various aspects. First, we show how price information on certain derivatives, that depend on multiple assets, leads to restrictions on possible inter-asset dependencies expressed either in terms of correlations if prices of \emph{basket options} are observable, or through restrictions on the set of admissible copulas if prices of options with \(\Delta\)-monotone or \(\Delta\)-antitone payoff function are observable. These market implied dependence restrictions can be translated into linear equality and inequality constraints specifying the set of admissible pricing measures. To that end, we prove in Proposition \ref{theaddcon} a model-independent super-hedging duality result which allows to include these additional constraints and thereby adapts results in \cite{acciaio2016model, beiglbock2013model,burzoni2017model,cheridito2017duality} among others to our setting.
This contribution can be seen in line with the various approaches that were recently established to improve model-free price bounds, see~\cite{eckstein2019martingale, lutkebohmert2019tightening, neufeld2021model, sester2019robust}. For improvements in the multi-asset case we refer to the recent contributions \cite{aquino2019bounds, eckstein2019robust,neufeld2020model, papapantoleon2020detection}. In \cite{neufeld2020model}  algorithms were developed to exactly compute price bounds using market implied information in a single-period model.  While most of the mentioned approaches yield tighter price bounds mainly through restrictions on admissible pricing measures based on the distributions of single underlying securities,  our approach includes restrictions imposed on the inter-asset dependencies.

 Second, we utilize the common component dependence model approach, in which one assumes that the inter-asset dependencies of all assets with respect to a specific reference asset are known or can be derived from price information. In this situation the maximal inter-asset dependencies can no longer be described by copulas and we therefore use the concept of quasi-copulas. In Theorem \ref{theqcub}, as a main result of our paper, we characterize the supermodular ordering of upper products in common component dependence models which enables us to derive price bounds of a broad class of multi-asset derivatives with supermodular payoff functions as analytical expressions in dependence of the limiting quasi-copulas. Our result is based on an application of a multivariate integration by parts formula (see \cite{Ansari-2021}) and generalizes \cite[Theorem 1]{Ansari-Rueschendorf-2020} to quasi-copulas.
 
Finally, we provide several numerical examples based on simulated and real data to illustrate the significant improvement of price bounds when inter-asset dependencies are taken into account. More specifically,  we show in many relevant cases how upper and lower price bounds can be substantially tightened when the set of admissible pricing measures is reduced due to market-implied dependencies. Note that in the common component dependence model approach, knowledge of only a few prices written on pairs of assets can already considerably improve the price bounds for options written on several underlyings. \\

The remainder of the paper is as follows. In Section~\ref{section_main_results} we present the underlying setting and derive an adjusted model-free pricing-hedging duality. Section~\ref{sec dependence modelling} introduces the concept of copulas and quasi-copulas and the most important associated results. In Section~\ref{sec improved price bounds} we explain how we can use price information of traded derivatives to derive restrictions on the set of resulting pricing measures and compatible copulas. In Section~\ref{section_numerics} we provide several examples illustrating how model-free price bounds can be computed within our approach and how existing conventional price bounds can be improved. The proofs of all mathematical statements are provided in Section~\ref{section_proofs}. 


\section{Setting and Duality Result}\label{section_main_results}

The underlying problem of the present article is a model-independent approach to the pricing of financial derivatives depending on several assets.
At time $t_0\in \R$, we consider a financial market with $d\in \N \cap [2,\infty)$ securities with non-negative values $S_0^1,\dots,S_0^d\in \R_+$, and we denote by  $S:=(S_{t_i}^k)^{k=1,\dots,d}_{i=1,\dots,n}$ their future values at times $t_1<t_2<\dots<t_n$ for $n\in \N$. We model $S$ by the canonical process on $(\R_+^{n d},\mathcal{B}(\R_+^{n d}))$, where $\mathcal{B}(\R_+^{n d})$ denotes the Borel-$\sigma$-algebra on $\R_+^{n d}$, i.e., the components of $S$ are defined via
$
S_{t_i}^k:(x_{1}^1,\dots,x_{n}^d) \mapsto x_i^k.
$
For simplicity, we normalize interest rates to zero and assume absence of dividends. This means, $S_{t_i}^k$ denotes the price of the $k$-th security at time $t_i$.
Further, we fix some payoff function $c:\R_+^{n d} \rightarrow \R$ of a financial derivative depending on $S$. 
Our goal is to calculate an arbitrage-free price interval for $c$ in a model-independent way, i.e., by using only information that is implied by market prices without imposing any assumptions on the dynamics or joint distributions of $S$. Therefore, we proceed as follows to define our set of pricing measures.
\begin{enumerate}
\item[(i)] First, we observe for all $k=1,\dots,d$, $i=1,\dots,n$ prices of European call options written on the $k$-th security maturing at $t_i$ for a continuum of strikes. We refer to such options as \textit{liquidly traded} options. According to \cite{breeden1978prices} we can then infer the one-dimensional risk-neutral marginal distributions \(\mu_i^k\) of \(S_{t_i}^k\) from this data for all $i,k$. This means for all $i,k$ and for any pricing measure $\Q$ that we have $\Q \circ {S_{t_i}^k}^{-1}=\mu_i^k$, where  $\mu_i^k$ has mean $S_0^k \in \R_+$.
Denote for each such $\mu=(\mu_i^k)_{1\leq i \leq n}^{1\leq k \leq d}$ by 
$$
\Pi(\mu):=\left\{\Q\in \mathcal{P}(\R_+^{nd})~\middle|~\Q \circ {S_{t_i}^k}^{-1}=\mu_i^k \text{ for all } i,k \right\}
$$
the set of transport plans that consists of all Borel probability measures on \(\R_+^{nd}\), denoted by $\mathcal{P}(\R_+^{nd})$, with univariate marginals \(\mu_1^1,\ldots,\mu_n^d\) having finite first moments equal to $S_0^1,\dots,S_0^d$. Further, we denote by $F_i^k(\cdot) =  \int_{-\infty}^ \cdot  \,\de \mu_i^k$ the cumulative distribution function of $\mu_i^k$.
\item[(ii)] Moreover, to ensure absence of model-independent-arbitrage\footnote{in the sense of \cite[Definition 1.2.]{acciaio2016model}}, we assume that for every pricing measure $\Q$ the martingale property
\begin{equation}\label{eq_martingaleprop}
\E_\Q[S_{t_i}|S_{t_j},\dots,S_{t_1}]=S_{t_j} ~~~\Q\text{-a.s. and for all } t_j \leq t_i
\end{equation}
holds true, where \((S_{t_i})_{0\leq i\leq n}\) is the \(d\)-variate process with components \(S_{t_i}=(S_{t_i}^{k})^{k=1,\ldots,d}\,\) for $i=0,1,\dots,n$.
%
According to a straightforward extension of \cite[Lemma 2.3]{beiglbock2013model}, the equality in \eqref{eq_martingaleprop} may be rewritten as 
$
\int_{\R^{nd}_+} \delta(x_1^1,\dots,x_j^d)(x_{j+1}^k-x_j^k) \,\de\Q(x_1^1,\dots,x_n^d)=0 
$
for all $k=1,\dots,d$, $j=0,\dots,n-1$ and $\delta \in C_b(\R^{jd}_+)$, which is the class of continuous and bounded functions on \(\R_+^{jd}\,.\)
We denote by $\mathcal{M}(\mu)\subset \Pi(\mu) \subset \mathcal{P}(\R_+^{nd})$ the set of martingale measures on $\R^{nd}_+$ with fixed univariate marginal distributions $\mu = (\mu_i^k)_{1\leq i \leq n}^{1\leq k \leq d}$.
\item[(iii)] Besides the marginal distributions and the martingale property, we impose additional linear constraints that are implied by observations on the market. These constraints additionally restrict the dependence structure of the underlying assets $S$. More precisely, we consider linear equality constraints of the form 
\begin{equation}\label{eq_eq_constraints}
\E_\Q[f_i^{\operatorname{eq}}(S)]=K_i^{\operatorname{eq}}
\end{equation}
for problem-tailored Borel-measurable functions $f_i^{\operatorname{eq}} :\R_+^{n d} \rightarrow \R$ and $K_i^{\operatorname{eq}} \in \R$ with $i$ in some index set $\mathcal{I^{\operatorname{eq}}}$. We will adjust the choices of $f_i^{\operatorname{eq}}$ to the specific problems. Additionally, we implement inequality constraints of the form 
\begin{equation}\label{eq_ineq_constraints}
\E_\Q[f_i^{\operatorname{ineq}}(S)]\leq K_i^{\operatorname{ineq}}
\end{equation}
for Borel-measurable $f_i^{\operatorname{ineq}} :\R_+^{n d} \rightarrow \R$, $K_i^{\operatorname{ineq}} \in \R$, and $i \in \mathcal{I}^{\operatorname{ineq}}$.
\end{enumerate}
The set of measures which fulfil these additional constraints is denoted by 
\begin{equation}\label{eq_restriction}
\begin{aligned}
\mathcal{M}^{\mathrm{lin}}_{f_i^{\operatorname{eq}},K_i^{\operatorname{eq}},\mathcal{I^{\operatorname{eq}}},\atop f_i^{\operatorname{ineq}},K_i^{\operatorname{ineq}},\mathcal{I^{\operatorname{ineq}}}}(\mu):=\mathcal{M}(\mu) &\cap \left\{\Q\in \mathcal{P}(\R_+^{nd})~\middle|~\E_\Q[f_i^{\operatorname{eq}}(S)]=K_i^{\operatorname{eq}} \text{ for all } i\in \mathcal{I}^{\operatorname{eq}}\right\}\\
&\cap \left\{\Q \in \mathcal{P}(\R_+^{nd})~\middle|~\E_\Q[f_i^{\operatorname{ineq}}(S)]\leq K_i^{\operatorname{ineq}} \text{ for all } i\in \mathcal{I}^{\operatorname{ineq}}\right\}\,.
\end{aligned}
\end{equation}
For the sake of readability we abbreviate this set by $\mathcal{M}^{\mathrm{lin}}$ and consider it as our set of pricing measures.

An arbitrage-free and model-independent upper price bound for a payoff $c$ under the above mentioned equality and inequality constraints can then be obtained by pursuing two different approaches. First, in the primal approach, we consider the supremum of the expected values among all martingale measures consistent with available price information on options and the imposed equality and inequality constraints given by
\begin{equation}\label{equ primal problem M_lin}
\overline{P}_{\mathcal{M}^{\mathrm{lin}}}:=\sup_{\Q\in \mathcal{M}^{\mathrm{lin}}} \E_\Q [c(S)].
\end{equation}
A corresponding lower price bound \(\underline{P}_{\mathcal{M}^{\mathrm{lin}}}\) can be obtained as the infimum over all measures $\Q\in\mathcal{M}^{\mathrm{lin}}$.

Second, in the dual approach, the price bounds of the derivative $c$ can be calculated by using trading strategies instead of pricing measures. For the upper bound we consider the problem of finding the cheapest super-replication price of $c$. More precisely, we consider strategies $\Psi_{(u_i^k),(\delta_i^k),(\alpha_i),(\beta_i)}:\R_{+}^{nd} \rightarrow \R$ of the form
\begin{equation}\label{eq_dual_strats1}
\begin{aligned}
\Psi_{(u_i^k),(\delta_i^k),(\alpha_i),(\beta_i)}(x_1^1,\dots,x_n^d):=\sum_{k=1}^d\sum_{i=1}^n u_i^k(x_i^k)&+\sum_{i=1}^{n-1}\sum_{k=1}^d\delta_i^k(x_1^1,\dots,x_i^d) (x_{i+1}^k-x_i^k)\\
&\hspace{-5cm}+\sum_{i\in \mathcal{I}^{\operatorname{eq}}}\alpha_i\left(f_i^{\operatorname{eq}}(x_1^1,\dots,x_n^d)-K_i^{\operatorname{eq}}\right)+\sum_{i\in \mathcal{I}^{\operatorname{ineq}}}\beta_i\left(f_i^{\operatorname{ineq}}(x_1^1,\dots,x_n^d)-K_i^{\operatorname{ineq}}\right)
\end{aligned}
\end{equation}
with 
$
u_i^{k} \in \mathfrak{C}:=\bigg\{u\colon \R_+ \to \R~\bigg|~ u(x)=a+bx+\sum_{i=1}^m \lambda_i(x-d_i)_+\,, a,b,\lambda_i,d_i\in \R\,, m\in \N \bigg\}
$
and with each $\delta_{i}^k \in C_b(\R_+^{id})$, $\alpha_i \in \R$, $\beta_i \in \R_+$ such that $\alpha_i=0$, $\beta_j=0$ for all but finitely many $i\in \mathcal{I}^{\operatorname{eq}}$, $j \in \mathcal{I}^{\operatorname{ineq}}$.
This means, we consider trading strategies allowing for static positions in the European options $u_i^k$, in derivatives with payoffs $f_i^{\operatorname{eq}}$ traded for a price $K_i^{\operatorname{eq}}$ and long positions in $f_i^{\operatorname{ineq}}$ traded for a price not higher than $K_i^{\operatorname{ineq}}$. Moreover, we consider dynamic self-financing trading positions $\delta_i^k$ in the underlying securities. In line with our model-free approach, we are interested in strategies which super-replicate the payoff of the derivative pointwise, i.e., for every possible path, independent of any associated probability. We use the notation $f\geq c$ to express pointwise inequalities, i.e., $f(x) \geq c(x)$ for all $x \in \R^{nd}_+$.  The following result shows that - under mild assumptions - minimizing the prices of such super-replication strategies yields the same value as maximizing expectations w.r.t.\,measures from $\mathcal{M}^{\mathrm{lin}}$.
To this end, we set
\begin{align*}
\mathcal{S}:=\bigg\{\Psi_{(u_i^k),(\delta_i^k),(\alpha_i),(\beta_i)} ~\bigg|~ &\exists~u_i^{k} \in \mathfrak{C},\delta_{i}^k \in C_b(\R_+^{id}),\alpha_i \in \R,\beta_i \in \R_+ \text{with } \alpha_i=0, \beta_j=0\\
&\text{ for all but finitely many }i\in \mathcal{I}^{\operatorname{eq}},~ j \in \mathcal{I}^{\operatorname{ineq}}\text{s.t. } \Psi_{(u_i^k),(\delta_i^k),(\alpha_i),(\beta_i)} \geq c \bigg\},
\end{align*}
and denote by 
$
\underline{\mathcal{D}}_{\mathcal{S}}(c):=\inf_{\Psi\in \mathcal{S}}\left\{\sum_{k=1}^d \sum_{i=1}^n\E_{\mu_i^k}[u_i^k]\right\}
$
the minimal price among all super-replicating strategies \(\Psi:=\Psi_{(u_i^k),(\delta_i^k),(\alpha_i),(\beta_i)}\) for the payoff $c$. When there is no ambiguity about the payoff $c$, we abbreviate the notation as $\underline{\mathcal{D}}_{\mathcal{S}}$.
We define  for $m \in \N$ by
\begin{align*}
C_{\operatorname{lin}}(\R_+^m)&:= \left\{f\colon \R_+^m \to \R~\middle|~ f\text{ continuous},~ \sup_{(x_1,\dots,x_m)\in \R_+^m} \frac{|f(x_1,\dots,x_m)|}{1+\sum_{i=1}^m x_i}  < \infty \right\}\,,\\
L_{\operatorname{lin}}(\R_+^m)&:=\left\{f\colon \R_+^m \to \R~\middle|~ f\text{ lower semicontinuous},~\sup_{(x_1,\dots,x_m)\in \R_+^m} \frac{|f(x_1,\dots,x_m)|}{1+\sum_{i=1}^m x_i}  < \infty \right\}\,,\\
U_{\operatorname{lin}}(\R_+^m)&:=\left\{f\colon \R_+^m \to \R~\middle|~ f\text{ upper semicontinuous},~\sup_{(x_1,\dots,x_m)\in \R_+^m} \frac{|f(x_1,\dots,x_m)|}{1+\sum_{i=1}^m x_i}  < \infty \right\}
\end{align*}
the set of continuous, lower semicontinuous, and upper semicontinuous functions, respectively, with at most linear growth. We then adapt the model-independent super-hedging duality results from, e.g., \cite{acciaio2016model}, \cite{beiglbock2013model}, \cite{burzoni2017model}, \cite{cheridito2021martingale}, \cite{dolinsky2014martingale}, \cite{Rueschendorf-2019}, and \cite{zaev2015monge}, to our situation by formulating the following proposition.

\begin{prop}[Duality with additional constraints]\label{theaddcon}
\label{thm_duality_linear_constraints} ~\\ Assume that $c \in U_{\operatorname{lin}}(\R^{nd}_+)$, $f_i^{\operatorname{eq}}\in C_{\operatorname{lin}}(\R_+^{nd}),\) and \(f_j^{\operatorname{ineq}} \in L_{\operatorname{lin}}(\R_+^{nd})$ for all $i \in \mathcal{I}^{\operatorname{eq}}$, $j \in \mathcal{I}^{\operatorname{ineq}}$, and assume that $\mathcal{M}^{\mathrm{lin}}\neq \emptyset$. Then it holds 
\begin{align}\label{gendures}
\overline{P}_{\mathcal{M}^{\mathrm{lin}}}= \underline{\mathcal{D}}_{\mathcal{S}}\,.
\end{align}
Moreover, there exists \(\Q\in \mathcal{M}^{\mathrm{lin}}\) such that
\begin{align}\label{eqmaxele}
\overline{P}_{\mathcal{M}^{\mathrm{lin}}} = \E_\Q [c(S)]\,.
\end{align}
\end{prop}

\begin{rem}
\label{rem_differences_martingale_property}
\begin{enumerate}~
\item[(a)] The set $\mathcal{M}(\mu)$ of martingale measures with fixed univariate marginals is non-empty if and only if the $d$-dimensional marginals $(\mu_i^1,
\dots,\mu_i^d)_{i=1,\dots,n}$ increase in convex order\footnote{A finite set of probability measures $\{\PP_1,\cdots,\PP_n\}$ on $\R^d$ is said to increase in convex order if $\int_{\R^d} f(x) \de \PP_i(x) \leq \int_{\R^d} f(x) \de \PP_{i+1}(x)$ for all convex functions $f:\R^d \rightarrow \R$ and for all $i=1,\dots,n-1$ such that the integrals are finite.}, see \cite{strassen1965existence}. Thus, the latter condition is necessary for the non-emptiness of $\mathcal{M}^{\mathrm{lin}}(\mu)$.
To derive sufficient conditions for the non-emptiness of the set $\mathcal{M}^{\mathrm{lin}}$ we proceed as follows in the case that both $\mathcal{I}^{\operatorname{ineq}}$ and $\mathcal{I}^{\operatorname{eq}}$ are at most countable. 
Assume $\mathcal{M}(\mu) \neq \emptyset$ and w.l.o.g. $\mathcal{I}^{\operatorname{eq}}=\N$, $\mathcal{I}^{\operatorname{ineq}}=\N$. We observe that if 
$
\inf_{\Q \in \mathcal{M}(\mu)}\E_\Q[f_1^{\operatorname{eq}}]\leq K_1^{\operatorname{eq}}\leq\sup_{\Q \in \mathcal{M}(\mu)} \E_\Q[f_1^{\operatorname{eq}}],
$
then we get as a convex combination of the minimal and maximal measure the existence of a measure $\Q \in \mathcal{M}(\mu)$ with $\E_\Q[f_1^{\operatorname{eq}}]=K_1^{\operatorname{eq}}$. We proceed inductively, and see that if for all $i=2,3,\dots,$ we have
\[
\inf_{\Q \in \mathcal{M}^{\mathrm{lin}}_{f_i^{\operatorname{eq}},K_i^{\operatorname{eq}},\{1,\dots,i-1\},\atop f_i^{\operatorname{ineq}},K_i^{\operatorname{ineq}},\emptyset}(\mu)} \E_\Q[f_i^{\operatorname{eq}}] 
 \leq K_i^{\operatorname{eq}} \leq \sup_{\Q \in \mathcal{M}^{\mathrm{lin}}_{f_i^{\operatorname{eq}},K_i^{\operatorname{eq}},\{1,\dots,i-1\},\atop f_i^{\operatorname{ineq}},K_i^{\operatorname{ineq}},\emptyset}(\mu)} \E_\Q[f_i^{\operatorname{eq}}],
\]
then it holds that$\mathcal{M}^{\mathrm{lin}}_{f_i^{\operatorname{eq}},K_i^{\operatorname{eq}},
\mathcal{I}^{\operatorname{eq}},\atop f_i^{\operatorname{ineq}},K_i^{\operatorname{ineq}},\emptyset}(\mu) \neq \emptyset$.
In the same way we can check consistency of the inequality constraints. Thus, if we have 
$
\sup_{\Q \in \mathcal{M}^{\mathrm{lin}}_{f_i^{\operatorname{eq}},K_i^{\operatorname{eq}},\mathcal{I}^{\operatorname{eq}},\atop f_i^{\operatorname{ineq}},K_i^{\operatorname{ineq}},\emptyset}(\mu)} \E_\Q[f_1^{\operatorname{ineq}}]\leq K_1^{\operatorname{ineq}},
$
and for all $i=2,3,\dots,$ that
$
\sup_{\Q \in \mathcal{M}^{\mathrm{lin}}_{f_i^{\operatorname{eq}},K_i^{\operatorname{eq}},\mathcal{I}^{\operatorname{eq}},\atop f_i^{\operatorname{ineq}},K_i^{\operatorname{ineq}},\{1,\dots,i-1\}}(\mu)} \E_\Q[f_i^{\operatorname{ineq}}]\leq  K_i^{\operatorname{ineq}},
$
then $\mathcal{M}^{\operatorname{lin}}\neq \emptyset$ holds.
\item[(b)] \label{rem_differences_martingale_propertyB} In the one-period case, i.e., if \(n=1\,,\) the martingale property \eqref{eq_martingaleprop} only constrains the marginal distributions of \(S_{t_1}^k\) for \(k=1,\ldots,d\,\) but not the dependence structure of \((S_{t_1}^1,\ldots,S_{t_1}^d)\,,\) because \eqref{eq_martingaleprop} simplifies to the mean constraint $\E_\Q[S_{t_1}^k]=S_0^k$ for some deterministic values $S_0^k \in \R_+$ for $k =1,\dots,d$ representing today's spot values of the respective securities.
\end{enumerate}
\end{rem}
A major contribution of this paper is the specification of situations, where the (countable) choices of $f_i^{\operatorname{eq}},~K_i^{\operatorname{eq}},~\mathcal{I}^{\operatorname{eq}}$ and $f_i^{\operatorname{ineq}},~K_i^{\operatorname{ineq}},~\mathcal{I}^{\operatorname{ineq}}$, respectively, can explicitly be inferred from market data. These specifications are given in detail in Section~\ref{sec improved price bounds} and require the concept of copulas and quasi-copulas which will be introduced in the following Section~\ref{sec dependence modelling}.

 As a main result of our paper, we characterize for a finite number of bivariate copulas their pointwise maximum (which is in general only a quasi-copula) by the supermodular comparison of copula products where we extend the supermodular ordering to the class of quasi-copulas; we refer to Theorem~\ref{theqcub}. This result is meaningful because dependence information that improves the upper Fr\'{e}chet bound yields in general only a quasi-copula as improved pointwise upper bound for the dependence structure, see, e.g., \cite{Bernard-2017,Lux-2017,Rueschendorf-2019,Nelsen-2001,Tankov-2011}. As an application, we determine in Section \ref{secadddep} improved price bounds for European basket options under dependence information related to the structure of a common component dependence model.


\section{Dependence Modelling}\label{sec dependence modelling}

In this section, we introduce the basic notions for copulas, quasi-copulas, and dependence orderings which we use in Section \ref{sec improved price bounds} to restrict the inter-asset dependencies when prices of specific financial derivatives are given.
Further, as a main result we extend with Theorem \ref{theqcub} the characterization of the supermodular ordering of upper products in common component dependence models from \cite[Theorem 1]{Ansari-Rueschendorf-2020} to the case of a quasi-copula as upper bound. We build on this result in Section \ref{sec improved price bounds} to derive closed-form expressions for improved price bounds of supermodular payoff functions like basket options, when dependence information related to the setting of a common component dependence model is available.

\subsection{Basic Notions}\label{section_copula_notions}\label{section_copula_main_results}

For the analysis of dependence structures, we consider some well-known function classes. We often focus on functions with non-negative domain because we are only interested in those functions which can be interpreted as payoff functions depending on the non-negative underlying assets. Denote by \(\overline{\R}_+:=\R_+\cup\{\infty\}\) the extended real non-negative numbers.

An \(m\)-variate \emph{copula} is a distribution function on \([0,1]^m\) with uniform univariate marginal distributions. 
Due to Sklar's theorem, every 
\(m\)-variate distribution function \(F\) can be decomposed into an \(m\)-copula \(C\) and its univariate marginal distribution functions \(F_1,\ldots,F_m\) by
\begin{equation}\label{eq_sklar}
F(x)=C(F_1(x_1),\ldots,F_m(x_m)),~~~x=(x_1,\ldots,x_m)\in \R^m\,.
\end{equation}
The copula \(C\) is uniquely determined on \(\mathfrak{Ran}(F_1)\times \cdots \times \mathfrak{Ran}(F_m)\,,\) where \(\mathfrak{Ran}(F_i)\) denotes the range of \(F_i\,.\) Conversely, for every copula \(C\) and for arbitrary univariate distribution functions \(F_1,\ldots,F_m\,,\) the function \(F\) defined by the right-hand side of \eqref{eq_sklar} is an \(m\)-variate distribution function, see, e.g., \cite{Durante-2016} and \cite{Nelsen-2006}. Denote by \(\cC_m\) the class of \(m\)-copulas, and by \(W^m\) and \(M^m\) the lower  and upper Fr\'{e}chet bound which are defined by \(W^m(u):=\max\left\{\sum_{i=1}^m u_i -m+1,0\right\}\) and \(M^m(u):=\min_{1\leq i\leq m} \{u_i\}\,,\) respectively,
for \(u=(u_1,\ldots,u_m)\in [0,1]^m\,.\) It holds that 
\begin{equation}\label{deffrebounds}
W^m(u)\leq C(u)\leq M^m(u)~~~ \text{for all } u\in [0,1]^m
\end{equation}
and for all copulas \(C\in \cC_m\). The bounds are sharp, where \(M^m\) is a copula for all \(m\in \N\) and where \(W^m\) is a copula only for \(m\leq 2\,,\) see \cite{Nelsen-2006} and \cite{Rueschendorf-1981}.

Knowledge of values or bounds of some functionals, for example, knowledge of prices or price bounds as given by the linear functionals in \eqref{eq_eq_constraints} and \eqref{eq_ineq_constraints}, may restrict the class of copulas to a subclass of dependencies. If the functionals are consistent w.r.t. the lower orthant ordering (i.e., w.r.t. the pointwise ordering of distribution functions), then the Fr\'{e}chet bounds \(W^m\) and \(M^m\) may be improved w.r.t. the lower orthant ordering, see \cite{Lux-2017}, where, however, the improved bounds are generally no longer copulas but quasi-copulas defined as follows.
Denote by \(\cF^m\) the class of \(m\)-variate distribution functions, and by \(\cF^m_+\) the subclass of \(m\)-variate distribution functions \(F\) defined by \eqref{eq_sklar} such that \(F_i(0)=0\) for all \(i\in \{1,\ldots,m\}\,.\)

\begin{defn}[Quasi-copula, quasi-distribution function]\label{def_quasicop}
Let \(m\in \N\,\).\footnote{We also allow the trivial case \(m=1\) which is considered in the definition of the quasi-expectation operator, see \eqref{defquexpop}.}
\begin{itemize}
\item[(a)]
A function \(Q\colon [0,1]^m\to [0,1]\) is called \(m\)-variate \emph{quasi-copula} if
\begin{enumerate}[(i)]
\item \label{def_quasicop1} \(Q\) is \emph{grounded}, i.e., \(Q(u)=0\) if at least one coordinate of \(u\) is \(0\,,\)
\item \label{def_quasicop1b}\(Q\) has uniform marginals, i.e., \(Q(1,\ldots,1,u_i,1,\ldots,1)=u_i\) for all \(u_i\in[0,1]\), \(1\leq i\leq m\,,\)
\item \label{def_quasicop2}\(Q\) is non-decreasing in each component, and
\item \label{def_quasicop3}\(Q\) fulfils the Lipschitz condition \(|Q(v)-Q(u)|\leq \sum_{i=1}^m |v_i-u_i|\) for all \(u=(u_1,\ldots,u_m), \) \(v=(v_1,\ldots,v_m)\in [0,1]^m\,.\)
\end{enumerate}
\item[(b)] We call a function \(H\colon \R^m\to [0,1]\) \emph{quasi-distribution function} if there exist \(F_1,\ldots,F_m\in \cF^1\) and an \(m\)-variate quasi-copula \(Q\) such that
\begin{align*}
H(x)=Q(F_1(x_1),\ldots,F_m(x_m))~~~\text{for all } x=(x_1,\ldots,x_m)\in \R^m\,.
\end{align*}
 We also say $H$ is the quasi-distribution function of $Q$ w.r.t.\,$F_1,\dots,F_m$.
\end{itemize}
We denote by \(\cQ_m\) the set of \(m\)-variate quasi-copulas, by \(\cH^m\) the class of \(m\)-variate quasi-distribution functions, and by \(\cH_+^m\subset \cH^m\) the subclass with \(F_i\in \cF_+^1\) for all \(i\in \{1,\ldots,m\}\,.\) 
\end{defn}
Note that \(\cC_m\subset \cQ_m\) and thus \(\cF^m\subset \cH^m\) as well as \(\cF^m_+\subset \cH^m_+\,.\) Moreover, the lower Fr\'{e}chet bound \(W^m\) is a quasi-copula for all \(m\,;\) as in the case of copulas, all quasi-copulas satisfy the bounds given in \eqref{deffrebounds}.
An important property of quasi-copulas is that for any subset \(\mathcal{N}\subset \cQ_m\,,\) the pointwise supremum $Q_{\mathcal{N}}$ defined by
\begin{align}\label{eqmxqc}
Q_{\mathcal{N}}(u):=\sup\{Q(u)\,|\, Q\in \mathcal{N}\},~\text{ for } u \in [0,1]^m,
\end{align}
is again a quasi-copula, see \cite[Theorem 2.2]{Nelsen-2004} (where the proof for \(m=2\) can be extended to arbitrary dimension \(m\geq 2\)). In contrast, the pointwise supremum of copulas is generally not a copula, which motivates using quasi-copulas.

For improving price bounds under dependence information, we make use of some dependence orderings defined on the class \(\cQ_m\) of quasi-copulas. To avoid technical difficulties, we will often consider functions \(f\colon \Theta^m\to \R\) that are continuous at the boundary of \(\Theta\,,\) where $\Theta\in \{\R_+,[0,1),[0,1]\}$, i.e., whenever \(a:=\inf(\Theta)\in \Theta\) and/or \(b:=\sup(\Theta)\in \Theta\,,\) respectively, then we will often require \(f\) to satisfy
\begin{align}\label{contboun1}
&\lim_{h\downarrow 0} f(x_1,\ldots,x_{i-1},a+h,x_{i+1},\ldots,x_m) = f(x_1,\ldots,x_{i-1},a,x_{i+1},\ldots,x_m)\\
&\nonumber \text{and/or} \\
\label{contboun2} &\lim_{h\downarrow 0} f(x_1,\ldots,x_{i-1},b-h,x_{i+1},\ldots,x_m) = f(x_1,\ldots,x_{i-1},b,x_{i+1},\ldots,x_m)\,,
\end{align}
respectively, for all \(i\in \{1,\ldots,m\}\) and for all \(x_j\in \Theta\,,\)  \(j\ne i\,.\)
 We define for a bounded function $f\colon \Theta^m \to \R$ which is in each component either increasing or decreasing, which fulfils \eqref{contboun2}, and which is defined on \(\Theta\in \{[0,1),[0,1],\R_+,\R\}\,,\) its \emph{survival function} \(\widehat{f}\colon \Theta^m\to \R\) by\footnote{We stress that the sum in the definition of the survival function in \eqref{defsurvfun} also takes into account the summand \(J=\emptyset\).} 
\begin{align}\label{defsurvfun}
\widehat{f}(x):=\sum_{J\subseteq \{1,\ldots,m\}} (-1)^{m-|J|} f(y)\,,
\end{align}
where \(y=(y_1,\ldots,y_m)\) satisfies for all \(i\in \{1,\ldots,m\}\) that \(y_i=\sup(\Theta)\) if \(i\in J\) and \(y_i=x_i\) if \(i\notin J\,\), and where we set
\begin{equation}\label{eq_def_limsup}
f(w_1,\ldots,w_{l-1}, \infty,w_{l+1},\ldots,w_m):=\lim_{z\to \infty} f(w_1,\ldots,w_{l-1},z,w_{l+1},\ldots,w_m).
\end{equation}
for all $(w_1,\dots,w_m) \in \Theta^m$, $l\in \{1,\dots,m\}$. 
We say \(\widehat{Q}\) is a \emph{quasi-survival function} if \(\widehat{Q}\) is the survival function of a quasi-copula \(Q\,\).
The orthant orders on \(\cQ_m\) are defined as follows, see \cite{Lux-2017}.

\begin{defn}[\(\leq_{\operatorname{lo}}\,,\) \(\leq_{\operatorname{uo}}\)]\label{def_orthant_orders}
Let $m \in \N$, and let $Q$, $Q'\in \cQ_m$ be quasi-copulas.
\begin{itemize}
\item[(a)] The quasi-copula \(Q\) is smaller than \(Q'\) in the \emph{lower orthant order}, written \(Q\leq_{\operatorname{lo}} Q'\,,\) if \(Q(u)\leq Q'(u)\) for all \(u\in [0,1]^m\,.\)
\item[(b)] The quasi-copula \(Q\) is smaller than \(Q'\) in the \emph{upper orthant order}, written \(Q\leq_{\operatorname{uo}} Q'\,,\) if \(\widehat{Q}(u)\leq \widehat{Q}'(u)\) for all \(u\in [0,1]^m\,.\)
\end{itemize}
\end{defn}

For a copula \(C \in \mathcal{C}_m\) and some measurable function $f:[0,1]^m \rightarrow \R$, we define the expectation operator \(\psi_f(C)\) as the Lebesgue-Stieltjes integral
\begin{align}\label{defexpop_0}
\psi_f(C):=&\int_{[0,1]^m} f(u) \de C(u).
\end{align}
Then, for a copula \(C \in \mathcal{C}_m\) associated by \eqref{eq_sklar} with the distribution function $F(\cdot)=\Q(X\leq \cdot)$, for some measure \(\Q\in \mathcal{P}(\R^{m}_+)\,,\) we obtain
\begin{align}\label{defexpop}
\begin{split}
\psi_c^{(F_1,\ldots,F_m)}(C) := &\psi_{c\,\circ (F_1^{-1},\ldots,F_m^{-1})}(C) 
= \int_{\R_+^m} c(x)\de F(x) = \E_\Q [c(X)]\,,
\end{split}
\end{align}
applying the transformation formula for Stieltjes integrals, see, e.g., \cite[Theorem (2)]{Winter-1997}. Hence, in financial contexts, the expectation operator $\psi_c^{(F_1,\ldots,F_m)}(C) $ can be interpreted as the price of \(c\) under some pricing measure \(\Q\,\).
 
For copulas \(C_1,C_2\in \cC_m\,,\) it is well-known that the orthant orders are characterized by\footnote{For a function \(f\colon {\overline{\R}}_{+}^m\supset \bigtimes_{i=1}^m[a_i,b_i] \to \R\,,\) the difference operator \(\triangle_\varepsilon^i\,,\) \(\varepsilon>0\,,\) \(1\leq i \leq m\,,\) is defined by \(\triangle_\varepsilon^if(x):=f((x+\varepsilon e_i)\wedge b)-f(x),\)
where $x\in \bigtimes_{i=1}^m[a_i,b_i]$, $b=(b_1,\dots,b_m)$,  \(e_i\) denotes the \(i\)-th unit vector, and \(\wedge\) denotes the componentwise minimum.
Then \(f\) is \emph{\(\Delta\)-monotone} / \emph{\(\Delta\)-antitone} if
\(\triangle^{i_1}_{\varepsilon_1}\cdots \triangle_{\varepsilon_k}^{i_k} f(x)\geq 0\,,\) resp., \((-1)^k\triangle^{i_1}_{\varepsilon_1}\cdots \triangle_{\varepsilon_k}^{i_k} f(x)\geq 0\,\) for all $x \in \bigtimes_{i=1}^m[a_i,b_i]$, $k \in \{1,\dots, m\}$, any subset \(J=\{i_1,\ldots,i_k\}\subseteq \{1,\ldots,m\}\,,\) and all \(\varepsilon_1,\ldots,\varepsilon_k>0\,.\) Note that, by definition, a \(\Delta\)-monotone function is componentwise increasing and a \(\Delta\)-antitone function is componentwise decreasing.}
\begin{align}\label{chaloord}
\begin{split}
C_1\leq_{\operatorname{lo}} C_2 ~~~ &\Longleftrightarrow ~~~ \psi_f(C_1)\leq \psi_f(C_2)~~~\text{for all }\Delta\text{-antitone functions }f\colon [0,1]^m\to \R\,,\\
C_1\leq_{\operatorname{uo}} C_2 ~~~ &\Longleftrightarrow ~~~ \psi_f(C_1)\leq \psi_f(C_2)~~~\text{for all }\Delta\text{-monotone functions }f\colon [0,1]^m\to \R\,,
\end{split}
\end{align}
such that the expectations exist, see \cite{Rueschendorf-1980,Rueschendorf-2013} and \cite{Mueller-Stoyan-2002}.
An important extension of the orthant orders on \(\cC_m\) is the \emph{supermodular ordering} defined by\footnote{\(f\colon {\overline{\R}}_{+}^m\supset \bigtimes_{i=1}^m[a_i,b_i] \to \R\) is \emph{supermodular} if \(\triangle^{i}_{\varepsilon_i} \triangle_{\varepsilon_j}^{j} f(x)\geq 0\) for all \(x\in \bigtimes_{i=1}^m[a_i,b_i]\,,\) \(i\ne j\) and \(\varepsilon_i,\varepsilon_j > 0\,.\)}
\begin{align}\label{supmodcop}
C_1\leq_{\operatorname{sm}} C_2 ~~~\colon \Longleftrightarrow ~~~ \psi_f(C_1)\leq \psi_f(C_2) ~~~\text{for all supermodular functions } f\colon [0,1]^m \to \R
\end{align}
such that the expectations exist. Note that every \(\Delta\)-monotone or \(\Delta\)-antitone function is supermodular and, thus, \(C_1\leq_{sm} C_2\) implies \(C_1\leq_{lo} C_2\) and \(C_1\leq_{uo} C_2\,.\) Denote by \(\cF_{\Delta}\,\)  the set of \(\Delta\)-monotone functions, by \(\cF_{\Delta}^-\) the set of \(\Delta\)-antitone functions, and by  \(\cF_{\operatorname{sm}}\) the set of supermodular functions.

An extension of the right-hand side of \eqref{chaloord} and of \eqref{supmodcop} to quasi-copulas is not immediate since a quasi-copula does not in general induce a signed measure (see \cite{Nelsen-2010}) and therefore the integral in \eqref{defexpop_0} may not be defined. However, by an application of an integration by parts formula for measure-inducing functions, an extension to quasi-copulas may be achieved, which enables to determine bounds for the price in \eqref{defexpop} under dependence information as discussed before.

Therefore, denote by $\cB(\Theta^m)$ the Borel \(\sigma\)-algebra on \(\Theta^m\,.\) A left-continuous, resp., right-continuous function\footnote{We call a multivariate function left-continuous/right-continuous if the function is componentwise left-continuous/componentwise right-continuous at every point.} \(g\colon \Theta^m \to \R\), \(\Theta\in \{[0,1),[0,1],\R_+\}\), is said to be \emph{measure-inducing} if there exists a signed measure \(\eta_g\) on $\cB(\Theta^m)$ such that
\begin{align}\label{eqindmeas}
\eta_g([x_1,x_1+\varepsilon_1)\times \cdots \times [x_m,x_m+\varepsilon_m)) &= \triangle_{\varepsilon_1}^1 \cdots \triangle_{\varepsilon_m}^m g(x)\,, ~~~\text{resp.,}\\
\label{eqindmeasu}
\eta_g((x_1,x_1+\varepsilon_1]\times \cdots \times (x_m,x_m+\varepsilon_m]) &= \triangle_{\varepsilon_1}^1 \cdots \triangle_{\varepsilon_m}^m g(x)
\end{align}
for all \(x=(x_1,\ldots,x_m)\in \Theta^m\) and \(\varepsilon_1,\ldots,\varepsilon_m\in \R_+\,.\)

Further, for \(I=\{i_1,\ldots,i_k\}\subseteq \{1,\ldots,m\}\,,\) the \emph{\(I\)-marginal} \(g_I\) of \(g\) is defined by
\begin{align}\label{deffmarg}
g_I\colon \Theta^{k} \ni (u_{i_1},\ldots,u_{i_k}) \to g(u_1,\ldots,u_m)\,, ~~~\text{where }u_j=0 ~\text{for all }j\notin I\,.
\end{align}

In particular, by \eqref{eqindmeas} and \eqref{eqindmeasu}, for all \(I\subseteq\{1,\ldots,m\}\,,\) \(I\ne \emptyset\,,\) the \(I\)-marginal of every left-/ right-continuous \(\Delta\)-monotone function induces a non-negative measure, and the \(I\)-marginal of every left-/right-continuous \(\Delta\)-antitone function induces a non-negative measure on \(\cB([0,1]^{|I|})\) if \(|I|\) is even and a non-positive measure if \(|I|\) is odd, see, e.g., \cite[Corollary 2.24]{Ansari-2021}, cf. \cite[Proposition 2.1]{papapantoleon2020detection}.
In general, a left-/right-continuous measure-inducing function defined on a compact domain like \([0,1]^m\) always induces a finite signed measure, and if all \(I\)-marginals are measure-inducing, this equivalently means that the function has bounded Hardy-Krause variation, see \cite[Theorem 2.12]{Ansari-2021}, see also \cite[Theorem 3 and, for a precise definition of the Hardy-Krause variation, Section 2.1]{Aistleitner-2015}\footnote{In this reference, the authors show the statement for a right-continuous function \(f\colon [0,1]^m\to \R\,.\) However, it also applies for a left-continuous function \(g\colon [0,1]^m \to \R\,\) by setting \(g(x)=f(1-x)\,.\)}.

We will often assume  that for $\Theta\in \{\R_+,[0,1),[0,1]\}$ a measure-inducing function \(f\colon \Theta^m\to \R\) is continuous at the boundary of \(\Theta\). To this end, we denote by
\begin{align}\begin{split}
\cF_{\operatorname{mi}}^{\operatorname{c},\operatorname{l}}(\Theta):=&\{f\colon \Theta^m \to \R\mid f \text{ satisfies } \eqref{contboun1} \text{ and } \eqref{contboun2},\\ 
&~~~~~~~~~~~~~~~~~~~~~~~~  f_I \text{ is measure-inducing for all } I\subseteq \{1,\ldots,m\}\,, I\ne \emptyset\}
\end{split}
\end{align}
the class of measure-inducing functions that satisfy the continuity conditions at the boundary of the domain and for which all \(I\)-marginal functions are measure-inducing.

\begin{defn}[Quasi-expectation]~\label{def_quasi_ex}\\
Let \(H=Q\circ (F_1,\ldots,F_m):\R_+^m \rightarrow [0,1]\) be a quasi-distribution function of $Q \in \mathcal{Q}_m$ w.r.t.\,$F_1,\dots,F_m \in \mathcal{F}_+^1$ and let \(c\in \cF_{\operatorname{mi}}^{\operatorname{c},\operatorname{l}}(\R_+^m)\) be left-continuous. Then, the \emph{quasi-expectation} of \(c\) w.r.t.\,\(H\) is defined by
\begin{align}\label{defquexpop}
\int_{\R_+^m} c(x) \de H(x):=~ & \pi_c(\widehat{H}) :=\sum_{I\subseteq \{1,\ldots,m\} \atop I\ne \emptyset} \int_{\R_+^{|I|}} \widehat{H}_I(x)\de \eta_{c_I}(x) + c(0,\ldots,0)\\
\nonumber=& \sum_{k=1}^m\sum_{I\subseteq \{1,\ldots,m\} \atop I=\{i_1,\ldots,i_k\}} \int_{\R_+^{|I|}} \widehat{Q}_I(F_{i_1}(x_1),\ldots,F_{i_k}(x_{i_k}))\de \eta_{c_I}(x_1,\ldots,x_k) + c(0,\ldots,0)\,,
\end{align}
whenever the integrals exist.
We also write
\begin{align}\label{defquexpop1}
\pi_c^{\mu}(\widehat{Q}):= \pi_c^{(F_1,\ldots,F_m)}(\widehat{Q}):=\pi_c(\widehat{H}) = \pi_c(\widehat{Q}\circ (F_1,\ldots,F_m))
\end{align}
where, for \(\mu=(\mu_1,\ldots,\mu_m)\,,\) \(F_i\) is the distribution function of the marginal distribution \(\mu_i \in \mathcal{P}(\R_+)\,,\) \(1\leq i \leq m\,.\)
\end{defn}
Note that for \(1\leq i \leq m\,,\) every \(F_i\in \cF_+^1\) is, by definition of \(\cF_+^1\), increasing and fulfils \(F_i(0)=0\) which implies that \(\widehat{H}_I=\widehat{Q}_I\circ (F_1,\ldots,F_m)\) and thus proves the last equality in \eqref{defquexpop1}.

\begin{rem}\label{propconncomdom}
\begin{enumerate}[(a)]
\item By a multivariate integration by parts formula, for every copula \(C\in \cC_m\) and thus by \eqref{eq_sklar} for every distribution function \(F\in \cF_+^m\) it holds true that
\begin{align}\label{qeopexopeq}
\psi_f(C) = \pi_f(\widehat{C}) && \text{and} && \int_{\R_+^m} c(x) \de F(x) = \psi_c(F)= \pi_c(\widehat{F}) \,,
\end{align}
whenever \(f\in \cF_{\operatorname{mi}}^{\operatorname{c},\operatorname{l}}([0,1]^m)\) and \(c\in \cF_{\operatorname{mi}}^{\operatorname{c},\operatorname{l}}(\R_+^m)\) are left-continuous. More generally, for every left-continuous function \(\xi\in \cF_{\operatorname{mi}}^{\operatorname{c},\operatorname{l}}(\Theta^m)\) and for every bounded, grounded, right-continuous, and measure-inducing function \(h\colon \Theta^m\to \R\) satisfying continuity conditions \eqref{contboun1} and \eqref{contboun2}, it holds true that
\begin{align}\label{eqqop}
\psi_\xi(h)=\pi_\xi(\widehat{h})
\end{align}
whenever the integrals exist,
see \cite[Theorem 3.1]{Ansari-2021}, with \(\psi_\xi\) and \(\pi_\xi\) defined as in \eqref{defexpop_0} and \eqref{defquexpop}.
\item \label{propconncomdom2} For every left-continuous, measure-inducing function \(c\colon \R_+^m \to \R\) and for all distribution functions {\(F_1,\ldots,F_m\in \cF_+^1\)} such that \(g:=c\circ (F_1^{-1},\ldots,F_m^{-1})\in \cF_{\operatorname{mi}}^{\operatorname{c},\operatorname{l}}([0,1)^m)\,,\) a sufficient condition for the existence of 
\begin{align}\label{reptraf}
\pi_c(\widehat{Q}\circ \left(F_1,\ldots,F_m)\right)= \pi_c\left(\reallywidehat{Q\circ(F_1,\ldots,F_m)}\right)=\pi_g(\widehat{Q}) 
\end{align}
is that \(u\mapsto g_I(u,\ldots,u)\) is Lebesgue-integrable on \([0,1)\) for all \(I\subseteq\{1,\ldots,m\}\,,\) \(I\ne \emptyset\,,\) see  \cite[Corollary 3.12]{Ansari-2021}. Note that the last equality in \eqref{reptraf} holds true by the transformation formula for the integration by parts operator which is similar to \eqref{defexpop} given by
\begin{align}\label{margtrafgs}
\pi_c(h\circ (F_1,\ldots,F_m))=\pi_{c\,\circ(F_1^{-1},\ldots,F_m^{-1})}(h)
\end{align}
for all 
functions \(g:=c\circ(F_1^{-1},\ldots,F_m^{-1})\in \cF_{\operatorname{mi}}^{\operatorname{c},\operatorname{l}}([0,1)^m)\) 
and for all distribution functions \(F_1,\ldots,F_m\in \cF_+^1\,,\) whenever the integrals exist, see \cite[Proposition 3.6]{Ansari-2021}.
\end{enumerate}
\end{rem}

When lower and upper orthant dependence restrictions are imposed by quasi-copulas, analytic expressions for improved price bounds w.r.t. lower and upper orthant payoff functions are given by the characterizations
\begin{align}\label{charortord}
Q\leq_{\operatorname{lo}} Q'~~~\Longleftrightarrow~~~ \pi_f(\widehat{Q})\leq \pi_f(\widehat{Q'})~\text{for all left-continuous }f\in \cF_\Delta^-\,,\\
\nonumber Q\leq_{\operatorname{uo}} Q'~~~\Longleftrightarrow~~~ \pi_f(\widehat{Q})\leq \pi_f(\widehat{Q'})~\text{for all left-continuous }f\in \cF_\Delta\,,
\end{align}
such that the expectations exist, see \cite[Theorem 5.5]{Lux-2017}. This, in particular, implies for a copula \(C\) with \(C\leq_{\operatorname{lo}} Q'\) resp. \(C\leq_{\operatorname{uo}} Q'\) that \(\psi_f(C)\leq \pi_f(\widehat{Q})\) for \(f\) \(\Delta\)-antitone/-monotone.
In Theorem \ref{corsmifm}, we also derive analytic expressions for improved price bounds w.r.t.\,supermodular payoff functions, which requires an extension of the supermodular ordering to the class \(\cQ_m\,.\) To that end, denote by \(C^k([0,1]^m)\equiv C^k\,,\) \(k\in \N\cup\{\infty\}\,,\) the class of functions \(f\colon [0,1]^m\to \R\) such that all (mixed) partial derivatives of order \(k\) exist and are continuous.
We make use of the following lemma.

\begin{lem}[\cite{Ansari-2021}, Corollary 2.19]\label{lemmind}
Let \(f\colon [0,1]^m\to \R\) be a \(C^m\)-function. Then \(f\in \cF_{\operatorname{mi}}^{\operatorname{c},\operatorname{l}}([0,1]^m)\,.\)
\end{lem}

Since the supermodular ordering on \(\cC^m\) is generated by the class \(C^\infty\cap \cF_{\operatorname{sm}}\) of smooth supermodular functions, see \cite[Theorem 3.2]{Denuit-2002}, and since all \(I\)-marginals of a smooth function induce a signed measure due to Lemma \ref{lemmind}, we can extend the supermodular ordering to \(\cQ_m\) as follows.

\begin{defn}[Supermodular ordering for quasi-copulas]\label{def_sm_qc}~\\
Let \(Q,Q'\in \cQ_m\,.\) Then \(Q\) is said to be smaller than \(Q'\) in the supermodular ordering, written \(Q\leq_{\operatorname{sm}} Q'\,,\) if \(\pi_f(\widehat{Q})\leq \pi_f(\widehat{Q'})\) for all supermodular and left-continuous functions $f\in \cF_{\operatorname{mi}}^{\operatorname{c},\operatorname{l}}([0,1]^m)$.
\end{defn}

Note that \(f\) in the above definition is defined on a compact domain and thus induces a finite measure. Hence, \(\pi_f(\widehat{Q})\) is finite for every quasi-copula \(Q\in \cQ_m\,,\) compare \cite[Theorem 3.7]{Ansari-2021}.

\subsection{Common component dependence models}\label{sec IFM}


In this section we give a brief overview on the notion of the upper product of bivariate copulas which describes the worst case dependence structure (w.r.t.\,\(\leq_{\operatorname{sm}}\)) in partially specified factor models (PSFMs). This allows to establish Theorem \ref{theqcub} which yields improved price bounds for supermodular payoff functions. 
In Section \ref{sec improved price bounds}, we build on Theorem \ref{theqcub}  which extends the upper product ordering result in \cite[Theorem 1]{Ansari-Rueschendorf-2020} to quasi-copulas, to derive an
upper price bound for a supermodular payoff function which improves the comonotone standard bound\footnote{For \(\R\)-valued random variables \(X_1,\ldots,X_m\) with distribution functions \(F_1,\ldots,F_m\,,\) the vector \(X^c:=(X_1^c,\ldots,X_m^c):=(F_1^{-1}(U),\ldots,F_m^{-1}(U))\,,\) with \(U\) uniformly distributed on \((0,1)\,,\) is comonotone. It holds that \(X_i^c\eqd X_i\) for all \(i\) and the copula of \(X^c\) is the upper Fr{\'e}chet copula \(M^m\,.\) Hence, \(X^c\) has the maximal distribution (w.r.t.\,the supermodular ordering) in the class of all distributions with fixed marginals \(F_1,\ldots,F_m\,,\) and thus it is referred to as (comonotone) standard upper bound.} based on knowledge of the marginal distributions.

In a PSFM, a random vector \(X=(X_i)_{1\leq i\leq m}:=(f_i(Z,\varepsilon_i))_{1\leq i\leq m}\,,\) is expressed through Borel-measurable functions $f_i \colon \R^{2}\to\R$
of an \(\R\)-valued random variable \(Z\) and \(\R\)-valued random variables \(\varepsilon_i\) for $i=1,\dots,m$,  where the common (risk) factor \(Z\) is w.l.o.g.\ assumed to be independent of \((\varepsilon_i)_{1\leq i\leq m}\). Moreover, the bivariate distributions \((X_i,Z)_{1\leq i\leq m}\) are specified, i.e., the univariate distributions of \(X_i\) and \(Z\) as well as the common copula \(D^i\) are known for all \(1\leq i \leq m\,.\) However, in contrast to the usual independence assumption, the dependence structure among the vector of idiosyncratic risks \((\varepsilon_i)_{1\leq i \leq m}\) is not specified, see \cite{Bernard-2017}.
The maximal random vector (w.r.t. \(\leq_{\operatorname{sm}}\)) in the PSFM is given by the conditionally comonotone vector \(X_Z^c:=\left(F_{X_i|Z}^{-1}(U)\right)_{1\leq i \leq m}\), where $F_{X_i|Z}^{-1}$ denotes the generalized inverse of the conditional distribution function of $X_i$ given $Z$ and where \(U\) is uniformly distributed on \((0,1)\) and independent of \(Z\,.\) If \(Z\) has a continuous distribution function, the copula of \(X_Z^c\), in the sense of \eqref{eq_sklar}, is given by the upper product \(\bigvee_{i=1}^m D^i \equiv D^1\vee \cdots \vee D^m\) of the bivariate copulas \(D^1,\ldots,D^m\) which is an \(m\)-copula defined by
$
\bigvee_{i=1}^m D^i(u) := \int_0^1 \min_{1\leq i\leq m}\{\partial_2 D^i(u_i,t)\} \de t~~~ \text{for } u=(u_1,\ldots,u_m)\in [0,1]^m\,,
$
see \cite{Ansari-Rueschendorf-2018},
where \(\partial_2\) denotes the partial derivative w.r.t.\,the second component.

In a (partially specified) \emph{common component dependence model}\footnote{This model is in \cite{Ansari-Rueschendorf-2020} denoted as \emph{internal factor model}.} (CCD model), it is assumed that the factor \(Z=X_1\) is a component of the vector \((X_1,\ldots,X_m)\) and, thus, the first bivariate dependence constraint \(D^1\) is imposed by the upper Fr\'{e}chet copula, i.e., \(D^1=M^2\,.\) For a bivariate copula \(E\in \cC_2\,,\) the class of CCD models with dependence specifications \(D^k\leq_{\operatorname{lo}} E\,,\) \(k=2,\ldots,m\,,\) has a greatest element w.r.t. the supermodular ordering given by the \(m\)-variate upper product
\begin{align}\label{eqcococo}
M^2\vee E\vee \cdots \vee E \,(u)=E\left(\min_{2\leq i \leq m}\{u_i\},u_1\right)\,,
\end{align}
for \(u=(u_1,\ldots,u_m)\in [0,1]^m\,,\)
see \cite[Theorem 1]{Ansari-Rueschendorf-2020}.

As a main result, Theorem \ref{theqcub} extends this result to dependence specifications \(D^k\leq_{\operatorname{lo}} Q_2\,,\) \(k=2,\ldots,m\,,\) for a fixed given bivariate quasi-copula \(Q_2\in \cQ_2\) which serves as an upper bound for each $D^2\,,\ldots,D^m\,$. Thus, if $Q_2$ is implied by market price information, this allows us to incorporate bivariate dependence information inferred from market prices. This is of particular relevance as,  according to \cite[Theorem 3.1]{Lux-2017}, such price information often only corresponds to a quasi-copula that serves as an upper bound for the dependence structure.

We use that whenever a left-continuous function \(f\colon [0,1]^m\to \R\) is supermodular and componentwise increasing/componentwise decreasing and \(m\geq 2\,,\) the function \(\phi_f\colon [0,1]^2\to \R\) defined by 
\begin{align}\label{defphif}
\phi_f(x_1,x_2):=f(x_2,x_1,\ldots,x_1)
\end{align}
is \(\Delta\)-monotone/-antitone and, thus, for all \(I\subseteq\{1,\ldots,m\}\,,\) \(I\ne \emptyset\,,\) the \(I\)-marginal induces a positive measure by \eqref{eqindmeas}, see, e.g. \cite[Corollary 2.24]{Ansari-2021}. This enables us to establish an \(m\)-variate quasi-copula \(Q^*\) as an upper bound (w.r.t.\,the supermodular ordering) for the upper products \(M^2\vee D^2\vee \cdots \vee D^m\,,\) \(D^k\leq_{\operatorname{lo}} Q_2\,,\) which describe the worst case dependence structures in CCD models. The quasi-copula \(Q^* \in \mathcal{Q}_m\), 
\begin{align}\label{deqcqs}
Q^*(u):=Q_2\big(\min_{2\leq i \leq m}\{u_i\},u_1\big)\,,~~~ u=(u_1,\ldots,u_m)\in [0,1]^m\,,
\end{align}
relates to the conditionally comonotone structure in \eqref{eqcococo} and can be associated with the two-dimensional case. The non-intuitive arrangement of the arguments in \eqref{defphif} and \eqref{deqcqs} is due to the definition of the upper product where the copulas in the integrand are differentiated w.r.t.\,the second component.

The main result of this section is the characterization provided in the following theorem. 



\begin{thm}\label{theqcub}
Let \(D^2,\ldots,D^m\in \cC_2\) be bivariate copulas and \(Q_2\in \cQ_2\) a bivariate quasi-copula. Then, \(Q^*\) defined by \eqref{deqcqs}
is a quasi-copula, and the following statements are equivalent.
\begin{itemize}
\item[(a)] \label{theqcub1}\(D^i\leq_{\operatorname{lo}} Q_2\) for all \(2\leq i\leq m\,,\)
\item[(b)] \label{theqcub2}\(M^2\vee D^2\vee \ldots\vee D^m \leq_{\operatorname{lo}} Q^*\,,\)
\item[(c)] \label{theqcub2a}\(M^2\vee D^2\vee \ldots\vee D^m \leq_{\operatorname{uo}} Q^*\,,\)
\item[(d)] \label{theqcub2b}\(M^2\vee D^2\vee \ldots\vee D^m \leq_{\operatorname{c}} Q^*\,,\)
\item[(e)] \label{theqcub3}\(M^2\vee D^2\vee \ldots\vee D^m \leq_{\operatorname{sm}} Q^*\,.\)
\item[(f)] \label{theqcub6}\(\psi_f({M^2\vee D^2\vee \cdots \vee D^m})\leq \pi_{\phi_f}(\widehat{Q_2})\) \\
for all 
left-continuous, supermodular functions \(f\colon [0,1)^m\to \R\) which are componentwise increasing/componentwise decreasing such that \((\phi_f)_I\) is Lebesgue integrable on \([0,1)^{|I|}\) for \(I\subseteq\{1,2\}\,,\) \(I\ne \emptyset\,,\) where \(\phi_f\) is defined by \eqref{defphif}, and such that $f$ is lower bounded by some function which is integrable w.r.t.\,$M^2\vee D^2\vee \ldots\vee D^m$.
\end{itemize}
\end{thm}

\begin{rem}\label{remmaithean1}
\begin{enumerate}[(a)]
\item[(a)] \label{remmaithean1a} In the special case that \(Q_2=E\in \cC_2\) is a copula, the upper bound \(Q^*\) in Theorem \ref{theqcub} simplifies by \eqref{deqcqs} and \eqref{eqcococo} to \(M^2\vee E\vee \cdots \vee E\,,\) which implies the result from
\cite[Theorem 1]{Ansari-Rueschendorf-2020}.
\item[(b)] \label{remmaithean1b} Let \(D^2,\ldots,D^m\in \cC_2\) be copulas and \(Q^2,\ldots,Q^m\in \cQ_2\) be quasi-copulas with \(D^i\leq_{\operatorname{lo}} Q^i\) for all \(i=2,\ldots, m\,.\) Then, \(Q_2(u):=\max_{i=2,\dots,m}\{Q^i(u)\}\,,\) \(u\in [0,1]^2\,,\) is by \eqref{eqmxqc} a quasi-copula. Hence, Theorem \ref{theqcub}~(e) implies that
$
M^2\vee D^2\vee \cdots \vee D^m \leq_{\operatorname{sm}} Q^*
$
with \(Q^*\) defined by \eqref{deqcqs}.
\item[(c)] Note that Theorem~\ref{theqcub}~(f) improves the standard bound \(\E[c(X_1^c,\ldots,X_m^c)]\) for the expectation \(\E[c(X_1,\ldots,X_m)]\) of a random vector \((X_1,\ldots,X_m)\) w.r.t.\,a continuous and supermodular payoff function \(c\colon [0,1]^m \to \R\) 
if for all $k$ the copula of $(X_1,X_k)$ is upper bounded by $Q_2\in \mathcal{Q}_2$ in the lower orthant order, even if \(c\) is not measure-inducing. 
An example of such a payoff function is the continuous and supermodular function \(c(u):=\left(\sum_{i=1}^m u_i-K\right)_+\,,\) \(u=(u_1,\ldots,u_m)\in [0,1]^m\,,\) for $K>0$. But \(c\) is measure-inducing only if \(m\leq 2\), which can be seen from the fact that the lower Fr\'{e}chet bound \(W^m\) induces a signed measure if and only if \(m\leq 2\,,\) see \cite[Theorem 2.4]{Nelsen-2010}. However, since \(\phi_c\) given by \(\phi_c(u_1,u_2)=\left((m-1)u_1+u_2-K\right)_+\) is measure-inducing, we obtain by Theorem \ref{theqcub}~(f) that \(\pi_{\phi_c}(\widehat{Q_2})\) is an upper bound for \(\psi_c({M^2\vee D^2\vee \cdots \vee D^m})\) if \(D^i\leq_{\operatorname{lo}} Q_2\) for all \(i\in \{2,\ldots,m\}\,,\) see also Lemma~\ref{lembasopt}. In Example~\ref{exasmpf} we apply the characterization in Theorem~\ref{theqcub}~(f) and determine an improved upper price bound for a basket call option under dependence information related to an common component dependence model.
\end{enumerate}
\end{rem}

\section{Improved Price Bounds under Dependence Information}\label{sec improved price bounds}

In this section, we make use of the notions and results of Section~\ref{sec dependence modelling} to derive price bounds on financial derivatives under dependence restrictions.
First, we consider the case of \(\Delta\)-monotone or \(\Delta\)-antitone payoff functions and derive price bounds that take into account upper and lower quasi-copula bounds which can be inferred from market prices of multi-asset derivatives. These bounds are derived analogously to the approach from \cite{Lux-2017}. Through an application of the duality in Proposition \ref{theaddcon}, we show how these price bounds can be significantly improved when the martingale property is incorporated as a linear constraint.
Further, we state an upper price bound for supermodular payoff functions, where we assume a common component dependence model with additional dependence information derived from multi-asset options which restrict the bivariate \((1,k)\)-marginal copula of \(S_{t_i}=(S_{t_i}^1,\ldots,S_{t_i}^d)\) for all $k=2,\ldots, d$, i.e. the copula \(C_{\Q_i^{1,k}}\) of the risk-neutral distribution \(\Q_i^{1,k}\) of \((S_{t_i}^1,S_{t_i}^k)\,,\) \(k=2,\ldots,d\,,\) under some \(\Q \in \cM(\mu)\).

 We demonstrate that already the information on prices on a small amount of multi-asset derivatives traded in the market reduces the arbitrage price bounds of the financial derivative under consideration. This is crucial, as multi-asset derivatives are typically only traded OTC and hence only limited information on their prices is available to a financial trader.



Finally, in Section~\ref{secadddepinreltocor}, we present various scenarios where dependence information related to the risk-neutral correlation is inferred from option prices and incorporated by linear constraints restricting the class \(\mathcal{M}(\mu)\) of martingale measures with fixed marginal distributions.

\subsection{Dependence Information Through Copulas}\label{secadddep}

For any $i=1,\ldots,n,$ $k=1,\ldots, d$, and for any probability measure $\Q\in\mathcal{M}(\mu)$, we recall that we denote by $F^k_i$ the univariate marginal distribution function of the component $S_{t_i}^k$ under \(\Q\,.\) By Sklar's Theorem, the multivariate distribution function $F_\Q=\Q(S\leq \cdot)$ is decomposed by 
\begin{align}\label{theSklar}
F_\Q(x)=C_{\Q}(F_1^1(x_{1}^1),\ldots,F_n^d(x_n^d)),\quad \mbox{for all}\quad x=(x_1^1,\ldots,x_n^d)\in\mathbb{R}^{nd}_+,
\end{align}
into the univariate marginal distribution functions $F_i^k$ and a copula $C_\Q\in \cC_{nd}$ which describes the dependence structure among \(S\) under \(\Q\,.\)
As we will show in Section \ref{section_numerics}, when traded market prices of appropriate multi-asset options can be observed in the market, this allows to infer restrictions on the dependence structure of the underlying assets through (pointwise) upper and lower quasi-copula bounds \(\underline{Q}\) and \(\overline{Q}\) on the copula $C_\Q$.
In the sequel we explain how these bounds can be included as inequality constraints in the Problem \eqref{equ primal problem M_lin} to derive improved price bounds for a given payoff $c$.

First, we discuss some relevant classes of models where the copulas of the models, with associated risk-neutral distributions in \(\cM(\mu)\), are restricted w.r.t.\,the lower and upper orthant ordering, respectively.

\begin{enumerate}[1)]
\item \textbf{Copula bounds w.r.t.\,the orthant orders}\\ 
Let \(\underline{Q},\overline{Q}\in \cQ_{nd}\) be \(nd\)-variate quasi-copulas such that \(\underline{Q}\leq_{\operatorname{lo}}  \overline{Q}\,.\) 
Then, we consider the class
\begin{align}\label{eqdefclloor}
\cM^{\operatorname{lo}}_{\underline{Q},\overline{Q}} := \left\{\Q\in \cM(\mu)\,\middle|\, \underline{Q} \leq_{\operatorname{lo}} C_{\Q}\leq_{\operatorname{lo}} \overline{Q}\right\}.
\end{align}

If the quasi-copula $\underline{Q}$ coincides with the lower  Fr\'{e}chet bound \(W^{nd}\), or if $\overline{Q}$ coincides with the upper  Fr\'{e}chet bound  \(M^{nd}\), then the dependence structure in \eqref{eqdefclloor} is only restricted from one-side. Moreover, if both quasi-copulas $\underline{Q}$ and $\overline{Q}$ coincide with the respective Fr\'{e}chet bounds, then no additional dependence restriction on the class \(\cM(\mu)\) is imposed and we have that \(\cM^{\operatorname{lo}}_{W^{nd},M^{nd}}=\cM(\mu)\,.\)
The following observation turns out to be crucial for the implementation of copula constraints via trading strategies in Theorem~\ref{theloob}.
\begin{lem}\label{lem_indicator_constraints}
Let \(\underline{Q},\overline{Q}\in \cQ_{nd}\) and define for $x \in \R^{nd}$ the functions $f_x:=\one_{\{\cdot \leq x\}},~g_x:=\one_{\{\cdot < x\}}$, \(\widetilde{f}_x:=\one_{\{\cdot\, > x\}},~\widetilde{g}_x:=\one_{\{\cdot\, \geq x\}}\) where the inequalities in the indicator functions are meant componentwise. Let $F_i^k(\cdot) =  \int_{-\infty}^ \cdot  \,\de \mu_i^k$, $1\leq k \leq d$, $1\leq i \leq n$, and let \(\mathds{Q_+}\) denote the set of non-negative rational numbers. Then the following holds.
\begin{itemize}
\item[(a)]
We have that  \(\underline{Q} \leq_{\operatorname{lo}} C_{\Q}\leq_{\operatorname{lo}} \overline{Q}\) for  $\Q \in \cM(\mu)$ is equivalent to a countable number of inequality constraints of the form
\begin{align}\label{eqineqconstr}
\begin{array}{l}
\phantom{-}\E_{\Q} \left[g_x(S)\right]\leq \phantom{-}\overline{Q}(F_1^1(x_1^1),\ldots,F_n^d(x_n^d))\\[2mm]
\E_{\Q} \left[-f_x(S)\right]\leq -\underline{Q}(F_1^1(x_1^1),\ldots,F_n^d(x_n^d))
\end{array}
 ~\forall\, x=\left(x_1^1,\ldots,x_n^d\right)\in \mathds{Q}_+^{nd}\,.
\end{align}
\item[(b)]
We have that \(\underline{Q} \leq_{\operatorname{uo}} C_{\Q}\leq_{\operatorname{uo}} \overline{Q}\) for  $\Q \in \cM(\mu)$ is equivalent to a countable number of inequality constraints of the form
\begin{align}\label{eqineqconstr_uo}
\begin{array}{l}
\phantom{-}\E_{\Q} \left[\widetilde{f}_x(S)\right]\leq \phantom{-}\overline{Q}(F_1^1(x_1^1),\ldots,F_n^d(x_n^d))\\[2mm]
\E_{\Q} \left[-\widetilde{g}_x(S)\right]\leq -\underline{Q}(F_1^1(x_1^1),\ldots,F_n^d(x_n^d))
\end{array}
 ~\forall\, x=\left(x_1^1,\ldots,x_n^d\right)\in \mathds{Q}_+^{nd}\,.
\end{align}
\end{itemize}
\end{lem}

On account of Lemma~\ref{lem_indicator_constraints} the inequality constraints \eqref{eqineqconstr} specify $(f_i^{\operatorname{ineq}})_{i\in I^{\operatorname{ineq}}}$ and $(K_i^{\operatorname{ineq}})_{i\in I^{\operatorname{ineq}}}$ in \eqref{eq_restriction}.
In the following Theorem \ref{theloob}, which is partly a consequence of Proposition~\ref{thm_duality_linear_constraints}, we identify the upper price bound \(\overline{P}_{\cM^{\operatorname{lo}}_{\underline{Q},\overline{Q}}}\) of a payoff $c\in U_{\operatorname{lin}}(\R_+^{nd})$ with the infimal price over all super-replicating strategies which involve trading in digital options with payoffs $g_x$, $-f_x$, \(x\in \mathds{Q}_+^{nd}\,.\)  As a consequence of Definition~\ref{def_sm_qc}, if \(c\in \cF_\Delta^-\) is a \(\Delta\)-antitone, left-continuous, and measure-inducing payoff function, an upper bound for \(\overline{P}_{\cM^{\operatorname{lo}}_{\underline{Q},\overline{Q}}}\) is given by a quasi-expectation w.r.t.\,\(\overline{Q}\,.\)

Similar to the model in \eqref{eqdefclloor}, by means of Lemma~\ref{lem_indicator_constraints}~(b) we also derive improved price bounds for the class
$
\cM^{\operatorname{uo}}_{\underline{Q},\overline{Q}} := \{\Q\in \cM(\mu)\,|\, \underline{Q} \leq_{\operatorname{uo}} C_{\Q}\leq_{\operatorname{uo}} \overline{Q}\}
$
of risk-neutral distributions with dependence restrictions w.r.t. the upper orthant order. In this case, we obtain for every $\Delta$-monotone payoff function \(c\in \cF_\Delta\) which is left-continuous and measure-inducing, an upper bound for \(\overline{P}_{\cM^{\operatorname{uo}}_{\underline{Q},\overline{Q}}}\) which is given by a quasi-expectation w.r.t.\,\(\overline{Q}\,.\) Examples of typical \(\Delta\)-antitone and \(\Delta\)-montone payoff functions are provided in  \cite[Table 1]{Lux-2017,Roncalli-2001,Tankov-2011}.


\begin{thm}[Upper price bounds with \(\leq_{\operatorname{lo}}\)- and \(\leq_{\operatorname{uo}}\)-contraints]\label{theloob}~\\
\begin{itemize}
\item[(a)] \label{theloob1} Let \(c\in U_{\operatorname{lin}}(\R_+^{nd})\,.\) If \({\cM}^{\operatorname{lo}}_{\underline{Q},\overline{Q}}\ne \emptyset\,,\) then 
\begin{align}\label{theloob1a}
\overline{P}_{{\cM}^{\operatorname{lo}}_{\underline{Q},\overline{Q}}} &= \max_{\Q\in {\cM}^{\operatorname{lo}}_{\underline{Q},\overline{Q}}} \E_\Q \left[c(S)\right] = \underline{\mathcal{D}}_{\mathcal{S}}\,.
\end{align}
In particular, if \(c\in \cF_\Delta^-\cap U_{\operatorname{lin}}(\R_+^{nd})\) is left-continuous, then
\begin{align}\label{eqlosb}
\overline{P}_{{\cM}^{\operatorname{lo}}_{\underline{Q},\overline{Q}}} \leq \pi_c^{\mu} (\widehat{\overline{Q}})\,.
\end{align}
\item[(b)] \label{theloob2} Let \(c\in U_{\operatorname{lin}}(\R_+^{nd})\,.\) If \({\cM}^{\operatorname{uo}}_{\underline{Q},\overline{Q}}\ne \emptyset\,,\) then
\begin{align*}
\overline{P}_{{\cM}^{\operatorname{uo}}_{\underline{Q},\overline{Q}}} &= \max_{\Q\in {\cM}^{\operatorname{uo}}_{\underline{Q},\overline{Q}}} \E_\Q \left[c(S)\right] = \underline{\mathcal{D}}_{\mathcal{S}}\,.
\end{align*}
In particular, if \(c\in \cF_\Delta\cap U_{\operatorname{lin}}(\R_+^{nd})\) is left-continuous, then
\begin{align}\label{equosb}
\overline{P}_{{\cM}^{\operatorname{uo}}_{\underline{Q},\overline{Q}}} \leq \pi_c^{\mu} (\widehat{\overline{Q}})\,.
\end{align}
\end{itemize}
\end{thm}

%


\begin{rem}\label{remgenan}
\begin{enumerate}
\item[(a)] Similar to Theorem \ref{theloob}, we also obtain dual lower bounds under consideration of the martingale property as well as lower bounds \(\pi_c^{\mu} (\widehat{\underline{Q}})\) and \(\pi_c^{\mu} (\widehat{\underline{Q}})\) for \(\underline{P}_{{\cM}^{\operatorname{lo}}_{\underline{Q},\overline{Q}}}\) and \(\underline{P}_{{\cM}^{\operatorname{uo}}_{\underline{Q},\overline{Q}}}\) in the case of a left-continuous, measure-inducing, and \(\Delta\)-antitone / \(\Delta\)-monotone payoff function.
\item[(b)]  By incorporating the martingale condition as a linear constraint, Theorem \ref{theloob} improves the dual risk bounds considered in \cite[Theorem 3.2]{Rueschendorf-2019} as well as the quasi-copula bounds obtained in \cite{Lux-2017}.
\item[(c)]  Since the martingale property also restricts the dependence structure of \(\cM(\mu)\,,\) the class \({\cM}^{\operatorname{lo}}_{\underline{Q},\overline{Q}}\) might be empty for too restrictive choices of the quasi-copulas \(\underline{Q}\) and \(\overline{Q}\,.\) However, if the market is free of \emph{model-independent arbitrage}, compare \cite[Definition 1.2.]{acciaio2016model}, and if the dependence restrictions are inferred from option prices with continuous payoff functions, then there exists a martingale measure \(\Q\in \mathcal{M}(\mu)\) for which the bounds \(\underline{Q}\) and \(\overline{Q}\) for the copula of \(\Q\) are consistent. This follows by \cite[Theorem 1.3.]{acciaio2016model}, for which we additionally need to assume the existence of a convex superlinear payoff that can be bought.
\item[(d)] \label{remgenan1} 

%

Note that by Definition~\ref{def_orthant_orders}~(b) we have 
$$
 {\cM}^{\operatorname{uo}}_{\underline{Q},\overline{Q}} =\left\{\Q\in \mathcal{M}(\mu)~\middle|~\widehat{\underline{Q}}(u)\leq \widehat{C_\Q}(u)\leq \widehat{\overline{Q}}(u) \text{ for all } u\in [0,1]^{nd} \right\}.
$$
 This allows to extend  \eqref{equosb} to bounded measurable functions $\widehat{\underline{Q}}, \widehat{\overline{Q}}$, defined on $[0,1]^{nd}$ which are not survival functions of quasi-copulas. Indeed, the identity in \eqref{equosb} is still valid due to the positivity of the measures \(\eta_{c_I}\) induced by the \(\Delta\)-monotone function \(c\) applied in \eqref{defquexpop} for the definition of $\pi_c^\mu$ and due to the pointwise upper bound \(\widehat{Q^I}\) for \(\widehat{C_\Q^I}\) for all \(I\subseteq \{1,\ldots,nd\}\,\).
\end{enumerate}
\end{rem}


\item \textbf{Upper bounds in common component dependence models}\\
Liquidly traded options written on each pair of the underlying assets $S^1,\ldots, S^d$ which allow to derive inter-asset dependence information are often not available.
However, prices of derivatives written on a main reference asset \(S^1_{t_i}\) and another asset \(S^{k}_{t_i}\) at the same time \(t_i\) may be available for \(k=2,\ldots,d\,.\) For example, basket or digital options on \(S^1\) and \(S^2\) as well as on \(S^1\) and \(S^3\) may be traded. Therefore, we consider the case where prices of specific derivatives as a function of an asset \(S^1_{t_i}\) and assets \(S^{k}_{t_i}\) at time \(t_i\) are known for \(k=2,\ldots,d\,.\) This price information then implies by \cite[Theorem 3.1]{Lux-2017} an upper quasi-copula bound \(Q^k\in \cQ_2\) (w.r.t.\,the lower orthant ordering) for the copula \(C_{\Q_i^{1,k}}\) of the risk-neutral distribution \(\Q_i^{1,k}\) of \((S_{t_i}^1,S_{t_i}^k)\,,\) \(k=2,\ldots,d\,,\) under some \(\Q \in \cM(\mu)\) because then the value of the copula \(C_{\Q_i^{1,k}}\) is known on a  corresponding compact set.
Here \(\Q_i^{1,k}\) denotes the bivariate \((1,k)\)-marginal distribution of \(\Q\) at time \(t_i\,.\)
Taking the pointwise maximum $Q_2:=\max_{k=2,\dots,d}Q^k$ over these quasi-copula bounds then yields a quasi-copula \(Q_2\in \cQ_2\) as a pointwise upper bound for the associated copulas \(\{C_{\Q_i^{1,k}}\,, k=2,\ldots,d\},\) see Example~\ref{exasmpf}. 

Given marginal distributions $\mu=(\mu_1^1,\dots,\mu_n^d)$ with $\mu_i^k\in \mathcal{P}(\R_+)$ for all $i=1,\dots,n$, $k=1,\dots,d$, a quasi-copula $Q_2\in \cQ_2$, and some time $t_i$, we consider the class
\begin{align}\label{eqdefclinfamo}
\cM_{Q_2}^{\operatorname{CCD}}&:= \left\{\Q\in \cM(\mu)\,\middle|\, C_{\Q_i^{1,k}}\leq_{\operatorname{lo}} Q_2 ~\text{for all } k=2,\ldots,d\right\}
\end{align}
of measures \(\Q\) in \(\cM(\mu)\) such that for all $k=2,\dots,d$ the copula \(C_{\Q_i^{1,k}}\) associated with the bivariate component \((S_{t_i}^1,S_{t_i}^k)\) under \(\Q\) is upper bounded by \(Q_2\) w.r.t.\,the lower orthant ordering. In particular, this class is another example for a specification of the set $\mathcal{M}^{\operatorname{lin}}$. Elements of this class correspond to common component dependence models (CCD model) with bivariate dependence specification sets as described in Section \ref{sec IFM}. Here the market-implied dependence structure of \((S_{t_i}^1,\ldots,S_{t_i}^d)\) is restricted to  the subclass of copulas having the property that the bivariate marginal copula \(C_{\Q_i^{1,k}}\)  associated with \((S_{t_i}^1,S_{t_i}^k)\) belongs to the constrained specification set \(\{C\in \cC_2\,, C\leq_{\operatorname{lo}} Q_2\}\)  for all \(k=2,\ldots,d\,,\) see also \cite{Ansari-Rueschendorf-2020}. 
If \(Q_2\) is the upper Fr\'{e}chet copula \(M^2\), no dependence restrictions are imposed, hence \(\cM_{Q_2}^{\operatorname{CCD}} = \cM(\mu)\,.\) 

As a consequence of Theorem \ref{theqcub}, the dependence structure of \(S_{t_i}=(S_{t_i}^1,\ldots,S_{t_i}^d)\) in \(\cM_{Q_2}^{\operatorname{CCD}}\) has an upper bound w.r.t. \(\leq_{\operatorname{sm}}\) given by the quasi-copula in \eqref{deqcqs}.
This yields even for a supermodular payoff function \(c\in \cF_{\operatorname{sm}}\) a representation of an upper bound for \(\overline{P}_{\cM_{Q_2}^{\operatorname{\operatorname{CCD}}}}\) in form of an analytic expression depending on the quasi-copula \(Q_2\,.\)
To apply the duality result from Proposition~\ref{thm_duality_linear_constraints}, we denote by $\R_+^{nd}\ni x=(x_1^1,\dots,x_n^d)\mapsto \operatorname{proj}_i^k(x)=x_i^k \in \R_+$ the projection of $x$ onto its $(i,k)$-th component.
In this setting, we take inequality constraints of the form
\begin{align*}
\E_\Q [g_{x,y}(S_{t_i}^1,S_{t_i}^k)]\leq Q_2(F_{i}^1(x),F_{i}^k(y))\,,~~~g_{x,y}(\cdot):=\one_{\{\cdot < (x,y)\}}\,, (x,y)\in \mathds{Q}_+^2\,, 2\leq k\leq d\,,
\end{align*}
into account.
Then, we derive the following duality result.
\begin{thm}[Upper price bounds with contraints related to CCD models]\label{corsmifm}~\\
Assume that \(\cM_{Q_2}^{\operatorname{\operatorname{CCD}}}\ne \emptyset\,.\) Then, the following holds.
\begin{itemize}
\item[(a)]
Let \({c}\in U_{\operatorname{lin}}(\R_+^{nd})\), then 
\begin{align}\label{eqcorsmifm1}
\overline{P}_{\cM_{Q_2}^{\operatorname{\operatorname{CCD}}}} &= \max_{\Q\in \cM_{Q_2}^{\operatorname{\operatorname{CCD}}}} \E_\Q [c(S)] = \underline{\mathcal{D}}_{\mathcal{S}}\,.
\end{align}
\item[(b)]
Let \(\widetilde{c}\in \cF_{\operatorname{sm}}\cap C_{\operatorname{lin}}(\R_+^{nd})\) be componentwise increasing/componentwise decreasing. Then, we have that the function $c:=\left(\widetilde{c} \circ \operatorname{proj}_i^1,\cdots, \widetilde{c} \circ \operatorname{proj}_i^d\right) \in \cF_{\operatorname{sm}}\cap C_{\operatorname{lin}}(\R_+^{d})$ is also componentwise increasing/componentwise decreasing and
\begin{align}\label{eqcorsmifm2}
\overline{P}_{\cM_{Q_2}^{\operatorname{\operatorname{CCD}}}} \leq \pi_{\phi_{c\,\circ\left((F_i^1)^{-1},\ldots,(F_i^d)^{-1}\right)}} (\widehat{Q_2})\,\text{ for all } i=1,\dots,n,
\end{align}
with
\(\phi_{c\,\circ\left((F_i^1)^{-1},\ldots,(F_i^d)^{-1}\right)}\) defined by \eqref{defphif}.
\end{itemize}

\end{thm} 

In particular, Theorem~\ref{corsmifm}~(b) can be applied to payoff functions that depend on multiple assets but only on one specific maturity.

\begin{rem}
\begin{enumerate}[(a)]
\item Theorem \ref{corsmifm} allows a closed-form representation of an upper price bound for a componentwise increasing/componentwise decreasing supermodular payoff function by an integral of the bivariate dependence constraint \(Q_2\,.\) Examples for common supermodular payoff functions which are componentwise increasing/componentwise decreasing are listed in \cite[Table 1]{Lux-2017,Roncalli-2001,Tankov-2011}. In particular, if \(\varphi\colon \R\to \R\) is (increasing and) convex, then \(\varphi\left(\sum_{i=1}^d x_i\right)\) is (increasing and) supermodular.
\item A similar result as the equalitity in \eqref{eqcorsmifm1} holds true for lower bounds. However, an analogue of \eqref{eqcorsmifm2} for lower bounds cannot be obtained for general dimension since upper product ordering results are different from ordering results for lower products which correspond to lower bounds in partially specified factor models, see \cite{Ansari-Rueschendorf-2021}.
\end{enumerate}
\end{rem}

\end{enumerate}


\subsection{Dependence Information Related to Correlations}\label{secadddepinreltocor}
In addition to martingale and marginal constraints, we take into account additional market-implied dependence information related to the inter-asset correlation.

\begin{enumerate}[1.)]\setcounter{enumi}{2}
\item \textbf{Knowledge of risk-neutral correlations} \\
We incorporate additional information on the covariance of the assets.
This approach is motivated by the observation of basket options with payoff structure
\begin{equation}
(a_1S_{t_i}^k +a_2S_{t_j}^l-K)_+, \label{eq_basket_def}
\end{equation}
for some fixed weights $a_1,a_2 \in \R$ with $a_1  a_2 \neq 0$ and $k,l \in \{1,\dots,d\},~i,j\in \{1,\dots,n\}$.
If prices of such options are observable for all strikes $K\in \R$ and $\Q \in \mathcal{M}(\mu)$ is consistent\footnote{We call a measure $\Q$ consistent with a price $p$ for a derivative $c$ if it holds $\E_\Q[c]=p$.} with these prices, then we deduce that\footnote{Here we implicitly assume that the prices are differentiable as a function of the strike. If this is not the case, then we consider instead the right derivative and still obtain a one-to-one relation between the risk-neutral distribution of the sum and the basket option prices. Compare for the case of call options e.g. \cite[Lemma 2.2]{hobson2011skorokhod} and the discussion thereafter.}
$
\frac{\partial}{\partial K}\E_\Q\left[(a_1S_{t_i}^k +a_2S_{t_j}^l-K)_+\right]=\Q\left(a_1S_{t_i}^k +a_2S_{t_j}^l \leq K\right)-1.
$
We hence obtain the distribution\footnote{We remark that in practice one may only observe the prices of basket options for a (small) finite number of strikes. Depending on the number of observed strikes this might impose restrictions on the applicability of the approach. However, if the number of observed strikes is sufficiently high, then it is possible to consider finite differences to compute the inferred distribution and its second moments. Compare also, e.g., \cite{bauerle2021consistent} and \cite{cohen} where this method is applied to obtain the marginal distribution of stocks implied from call option prices.} of $a_1S_{t_i}^k +a_2S_{t_j}^l$. This allows in particular to compute its second moment. In addition, observe that
\begin{align} \label{eq_covariance_info_basket}
\E_\Q\left[S_{t_i}^kS_{t_j}^l\right]=\frac{\E_\Q\left[\left(a_1S_{t_i}^k +a_2S_{t_j}^l\right)^2\right]-a_1^2\E_\Q\left[\left(S_{t_i}^k\right)^2\right]-a_2^2\E_\Q\left[\left(S_{t_j}^l\right)^2\right]}{2a_1a_2},
\end{align}
assuming that the second moments of the marginal distributions exist.
Since $\Q \in \mathcal{M}(\mu)$, all values of the right-hand side of \eqref{eq_covariance_info_basket}
are known and hence so is the left-hand side. Moreover, by the martingale property, the correlation is given by
\begin{align}
\operatorname{Corr}_\Q\left(S_{t_i}^k,S_{t_j}^l\right)=\frac{\E_\Q\left[S_{t_i}^k S_{t_j}^l\right]-S_0^kS_0^l}{\sqrt{\E_{\mu_i^k}\Big[(S_{t_i}^k)^2\Big]-(S_0^k)^2}\sqrt{\E_{\mu_j^l}\Big[(S_{t_j}^l)^2\Big]-(S_0^l)^2}},
\end{align}
which by \eqref{eq_covariance_info_basket} is known, too, since $S_0^kS_0^l$ is some constant value. Therefore, price information on options of \eqref{eq_basket_def}-type for all strikes $K$ is sufficient to obtain information on the correlation between $S^k_{t_i}$ and $S^l_{t_j}$. To model the risk-neutral correlation $\rho_{ij}^{kl}:=\operatorname{Corr}_\Q\left(S_{t_i}^k,S_{t_j}^l\right) \in [-1,1]$ between $S_{t_i}^k$ and $S_{t_j}^l$ with respect to a measures $\Q \in \mathcal{M}(\mu)$, we specify the equality constraints in equation \eqref{eq_restriction} by $f_{(i,j,k,l)}^{\operatorname{eq}}\in C_{\operatorname{lin}}(\R_+^{nd})$ with
 \begin{equation}\label{eq_correlation_constraints_1}
f_{(i,j,k,l)}^{\operatorname{eq}}(x_1^1,\ldots,x_n^d)=\frac{x_i^k x_j^l-S_0^kS_0^l }{\sqrt{\E_{\mu_i^k}\Big[(S_{t_i}^k)^2\Big]-(S_0^k)^2}\sqrt{\E_{\mu_j^l}\Big[(S_{t_j}^l)^2\Big]-(S_0^l)^2}},~~(x_1^1,\ldots,x_n^d) \in \R_+^{nd}.
 \end{equation}
such that $\E_\Q\left[f_{(i,j,k,l)}^{\operatorname{eq}}\right]=K_{(i,j,k,l)}^{\operatorname{eq}}$,
where $K_{(i,j,k,l)}^{\operatorname{eq}}:=\rho_{ij}^{kl}$ for all measures $\Q$ consistent with the correlation structure.
In Example~\ref{exa_corr1}, we investigate two examples in the case $n=2,~d=2$ and include additional information on the correlation between the assets. This information leads to several constraints which restrict the set of possible pricing measures in different degrees and therefore effectively influence robust price bounds.

\item \textbf{Knowledge of the risk-neutral distribution of sums}\\
Next, we consider not just correlation information, but the entire information on the sum of the underlying assets. This corresponds to considering prices of basket options directly. 
Then we specify $ f_{(i,j,k,l,m)}\in C_{\operatorname{lin}}(\R_+^{nd})$ in \eqref{eq_eq_constraints}  by 
$
f_{(i,j,k,l,m)}^{\operatorname{eq}}(x_1^1,\ldots,x_n^d)=(a_1^{(i,j,k,l)}x_i^k +a_2^{(i,j,k,l)} x_j^l-K_m)_+$, $(x_1^1,\ldots,x_n^d) \in \R_+^{nd},
$
for all $i,j \in \{1,\dots,n\},\,\) \(k,l\in \{1,\ldots,d\},\,\) \(m\in \{1,\ldots,N_{(i,j,k,l)}\}\), where $a_1^{(i,j,k,l)},~a_2^{(i,j,k,l)} \in \R$ denote the corresponding weights of the basket options under consideration, $K_m \in \R$ the strike of the option and $N_{(i,j,k,l)}\in \N$ corresponds to the amount of observable options for this asset-maturity combination.
Moreover, we denote by \(K_{(i,j,k,l,m)}^{\operatorname{eq}}\) 
the price of the basket option with payoff function \(f_{(i,j,k,l,m)}^{\operatorname{eq}} \).
If the price information implied by basket options is consistent with risk-neutral correlations as considered in 3.), then respecting the prices instead of the correlations may lead to a further improvement of the price bounds, as not only the second moment of the underlying distribution is taken into account.
\item \textbf{Risk-neutral correlation is constant over time}\\
In Section~\ref{sec_correlation_constant}, we discuss situations for $d=2$  in which it is reasonable to assume for any $\Q \in \mathcal{M}(\mu)$ that
$
\operatorname{Corr}_\Q(S_{t_i}^1,S_{t_i}^2)=\operatorname{Corr}_\Q(S_{t_j}^1,S_{t_j}^2)$ for all $i,j \in \{1,\dots,n\}.$
This leads to equality constraints of the form
\begin{equation}\label{eq_correlation_constant_1}
\begin{aligned}
f^{\operatorname{eq}}_{(i,j)}(x_1^1,\dots,x_n^d)&=\Bigg(\tfrac{x_{i}^1x_i^2-S_{t_0}^1S_{t_0}^2}{\sqrt{\E_{\mu_i^1}\left[\left(S_{t_i}^1\right)^2\right]-\left(S_{t_0}^1\right)^2}\sqrt{\E_{\mu_i^2}\left[\left(S_{t_i}^2\right)^2\right]-\left(S_{t_0}^2\right)^2}}\\
&\hspace{1cm}-\tfrac{x_{j}^1x_{j}^2-S_{t_0}^1S_{t_0}^2}{\sqrt{\E_{\mu_j^1}\left[\left(S_{t_j}^1\right)^2\right]-\left(S_{t_0}^1\right)^2}\sqrt{\E_{\mu_j^2}\left[\left(S_{t_j}^2\right)^2\right]-\left(S_{t_0}^2\right)^2}}\Bigg),~~(x_1^1,\ldots,x_n^d) \in \R_+^{nd},
\end{aligned}
\end{equation}
and $K^{\operatorname{eq}}_{(i,j)}=0$ for all $i,j =1,\dots,n,~ i \leq j$.
\item \textbf{Risk-neutral correlation is bounded from below by the real world correlation}\\
In Section~\ref{sec_correlation_bounded}, we discuss situations in $d =2$ in which it makes sense to assume for every $\Q \in \mathcal{M}(\mu)$ that $\operatorname{Corr}_\Q(S_{t_i}^1,S_{t_i}^2)\geq \operatorname{Corr}_\PP(S_{t_i}^1,S_{t_i}^2)$, where $\PP\in \mathcal{P}(\R_+^{nd})$ denotes some underlying real-world measure.
Thus, we model the inequality constraints in \eqref{eq_ineq_constraints} by setting 
\begin{equation}\label{eq_corr_constraint_below_1}
f^{\operatorname{ineq}}_{i}(x_1^1,\dots,x_n^d)= -\frac{x_i^1x_i^2-S_{t_0}^1S_{t_0}^2}{\sqrt{\E_{\mu_i^1}\left[\left(S_{t_i}^1\right)^2\right]-\left(S_{t_0}^1\right)^2}\sqrt{\E_{\mu_i^2}\left[\left(S_{t_i}^2\right)^2\right]-\left(S_{t_0}^2\right)^2}},~~(x_1^1,\ldots,x_n^d) \in \R_+^{nd},
\end{equation}
and $K^{\operatorname{ineq}}_{i}=-\operatorname{Corr}_\PP(S_{t_i}^1,S_{t_i}^2)$ for $i=1,\dots,n$, where $\operatorname{Corr}_\PP(S_{t_i}^1,S_{t_i}^2)$ can often be estimated empirically with statistical methods.
\end{enumerate}

\section{Examples and Numerics}\label{section_numerics}
In accordance with the scenarios described in Sections~\ref{secadddep} and \ref{secadddepinreltocor}, we provide several examples for the improvement of the upper multi-asset price bound \(\overline{P}_{\cM(\mu)}\). This improvement is due to the consideration of appropriate market-implied dependence information. The following examples cover the case where a restriction on the dependence structure is imposed through (quasi-)~copulas as well as the case where additional information on the correlation is taken into account. The \emph{Python} codes for the numerical examples from this section are provided under \url{https://github.com/juliansester/improved-dependence-pricing}.

\subsection{Improved price bounds through copula bounds}\label{sec41a}
We consider for $K \in \R$ the payoff functions  
\begin{align}
\label{eqdefpof2} c_{1,K}(S_{t_1}^1,S_{t_1}^2,S_{t_1}^3,S_{t_2}^1,S_{t_2}^2,S_{t_2}^3) &:= \left( \min_{\substack{i=1,2 \\ k=1,2,3}}\{S_{t_i}^k\}- K \right)_+\,,\\
\label{eqdefpof3} c_{2,K}(S_{t_1}^1,S_{t_1}^2,S_{t_1}^3)&:=\left(\frac {S_{t_1}^1+S_{t_1}^2+S_{t_1}^3} 3-K\right)_+\,.
\end{align}
For every $K \in \R$, the payoff function \(c_{1,K}\) is \(\Delta\)-monotone and \(c_{2,K}\) is increasing and supermodular, but neither \(\Delta\)-antitone nor \(\Delta\)-monotone. For the sake of readability, we sometimes abbreviate \(c_i:=c_{i,K}\,,\) \(i=1,2\,.\)  In the following, we apply Theorem \ref{theloob} and Theorem \ref{corsmifm} to determine price bounds for these options under consideration of the martingale property and of copula bounds for the risk-neutral distributions inferred from dependence information based on prices of some options.
More specifically, we determine price bounds by considering minimal and maximal expectations w.r.t.\,measures from $\cM^{\operatorname{uo}}_{\underline{Q},\overline{Q}}$, and $\cM_{Q_2}^{\operatorname{CCD}}$, respectively, where the quasi-copulas $\underline{Q}$, $\overline{Q}$, and \(Q_2\) are inferred from option prices \(p_i^{k\ell}(K')\) of digital options $d_i^{k\ell}(K')$ with a payoff function defined by
\begin{align}\label{obsop}
d_i^{k\ell}(K')\big(S_{t_i}^k,S_{t_i}^\ell\big):= \one_{\{\max\{S_{t_i}^k,S_{t_i}^\ell\}\leq K'\}} ,~~k,\ell \in \{1,2,3\}, i \in \{1,2\}.
\end{align}
We assume that prices $p_i^{k\ell}(K')$ are observed in the market for strikes $ K'\in \cK:=\{K_1,\ldots,K_m\}$, where $m \in \N$ describes the number of observed digital options. Knowledge of such option prices restricts the set of consistent pricing measures and therefore, via Sklar's theorem, prescribes the values of the associated (survival) copula on a finite set. This set is implied by the choice of $\mathcal{K}$ and the marginal distributions \(\mu\,.\) In such a case, lower orthant and upper orthant copula bounds for the dependence structure of the underlying assets are given by the quasi-copulas obtained from the following two results, see \cite[Theorem 3.1 and Proposition A.1]{Lux-2017}. Several examples are provided in this section. More generally, prices of options with payoff functions that are increasing w.r.t.\,the lower or upper orthant ordering allow to infer bounds for the partially known copula of the underlying asset, see \cite[Theorem 3.3 and Proposition A.1]{Lux-2017} and  \cite[Table 1]{Lux-2017,Roncalli-2001,Tankov-2011}.

In the following Examples~\ref{exadelmo} and~\ref{exasmpf}, we determine upper price bounds for the options \(c_{1,K}\) and \(c_{2,K}\) for different strikes \(K\) under the assumption that prices of some digital options as specified in \eqref{obsop} are given.  We generate prices $p_i^{k\ell}$, similar to \cite[Example 6.8]{Lux-2017}, by assuming an underlying multivariate Black-Scholes model \(S=(S_t^1,S^2_t,S^3_t)_{t\geq 0}\) with
\begin{align}\label{eqsimpro}
S^k_{t}=S_0^k\exp\left(-\tfrac {(\sigma^k)^2} 2 {t}+\sigma^k X_{t}^k\right),~k=1,2,3,~t\geq 0,
\end{align}
where \(X=(X_t^1,X_t^2,X_t^3)_{t\geq 0}\) is a Brownian motion with dependent components that are distributed \mbox{\((X_1^1,X_1^2,X_1^3)\sim  N(0,\Sigma)\)} with covariance matrix
$
\Sigma=\begin{pmatrix} 1 & \rho_{12} & \rho_{13} \\
\rho_{12} & 1 & \rho_{23} \\
\rho_{13} & \rho_{23} & 1 \end{pmatrix}\,.
$
We specify the parameters \(S_0^k\,,\) \(\sigma^k\,,\) and \(\rho_{k\ell}\,,\) as well as the set of strikes \(\cK\) of the observed digital options in Example~\ref{exasmpf}.
%
\begin{exa}[\(\Delta\)-monotone payoff function]\label{exadelmo}~
We specify \(t_1=1\,,\) \(t_2=2\,,\) \(\sigma^k=0.5\) for all \(k=1,2,3\,,\) \(S_0^1=9\,,\) \(S_0^2=10\,,\) and \(S_0^3=11\) as well as the risk-neutral correlations \(\rho_{12}=\rho_{13}=\rho_{23}=0.8\,.\) Further, we assume that the prices \(p_2^{k\ell}(K')\) of the digital options \(d_2^{k\ell}(K')\,,\) \(1\leq k < \ell\leq d =3\,,\) can be observed for strikes \(K'\in \cK:=\{8,9,10,11,12\}\,.\)
Knowledge of such option prices \(p_2^{k\ell}(K')\) means knowledge of the value of the survival functions \(\widehat{C}^{k\ell}\) associated with the copula \(C^{k\ell}\) of \((S_{t_2}^k,S_{t_2}^\ell)\) given by
\begin{align}
\nonumber
\widehat{C}^{k\ell}(F_2^k(K'),F_2^\ell(K')) &=C^{k\ell}(F_2^k(K'),F_2^\ell(K'))+1-F_2^k(K')-F_2^\ell(K')\\
\label{eqdigoppri2}&= \E_\Q [d_2^{k\ell}(K')] +1-F_2^k(K')-F_2^\ell(K') =  p_2^{k\ell}+1-F_2^k(K')-F_2^\ell(K')
\end{align}
for all \(\Q\in \cM^{\operatorname{lin}}\)
fulfilling the equality constraint \(\E_\Q [d_2^{k\ell}(K')]=p_2^{k\ell}(K')\) for \(K'\in \cK\,.\)
Hence, we obtain from \cite[Proposition A.1]{Lux-2017} pointwise a lower bound \(\underline{\widehat{Q}}\) and an upper bound \(\widehat{\overline{Q}}\) for \(\widehat{C}\).

Now, for $K \in \R$, we compute the upper bound \(\pi_{c_{1}}^\mu(\widehat{\overline{Q}})\) in \eqref{equosb} for the price of the option \(c_1=c_{1,K}\) as specified in \eqref{eqdefpof2} under knowledge of the digital option prices \(p_2^{k\ell}(K')\,,\) \(1\leq k < \ell \leq 3\,,\) \(K'\in \cK\,.\)
To compute \(\pi_{c_{1}}^\mu(\widehat{\overline{Q}})\) we apply Theorem~\ref{theloob}~(b)
and use that \(\widehat{\overline{Q}}\) is a pointwise upper bound for \(\widehat{C}\,\), see also Remark~\ref{remgenan}~(d).

Figure~\ref{exa_42_martingale_improvement} illustrates the price bounds \(\pi_{c_{1}}^{\mu}(\widehat{\overline{Q}})\) and \(\overline{P}_{{\cM}^{\operatorname{uo}}_{\underline{Q},\overline{Q}}}\) for the option \(c_1=c_{1,K}\) obtained from Theorem \ref{theloob}~(b) in dependence of the strike $K$. Moreover, we illustrate the corresponding lower bounds. We observe that the bound which incorporates the martingale property improves the bound $\pi_{c_{2}}^\mu(\widehat{\overline{Q}})\,$ significantly. The price bound \(\overline{P}_{{\cM}^{\operatorname{uo}}_{\underline{Q},\overline{Q}}}\) is computed through an adaption of the algorithm provided in \cite{aquino2019bounds}, which relies on a neural network approximation of the optimal dual hedging strategy. We observe that in this setting including information on prices of digital options improves the price bounds only slightly, whereas in combination with the martingale property the price bounds can be improved significantly.
\begin{figure}[h!]
\begin{center}
    \includegraphics[width=0.7\textwidth]{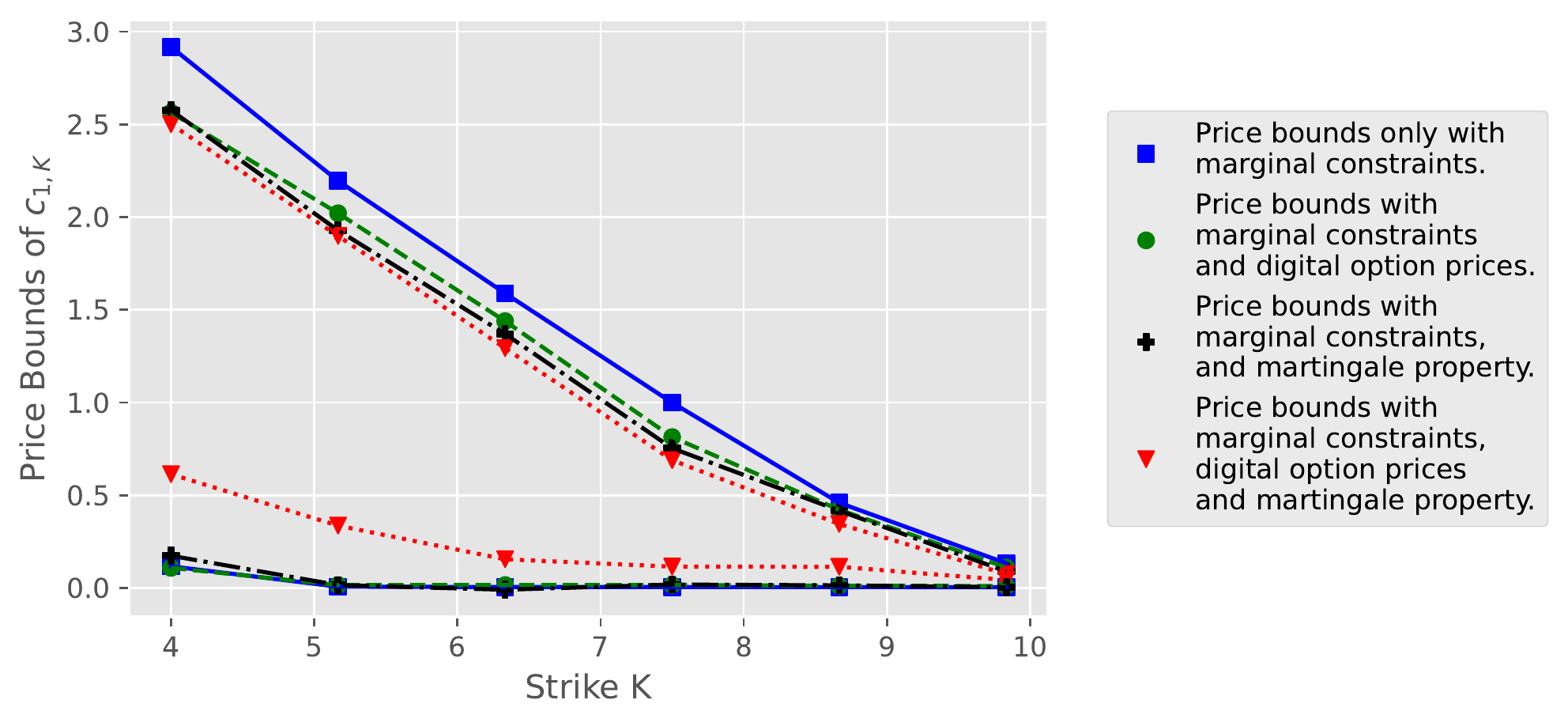}
                 \caption{In the setting of Example~\ref{exadelmo}, the figure depicts different lower and upper price bounds of $c_{1,K}$ in dependence on the strike $K$. We show price bounds without knowledge of prices of digital options, price bounds which additionally respect the martingale property, price bounds which one obtains after the inclusion of price information of digital option prices as well as \(\underline{P}_{{\cM}^{\operatorname{uo}}_{\underline{Q},\overline{Q}}}\), \(\overline{P}_{{\cM}^{\operatorname{uo}}_{\underline{Q},\overline{Q}}}\) which take into account the martingale property and the prices of digital options.}   \label{exa_42_martingale_improvement}
\end{center}
\end{figure}
\end{exa}

For the determination of an improved upper price bound for the basket call option \(c_{2,K}\) when dependence information is related to the setting of a common component dependence model, we make use of the following lemma, whose proof is provided at the end of Section~\ref{section_proofs}.

\begin{lem}[European basket options]\label{lembasopt}~\\ 
Let \(\mathfrak{C}(x_1,\ldots,x_d)=\left(\sum_{i=1}^d \alpha_i x_i-K\right)_+\) be the payoff function of the European basket call option with strike $K \in \R$ and let \(\mathfrak{P}(x_1,\ldots,x_d) = \left(K-\sum_{i=1}^d \alpha_i x_i\right)_+\) be the payoff function of the European basket put option with weights \(\alpha_i>0\,,\) \(1\leq i\leq d\,,\) and strike $K \in \R$. Then, the following statements hold true: 
\begin{itemize}
\item[(a)] \label{lembasopt1} \(\mathfrak{P}\) and \(\mathfrak{C}\) are measure-inducing if and only if \(d\leq 2\,.\)
\item[(b)]  \label{lembasopt2} Let \(F_1,\ldots,F_d\in \cF_+^1\) be continuous with finite first moments. If \(D^2,\ldots,D^d\in \cC_2\) and \(Q_2\in \cQ_2\) with \(D^i\leq_{\operatorname{lo}} Q_2\) for \(2\leq i\leq d\,,\) then
\begin{align}
\label{upppbcaa}\psi_{\mathfrak{C}}^{(F_1,\ldots,F_d)} (M^2\vee D^2\vee \cdots \vee D^d)& \leq \pi_{\phi_{\mathfrak{C}}}^{(G,F_1)}(\widehat{Q_2}),~~~ \text{and}\\
\nonumber \psi_{\mathfrak{P}}^{(F_1,\ldots,F_d)} (M^2\vee D^2\vee \cdots \vee D^d) &\leq \pi_{\phi_{\mathfrak{P}}}^{(G,F_1)}(\widehat{Q_2}),
\end{align}
where \(G\) is the distribution function defined by its generalized inverse 
$
G^{-1}(u):=\frac{\sum_{i=2}^d \alpha_i F^{-1}_i(u)}{\sum_{i=2}^d \alpha_i }\,,\) \(u\in [0,1]\,,
$
and where \(\phi_{\mathfrak{C}}\)
and \(\phi_{\mathfrak{P}}\) are
defined as in \eqref{defphif}.
\end{itemize}
\end{lem}

In the following example, we determine improved upper price bounds for the European basket call option \(c_{2,K}\) in the setting of a common component dependence model.

\begin{exa}[Supermodular payoff function]\label{exasmpf}
We determine the upper price bound \(\pi_{\phi_{c_2}}^{\mu_1}(\widehat{Q_2})\,,\) \(\mu_1=(\mu_1^k)^{k=1,2,3}\,,\) for the option \(c_2=c_{2,K}\) specified in \eqref{eqdefpof3} when the quasi-copula bound \(Q_2\) is inferred from prices \(p_1^{1\ell}(K')\) of the digital\footnote{The methodology can also be applied to any other option written on two assets with \(\Delta\)-monotone or \(\Delta\)-antitone payoff function like basked options. We mainly chose digital options for the sake of exposition since they result in simpler formulas.} options \(d_1^{1\ell}(K')\,,\) \(\ell=2,3\,,\) in \eqref{obsop} for  strikes \(K'\in \cK=\{8.5,9,9.5,10,10.5\}\,.\) Note that \(c_2\) is a continuous supermodular payoff function which is componentwise increasing but neither measure-inducing nor \(\Delta\)-antitone nor \(\Delta\)-monotone, see Lemma \ref{lembasopt} and compare \cite[Example 3.9.4]{Mueller-Stoyan-2002}. However, the transformed function \(\phi_{c_{2}}\) given by \eqref{defphif} is measure-inducing because it is \(\Delta\)-monotone, compare Lemma \ref{lembasopt}~(a).

To generate option prices $p_1^{1\ell}(K')$ according to the underlying model from \eqref{eqsimpro}, we specify \(t_1=1\,,\) the volatility \(\sigma=1\,,\) the initial time asset values \(S_0^1=10\,,\) \(S_0^2=9\,,\) and \(S_0^3=11\). For the correlation, we consider the four different cases 
$
(\rho_{12},\rho_{13})\in \{(-1,-1),(-0.5,-0.5),(0,0),(0.5,0.5)\}.
$

In Figure~\ref{fig_exa43_improvement}, we illustrate the standard upper price bound \(\overline{P}_{\mathcal{M}}\) based on knowledge of the marginals and the improved upper price bounds \(\pi_{\phi_{c_{2}}}^{(G,F_1^1)}(\widehat{Q_2})\)  for the payoff function \(c_2=c_{2,K}\) in dependence on the strike \(K\,.\) The improved bounds are inferred from prices of the digital option \(d_1^{1\ell}(K')\,,\) \(\ell=2,3\,,\) \(K'\in \cK\,,\) which are computed according to the multivariate Black-Scholes model with dependent components explained by \eqref{eqsimpro}. For an illustration, we choose different specifications to model the dependencies between the components of the underlying Brownian motion expressed by the correlations \(\rho_{1\ell}\,,\) \(\ell=2,3\,.\) We find that the more negatively correlated the components, the better the price bounds.

\begin{figure}[h!]
\begin{center}
    \includegraphics[width=0.75\textwidth]{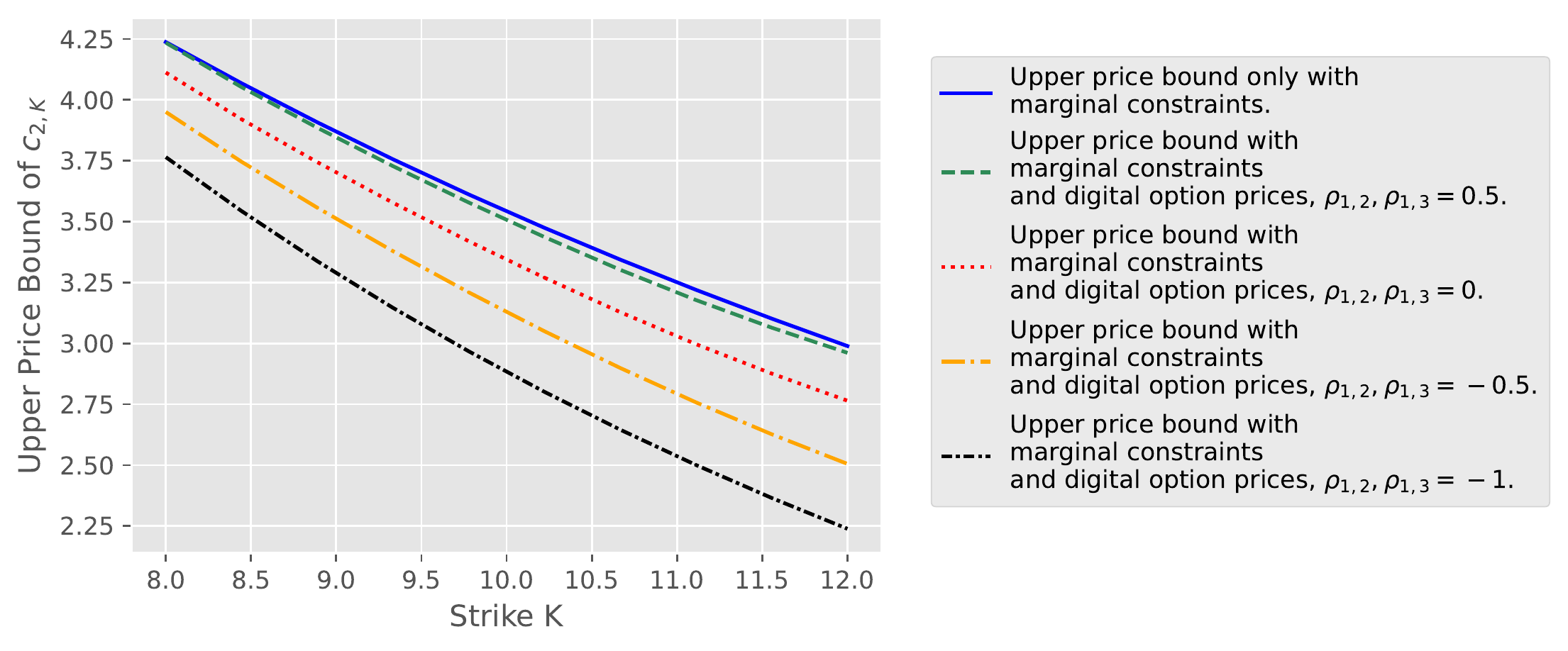}
                 \caption{Regarding Example~\ref{exasmpf}, we illustrate upper price bounds $\pi_{\phi_{c_{2}}}^{(G,F_1^1)}(\widehat{Q_2})$ for the basket put option $c_2=c_{2,K}$ in dependence on the strike $K$ for several correlations \(\rho_{1\ell}\) of the underlying Brownian motion in the Black-Scholes model from which the prices \(p_1^{1\ell}(K')\) of the digital options \(d_1^{1\ell}(K')\,,\) \(K'\in \cK\,,\) \(\ell=2,3\,,\) are calculated.
} 
\label{fig_exa43_improvement}
\end{center}
\end{figure}

%
\end{exa}

%

\subsection{Improved price bounds through correlations}~
In this section, we show within several examples how information on the risk-neutral correlation can improve model-independent price bounds of derivatives. Before discussing the improvement in an explicit setting in Example~\ref{exa_corr1}, we stress the influence of the chosen filtration for the martingale formulation on the set of admissible martingale measures and therefore on the resultant price bounds, as discussed in Remark~\ref{rem_differences_martingale_property}~(a).

Suppose for all examples in this section that $n=d=2$ and that $S$ has the following marginal distributions
\begin{equation}\label{eq_marginals_correlation}
\begin{aligned}
&S_{t_1}^1 \sim \mu_1^1=\mathcal{U}(\{8,10,12\}),&&S_{t_1}^2 \sim \mu_1^2=\mathcal{U}(\{8,10,12\}),\\
&S_{t_2}^1 \sim \mu_2^1=\mathcal{U}(\{7,9,11,13\}),&&S_{t_2}^2 \sim \mu_2^2=\mathcal{U}(\{4,7,10,13,16\}).
\end{aligned}
\end{equation}
We consider, similar to \cite[Example 5.12]{schmithals2018contributions} and \cite[Example 5.34]{schmithals2018contributions}, the following four payoff functions
\begin{align}\label{payofffc1c4}
\begin{split}
&c_3(S_{t_1}^1,S_{t_2}^1,S_{t_1}^2,S_{t_2}^2):=\left(1/4\cdot(S_{t_1}^1+S_{t_2}^1+S_{t_1}^2+S_{t_2}^2)-10\right)_+, \\
&c_4(S_{t_1}^1,S_{t_2}^1,S_{t_1}^2,S_{t_2}^2):=\left(10-\min\left\{S_{t_1}^1,S_{t_2}^1,S_{t_1}^2,S_{t_2}^2\right\}\right)_+, \\
&c_5(S_{t_1}^1,S_{t_2}^1,S_{t_1}^2,S_{t_2}^2):=\frac{1}{4}\left(S_{t_2}^2-S_{t_2}^1\right)_+\cdot\left(S_{t_1}^2-S_{t_1}^1\right)_+, \\
&c_6(S_{t_1}^1,S_{t_2}^1,S_{t_1}^2,S_{t_2}^2):=\left(\frac{S_{t_2}^1-S_{t_1}^1}{S_{t_1}^1}\right)^2\cdot \left(\frac{S_{t_2}^2-S_{t_1}^2}{S_{t_1}^2}\right)^2.
\end{split}
\end{align}
All numerical price bounds in this setting are computed using a linear programming approach, compare for further details e.g. \cite{guo2019computational} and \cite{henry2013automated}, which deliver a fast and accurate solution for a small amount of marginals. Note that our approach does not depend on the numerical method applied as long as it is computationally feasible. In fact, Figure \ref{exa_42_martingale_improvement} was derived based on the neural networks approach of \cite{eckstein2021computation} which avoids a discretization of the continuous marginals and is also applicable in higher dimensions.

\begin{exa}\label{exa_corr1}
In Figure~\ref{fig_3d_correlation}, we combine correlation information at times $t_1$ and $t_2$ and study the impact on the lower and upper price bound of $c_3, c_4,c_5$, and \(c_6\). 
As a result, we obtain a significant improvement of the price bounds for each of the payoff functions.

\begin{figure}[h!]
\begin{center}
	\subfigure[Price Bounds of $c_3$]{
    \includegraphics[width=0.45\textwidth]{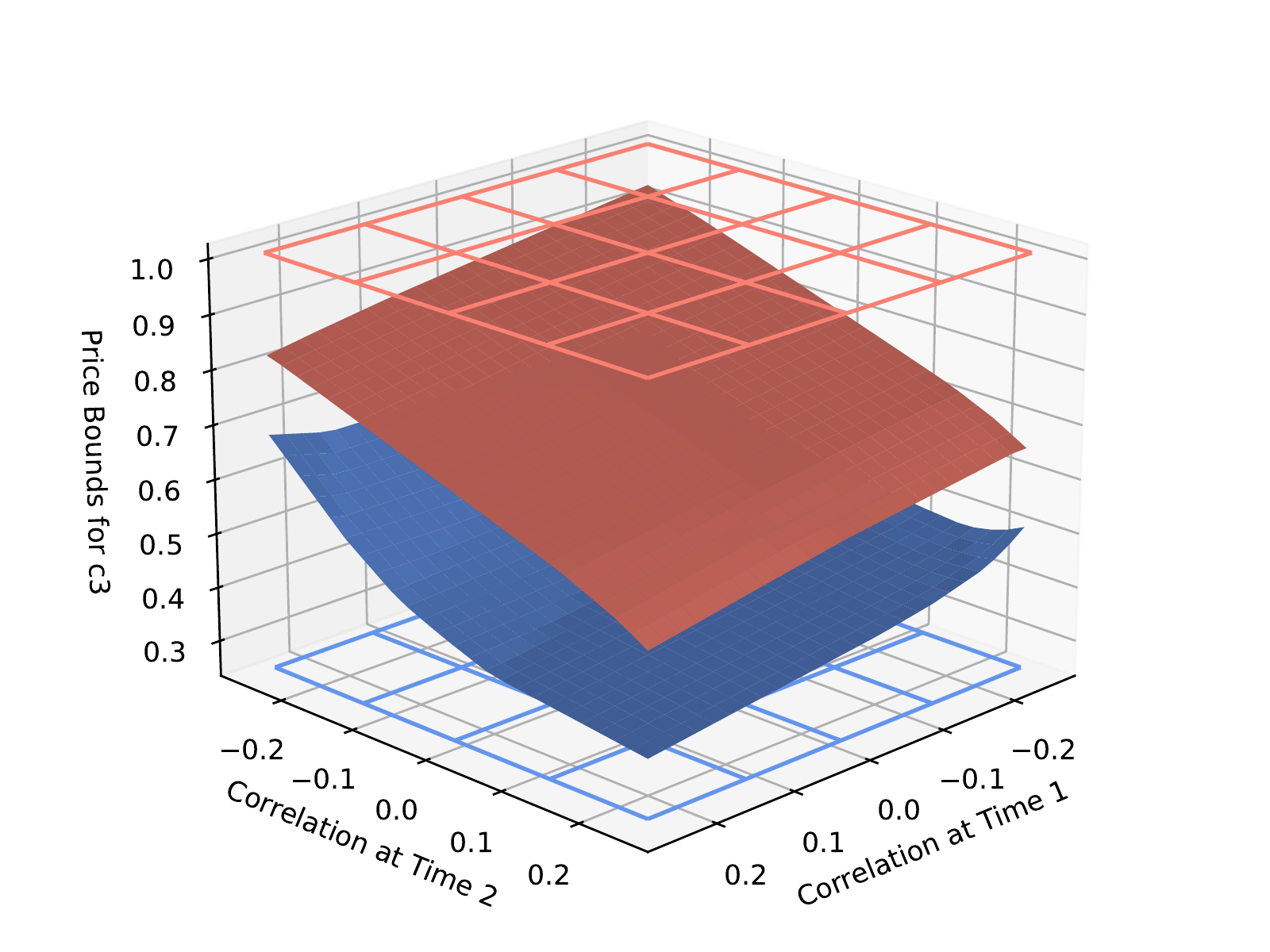}}
    \subfigure[Price Bounds of $c_4$]{
    \includegraphics[width=0.45\textwidth]{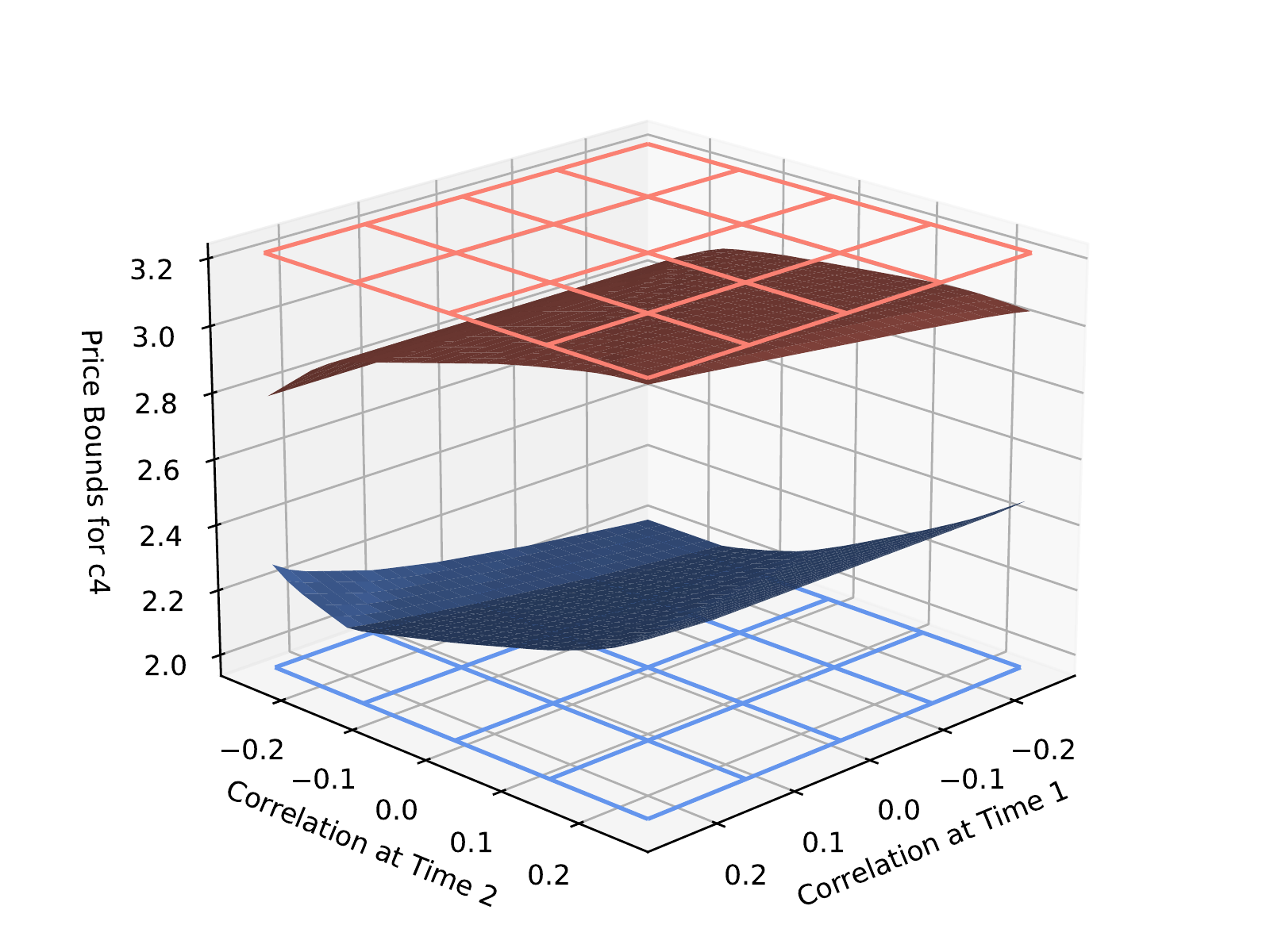}}
    \subfigure[Price Bounds of $c_5$]{
    \includegraphics[width=0.45\textwidth]{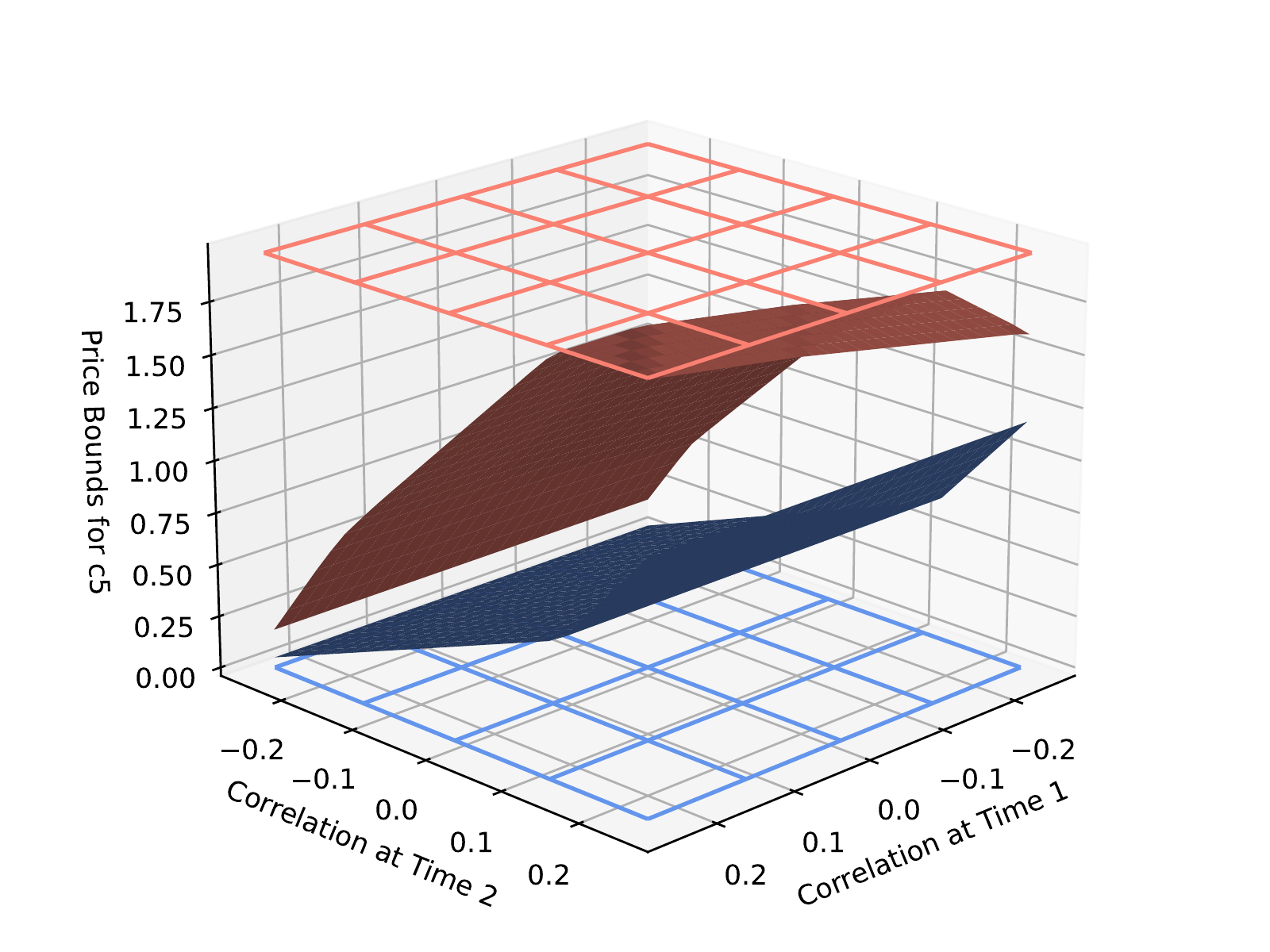}}
    \subfigure[Price Bounds of $c_6$]{
    \includegraphics[width=0.45\textwidth]{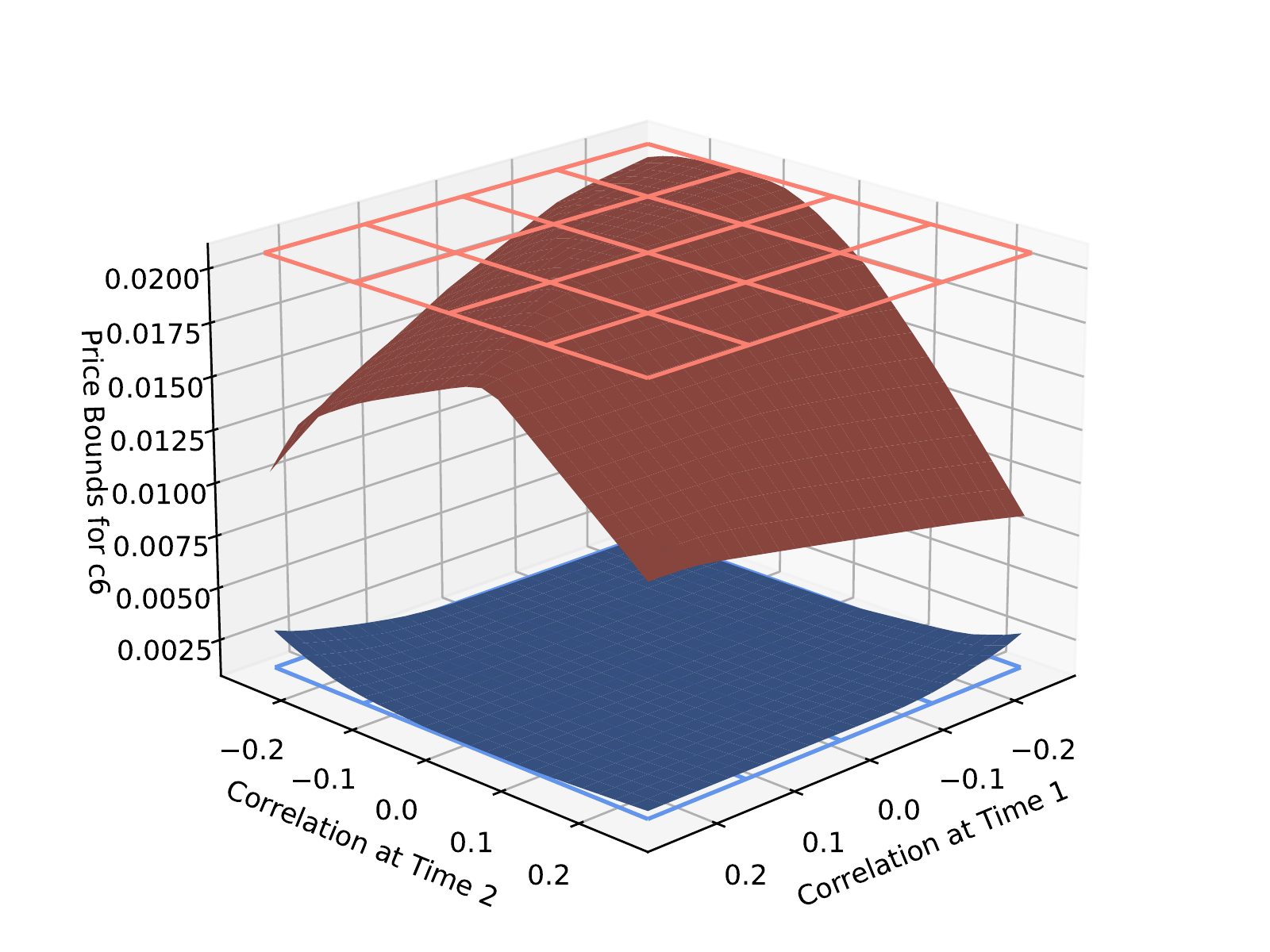}}

              \caption{This figure shows, in the setting of Example~\ref{exa_corr1}, the impact of combined information on the correlations between $S^1$ and $S^2$ at $t_1$ as well as at $t_2$ on the lower (blue) and upper (red) price bounds of derivatives \(c_3,c_4,c_5\), and \(c_6\,.\) The bounds without the consideration of additional information are indicated by colored wireframes.}   \label{fig_3d_correlation}
                  \end{center}
\end{figure}
\end{exa}

\subsection{Additional market-implied assumptions}
In this section, we study how to take into account several additional conditions that reflect observations made on financial markets. In contrast to the inclusion of conditions that are directly linked to the prices of basket options and/or other liquidly traded options these conditions are rather implied by properties that can be observed on financial markets. 
In addition to equality constraints for a pricing measure $\Q$, we will in the sequel also consider inequality constraints. 
For the sake of illustration, we formulate in the following all assumptions only for two underlying assets (i.e. the case $d=2$). It is then straightforward to generalize the implied conditions to a larger number of underlying securities.

\subsubsection{Correlation is constant over time}\label{sec_correlation_constant}
In this section, we assume the risk-neutral correlation between two securities to be constant over time, i.e.,
\begin{equation}\label{eq_correlation_constant}
\operatorname{Corr}_\Q(S_{t_i}^1,S_{t_i}^2)=\operatorname{Corr}_\Q(S_{t_j}^1,S_{t_j}^2)\text{ for all } i,j \in \{1,\dots,n\}.
\end{equation}

The study \cite{adams2017correlations} finds that real-world correlations can reasonably be considered to be constant over time. Based on an empirical analysis of pairs of $40$ stocks, bonds, commodities, and currencies, their findings imply that, for $26\%$ of the pairs, constant real-world correlations for the whole period \(2000\)--\(2014\) can be assumed. For $54\%$ of pairs there appeared exactly one break in the correlation relationship and for $10\%$ there were two breaks. Only for the remaining $10\%$ of the pairs there were three and more breaks in this $14$ year period. The breaks were mostly corresponding to respective crises.

Moreover, \cite{buss2012measuring} assume that $\operatorname{Corr}_\Q(S_{t_i}^1,S_{t_i}^2)-\operatorname{Corr}_\PP(S_{t_i}^1,S_{t_i}^2)=\alpha \left(1-\operatorname{Corr}_\PP(S_{t_i}^1,S_{t_i}^2)\right)$ for some constant $\alpha \in (0,1)$, where $\PP$ denotes the underlying real-world probabilty measure\footnote{Note that real-world correlations can be estimated based on historical observations.}. In combination with \cite{adams2017correlations}, this motivates us to assume a constant risk-neutral correlation as in equation \eqref{eq_correlation_constant}.

Regarding all these empirical findings, equation \eqref{eq_correlation_constant} should not be assumed in all market situations, but can be a reasonable assumption when no break in the correlation relationship,  e.g.  due to a change of market behavior,  is expected or when the time period is short. We refer to \cite{ghosh2021forecasting,krauss2017deep} for a discussion of various time periods over the last $20$ years. We also stress that most options of interest have rather short maturities such that a breakdown in the correlation relationship until maturity is rather unlikely.
The condition \eqref{eq_correlation_constant} can be included  as equality constraints in the dual formulation of the robust pricing problem, as shown in \eqref{eq_correlation_constant_1}.

\subsubsection{Correlation is bounded from below by the real world correlation}\label{sec_correlation_bounded}
Following the argumentation in \cite{buss2012measuring}, i.e., assuming the existence of some $\alpha \in (0,1)$ such that for all $i=1,\dots,n$, it holds $\operatorname{Corr}_\Q(S_{t_i}^1,S_{t_i}^2)-\operatorname{Corr}_\PP(S_{t_i}^1,S_{t_i}^2)=\alpha \left(1-\operatorname{Corr}_\PP(S_{t_i}^1,S_{t_i}^2)\right) \geq 0$, we obtain the condition
\begin{equation}\label{eq_correlation_constraint_below}
\operatorname{Corr}_\Q(S_{t_i}^1,S_{t_i}^2) \geq \operatorname{Corr}_\PP(S_{t_i}^1,S_{t_i}^2) \text{ for all } i =1,\dots,n.
\end{equation}
Since in most situations, the right-hand side ${\operatorname{Corr}_\PP(S_{t_i}^1,S_{t_i}^2)}$ can be well estimated using historical data (see e.g. \cite{eun1984estimating}), we obtain a lower bound for the risk-neutral correlation $\operatorname{Corr}_\Q(S_{t_i}^1,S_{t_i}^2)$.
In combination with the assumption of a time-independent correlation we assume the same lower bound for all risk-neutral correlations. In general, the higher the lower bound for the correlation, the more restrictive is the resulting linear constraint and consequently more significant improvement of robust price bounds can be expected.

\begin{rem}
The estimation of ${\operatorname{Corr}_\PP(S_{t_i}^1,S_{t_i}^2)}$ may be subject to uncertainty, such that ${\operatorname{Corr}_\PP(S_{t_i}^1,S_{t_i}^2)}$ lies within some confidence interval $[\underline{c},\overline{c}]$ with a pre-specified probability. In this case, we can substitute \eqref{eq_correlation_constraint_below} by
\begin{equation}\label{eq_corr_ineq_1}
\operatorname{Corr}_\Q(S_{t_i}^1,S_{t_i}^2) \geq \underline{c}  \text{ for all } i =1,\dots,n.
\end{equation}
\end{rem}

Equation \eqref{eq_corr_constraint_below_1} allows implementing dual strategies of the form described in \eqref{eq_dual_strats1} to incorporate \eqref{eq_corr_ineq_1}, where $\operatorname{Corr}_\PP(S_{t_i}^1,S_{t_i}^2)$ is estimated.

We test in two numerical examples the effect of the additional constraints on the associated robust price bounds.

\begin{exa}\label{exa_improvement_additional_info_table}

We suppose again that $S=(S^1,S^2)$ has the marginal distributions as specified in \eqref{eq_marginals_correlation}, and we consider the four payoff functions in \eqref{payofffc1c4}.
In Table \ref{tbl_improv} we summarize the improvement obtained through incorporating different additional conditions. Price bounds that are improved under additional assumptions are written bold.  The results are computed using a linear programming approach.

\begin{table}[h!]
\resizebox{\textwidth}{!}{
\begin{tabular}{lcccccc} \toprule

 & No additional & Constant & Correlation & Correlation  & Constant corr. & Constant corr.\\
 & assumptions & correlation & lower bounded & lower bounded  & lower bounded & lower bounded\\ 
  &  &  &  by $-0.5$ &  by $0.5$ &  by $-0.5$ &  by $0.5$\\ 
\midrule

\hspace{0.5cm}$\inf_{\Q\in \mathcal{M}^{\mathrm{lin}}} \E_\Q [c_i(S)]$ \\
\midrule
$c_3$&$0.25$ & $\textbf{0.2781}$ &$\textbf{0.3179}$ &$\textbf{0.5375}$  &$\textbf{0.329}$ &$\textbf{0.639}$\\  
$c_4$  &$1.9611$  &$1.9611$  &$1.9611$ &$1.9611$ &$1.9611$  &$1.9611 $\\
$c_5$ & $0.0 $ &$\textbf{0.0795} $ &$0.0$  &$0.0$ &$\textbf{0.0795}$&$\textbf{0.0795}$  \\ 
$c_6$   & $0.0012$ & $0.0012$ & $0.0012$ & $0.0012$ & $0.0012 $ & $\textbf{0.0014}$  \\ 
\midrule

\hspace{0.5cm}$\sup_{\Q\in \mathcal{M}^{\mathrm{lin}}} \E_\Q [c_i(S)]$ \\
\midrule
$c_3$ & $1.0111  $ & $\textbf{0.9781}  $ & $1.0111 $& $1.0111$  & $\textbf{0.9781}$ & $\textbf{0.9781}$\\
$c_4$ & $3.2167 $ & $ \textbf{3.198}$ & $\textbf{3.1615 }$ & $\textbf{2.9714}$ & $\textbf{3.1615 }$  & $\textbf{2.893}$\\
$c_5$ & $1.9778 $ & $1.9778 $ & $1.9778 $ & $\textbf{0.8083}$  & $1.9778 $ & $\textbf{0.6784}$ \\ 
$c_6$ & $0.0207$ & $0.0207$ & $0.0207 $  &$0.0207$ & $0.0207$ & $0.0207$ \\ 
\bottomrule
\end{tabular}}
\caption{Improvement of the price bounds described in Example~\ref{exa_improvement_additional_info_table} under different additional assumptions}\label{tbl_improv}
\end{table}

Although the marginal distributions possess a quite simple discrete structure, the results described in Table~\ref{tbl_improv} allow several important insights concerning the effect of additional constraints on the resultant price bounds. First, we observe that the improvements are highly payoff-dependent. Indeed, while the price bounds for $c_7$ are barely affected through the inclusion of additional constraints, the price bounds for $c_4$ can be improved strongly by any kind of constraint we investigated. Second, the improvements can either concern only the lower bound (e.g. correlation constrained from below by -$0.5$ for $c_4$), only the upper bound (correlation constrained from below by $0.5$ for $c_6$) or affect both bounds (constant correlation for $c_4$). Third, a combination of different constraints can improve the price bounds even more than the sum of the improvements of both constraints when considered separately (upper bound of $c_5$ in the case that the correlation is constant and constrained from below by $0.5$).
\end{exa}

\subsection{Real-world examples}
In this section we study price bounds of multi-asset derivatives with underlying marginal distributions that are implied from real market data. In particular, we study how the price bounds behave under additional constraints on the joint distributions.

\subsubsection*{Deriving the marginals}

On $t_0=$ 17th August 2020, we observe prices of put and call options written on $S^1:=$ the stock of \emph{Apple Inc.} and on $S^2:=$ the stock of \emph{Microsoft Corp.} We take into account options with maturities lying $11$ days and $32$ days ahead respectively. This means we set 
$
t_1 - t_0 = 11/365$ and $t_2 - t_0 = 32/365.
$

We consider mid prices of call and put options, i.e., we take the average of bid and ask prices. These prices are then cleaned in two ways: the mid prices shall not allow for static arbitrage (call prices should decrease w.r.t.\,increasing strikes, put prices should increase w.r.t.\,increasing strikes). Further, we exclude butterfly arbitrage involving these prices, basically meaning prices as a function of the strikes should possess a convex shape.

After having cleaned the prices we apply the Breeden-Litzenberger result\footnote{We refer also to \cite{talponen2014note} for a multidimensional version of \cite{breeden1978prices}, as well as  \cite{neufeld2022numerical} for a non-asymptotic version of \cite{breeden1978prices}, \cite{talponen2014note}.} in \cite{breeden1978prices} to obtain marginal distributions associated to the underlying securities at maturities $t_1,t_2$. The density of the marginals can be computed as the second derivative of the prices w.r.t. the strikes. For this step, to approximate the second derivative, we use the finite differences method, i.e., given strikes $(K_j)_{j=1,\dots,N_{\operatorname{strikes}}}$ with $N_{\operatorname{strikes}} \in \N$ and mid (call or put) prices $\left(\operatorname{P}(K_{j},t_i)\right)_{j=1,\dots,N_{\operatorname{strikes}}}$, the time-$t_i$ density $\operatorname{p}_i(K_j)$ evaluated at $K_j$ for $j =2,\dots,N_{\operatorname{strikes}}$ is approximated by
$
\frac{\partial^2 \operatorname{P}(K,t_i)}{\partial K^2}\big|_{K=K_j} \approx \operatorname{p}_i(K_j):=\frac{\operatorname{P}(K_{j+1},t_i)-2\operatorname{P}(K_{j},t_i)+\operatorname{P}(K_{j-1},t_i)}{(K_{j+1}-K_{j-1})^2}
$
and we further set $\operatorname{p}_i(K_1)=\operatorname{p}_i(K_{N_{\operatorname{strikes}}})=0$.
We then approximate the one-dimensional marginal distribution of the asset through
$
S_{t_i}^k \sim \frac{1}{\sum_{j=1}^{N_{\operatorname{strikes}}} \operatorname{p}_i(K_j)} {\sum_{j=1}^{N_{\operatorname{strikes}}}} \delta_{K_j}\operatorname{p}_i(K_j) ~~~ \text{ for }i,k=1,2,
$
where $\delta_{K_j}$ denotes the Dirac measure at point $K_j$.

To ensure an increasing convex order of the marginals of each stock we equalize the means of $S_{t_1}^j,S_{t_2}^j$ for $j = 1,2$, compare also \cite{alfonsi2019sampling}. Finally we apply $\mathcal{U}$-quantization introduced in \cite[Section 2.4.]{baker2012martingales} in a similar way as in \cite[Section 3]{neufeld2021deep} such that each marginal is supported on $20$ values which can then be implemented into a linear program to compute robust price bounds. Additionally, we remark that we neglect interest rates and dividend yields for these rather short maturities.

\subsubsection*{Computation of price bounds under correlation information}
We study the payoff functions of derivatives \(c_3,c_4,c_5,\) and \(c_6\) given by \eqref{payofffc1c4}, where we modify $c_3$ and $c_4$ by considering a strike of $250$, i.e., we have
$
c_3(S_{t_1}^1,S_{t_2}^1,S_{t_1}^2,S_{t_2}^2):=\left(1/4\cdot(S_{t_1}^1+S_{t_2}^1+S_{t_1}^2+S_{t_2}^2)-250\right)_+ $ and $c_4(S_{t_1}^1,S_{t_2}^1,S_{t_1}^2,S_{t_2}^2):=\left(250-\min\left\{S_{t_1}^1,S_{t_2}^1,S_{t_1}^2,S_{t_2}^2\right\}\right)_+.
$
In Figure~\ref{fig_real_data_correlations_improvement}, we display the influence of information on the time-$t_1$ and time-$t_2$ correlation, respectively, on the price bounds of these derivatives. 
As elaborated, such information can be extracted from prices of basket options, if observable. Since we have no access to price quotes of basket options, we instead show the improvement obtained if certain levels of correlation are given as an input.
As already observed in the examples with artificial marginals, in general, the improvement of the price bounds becomes stronger for information concerning the time $t_2$ correlation.
\vspace{-0.5cm}
\begin{figure}[h!] 
\begin{center}
{\includegraphics[width=0.95\textwidth]{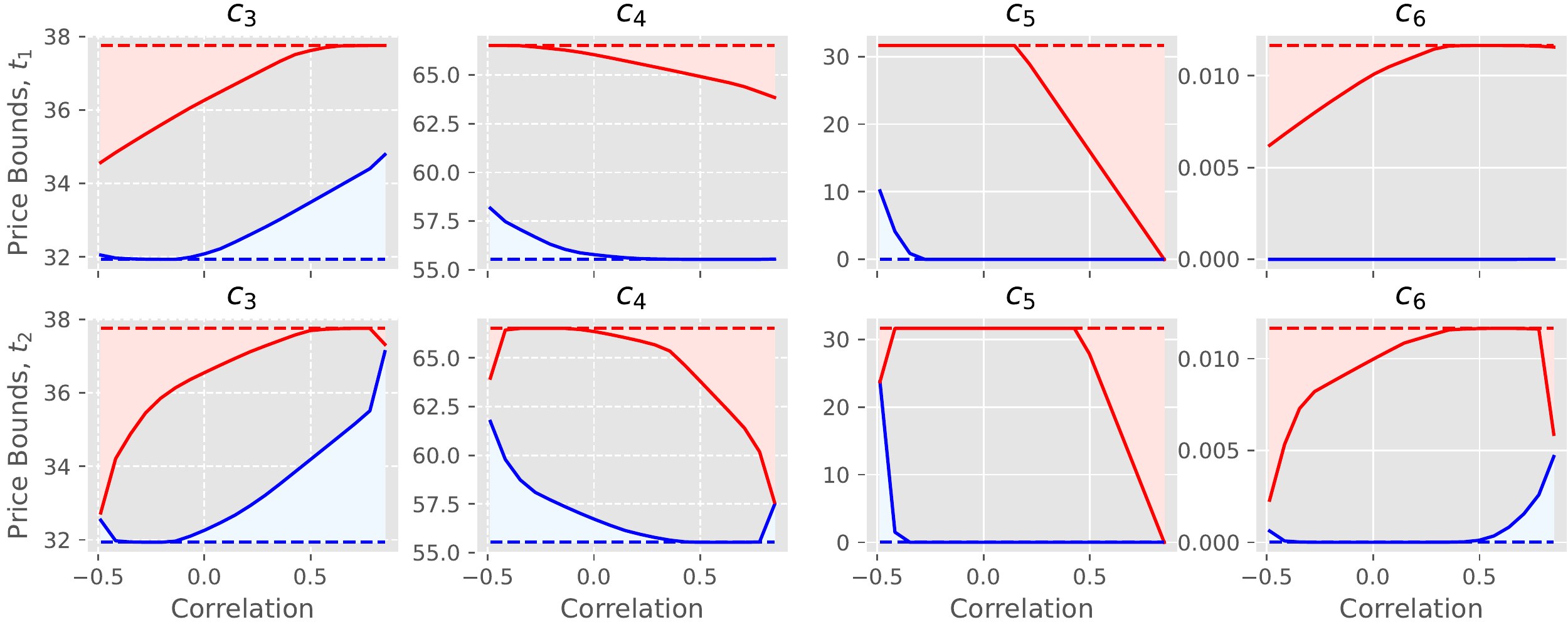}}
\end{center}
\caption{The price bounds of the options $c_3,c_4,c_5,c_6$ in dependence of correlation information (either regarding time $t_1$ or time $t_2$) while the marginal distributions are implied by vanilla option prices written on Apple and Microsoft.}\label{fig_real_data_correlations_improvement}
\end{figure}

\subsubsection{Computation of price bounds under additional assumptions}
Eventually, we investigate the influence of additional assumptions on the price bounds.
We first observe that, under the real-world measure $\PP$, which is here set to be the empirical measure based on historical data from $2$ January $2018$ until
$17$th August $2020$, the stocks of \emph{Apple Inc.}\,and \emph{Microsoft Corp.}\,seem to be highly correlated.


The idea is to make use of this apparently strong relation between the two assets to obtain tighter price bounds for derivatives $c_i$ written on both assets by using only such pricing measures that are consistent with an assumption on the strictly positive correlation.

To obtain an indication for the level of the correlation between the two assets in an $11$ and $32$ day period, we consider the empirical bivariate return distribution of the two assets in an observation period ranging from $2$nd January $2018$ until $22$ July $2020$. From this empirical distribution we simulate in a bootstrapping approach $100,000$ paths of length $11$ and $32$ respectively, and then compute the respective correlation coefficients. We obtain an estimate of $0.7948$ for $\operatorname{Corr}_\PP(S_{t_1}^1,S_{t_1}^2)$ and an estimate of $0.7952$ for $\operatorname{Corr}_\PP(S_{t_2}^1,S_{t_2}^2)$.

 Thus, evidence is provided to include the weaker assumption $\operatorname{Corr}_\Q(S_{t_i}^1,S_{t_i}^2) \geq \operatorname{Corr}_\PP(S_{t_i}^1,S_{t_i}^2)\geq 0.75$ for $i=1,2$ and pricing measures $\Q$. Here, $0.75$ can obviously be substituted by any other number associated to another degree of physical correlation that is believed to be more accurate. The higher this number, the more improvement of the price bounds can be expected. However, to compute price bounds among consistent martingale measures $\Q$, this number must lie within the interval of correlations that are consistent with the marginals.
According to Remark~\ref{rem_differences_martingale_property}, the bounds of this interval can be computed through minimizing and maximizing
$
 \Q \mapsto \E_\Q\left[\frac{S_{t_i}^1S_{t_i}^2-S_{t_0}^1S_{t_0}^2}{\sqrt{\E_{\mu_i^1}[\left(S_{t_i}^1\right)^2]-\left(S_{t_0}^1\right)^2}\sqrt{\E_{\mu_i^2}[\left(S_{t_i}^2\right)^2]-\left(S_{t_0}^2\right)^2}}\right]
$
w.r.t.\,measures $\Q \in \mathcal{M}(\mu)$. We compute the bounds as solutions of linear programming problems using the empirical distributions derived according to the presented discretization scheme and obtain at $t_1$ the interval $[-0.588,0.799]$ and at $t_2$ the interval of possible prices given by $[-0.487,0.8490]$.

In Table~\ref{tbl_improv_real_data}, we display the results revealing that indeed the assumption on a lower bound of the correlation has a strong impact on the quality of the price bounds. The assumption on constant correlations has only in combination with the assumption of the lower bound of the correlation an influence on the lower bound of $c_4$, while for the other price bounds we cannot report any influence of the assumption of constant correlations. For the sake of readability, price bounds showing improvement in comparison with the original bounds are displayed in bold characters.




{\small
\begin{table}[h!]
\begin{adjustbox}{angle=0}
\begin{tabular}{lcccc} \toprule

 & No additional & Constant & Correlation & Constant correlation \\
 & assumptions & correlation & lower bounded by $0.75$ & lower bounded by $0.75$\\
\midrule

\hspace{0.5cm}$\inf_{\Q\in \mathcal{M}^{\mathrm{lin}}} \E_\Q [c_i(S)]$ \\
\midrule
$c_3$&$31.9339$ & $\textbf{31.9344}$ &$\textbf{34.4798}$ &$\textbf{35.3573}$\\  
$c_4$  &$55.5357$  &$55.5357$  &$\textbf{ 55.5359}$ &$\textbf{55.5398}$  \\
$c_5$ & $0.0 $ &$0.0 $ &$0.0$  &$0.0$  \\ 
$c_6$   & $0.0$ & $0.0$ & $0.0$ & $\textbf{0.0023}$ \\ 

\midrule

\hspace{0.5cm}$\sup_{\Q\in \mathcal{M}^{\mathrm{lin}}} \E_\Q [c_i(S)]$ \\
\midrule
$c_3$ & $37.7576  $ & $37.7576  $ & $37.7576 $& $37.7576 $ \\
$c_4$ & $ 66.5152 $ & $ \textbf{66.3507} $ & $\textbf{64.0369 } $ & $\textbf{60.6242 } $\\
$c_5$ & $31.6319 $ & $31.6319 $ & $\textbf{2.1662} $ & $\textbf{2.1662} $  \\ 
$c_6$ & $0.0116$ & $0.0116$ & $0.0116  $  & $0.0116  $\\ 

\bottomrule
\end{tabular}
}
\caption{Improvement of the price bounds under different additional assumptions}\label{tbl_improv_real_data}
\end{adjustbox}
\end{table}
}
\section{Proofs}\label{section_proofs}
In this section, we provide all proofs that were omitted in the previous sections of the paper.

To verify the lower semicontinuity of $\mathcal{M}^{\operatorname{lin}} \ni \Q \mapsto \E_\Q[c]$ and to prove the compactness of the set \(\cM^{\mathrm{lin}}\) we make use of the following continuity result, see \cite[Lemma 2.2]{beiglbock2013model} and \cite[Lemma 4.3]{villani2008optimal}.

\begin{lem}\label{lemcont}
Let \(f\colon \R^{nd} \to \R\) be (lower/upper semi-)continuous and linearly bounded. Then, the mapping
$
\pi \mapsto \int_{\R^{nd}} f(x) \de \pi(x)
$
is (lower/upper semi-)continuous on \(\Pi(\mu)\) w.r.t.\,weak topology.
\end{lem}

%

\begin{proof}[Proof of Proposition~\ref{theaddcon}]\label{protheaddcon}
To prove the duality in \eqref{gendures}, we abbreviate
\begin{align*}
\Phi_{(\delta_i^k),(\alpha_i),(\beta_i)}(x)&:=\Psi_{(u_i^k),(\delta_i^k),(\alpha_i),(\beta_i)}(x)-\sum_{i=1}^n\sum_{k=1}^d u_i^k(x_i^k)\\
\phantom{:}= \sum_{i=1}^{n-1}\sum_{k=1}^d \delta_i^k&(x_1,\ldots,x_i)(x_{i+1}^k-x_i^k) + \sum_{i \in \mathcal{I}^{\operatorname{eq}}}\alpha_i (f_i^{\operatorname{eq }}(x)-K_i^{\operatorname{eq }}) + \sum_{i \in \mathcal{I}^{\operatorname{ineq}}} \beta_i (f_i^{\operatorname{ineq }}(x)-K_i^{\operatorname{ineq }}) 
\end{align*}
for $x=(x_1,\ldots,x_n)=(x_1^1,\ldots,x_n^d)\in \R_+^{nd}$, $\delta_i^k \in C_b(\R_+^{id})$, $\alpha_i \in \R$, $\beta_i \in \R_+$, \(i=1,\dots,n,k=1,\dots,d\,\)  such that $\alpha_i=0, \beta_j=0$ for all but finitely many $i\in \mathcal{I}^{\operatorname{eq}},~ j \in \mathcal{I}^{\operatorname{ineq}}$. Note that $\Phi_{(\delta_i^k),(\alpha_i),(\beta_i)} \in L_{\operatorname{lin}}(\R_+^{nd})$, because all \(\delta_i^k\) and \(f_i^{\operatorname{eq }}\) are linearly bounded and continuous for all $i,k$, and all \(\beta_if_i^{\operatorname{ineq }}\) are linearly bounded and lower semicontinuous. Therefore, we obtain that
\begin{align}
\nonumber \underline{\cD}_{\cS} &= \inf_{(u_i^k),(\delta_i^k),(\alpha_i),(\beta_i)}~\inf_{\Psi_{(u_i^k),(\delta_i^k),(\alpha_i),(\beta_i)}\geq ~c} \left\{\sum_{i=1}^n \sum_{k=1}^d \int_{\R_+} u_i^k(x_i^k) \de \mu_i^k(x_i^k)\right\}\\
\nonumber &= \inf_{(u_i^k),(\delta_i^k),(\alpha_i),(\beta_i)} ~\inf_{u_1^1\oplus \cdots \oplus u_n^d \geq ~c -\Phi_{(\delta_i^k),(\alpha_i),(\beta_i)} }\left\{\sum_{i=1}^n \sum_{k=1}^d \int_{\R_+} u_i^k(x_i^k) \de \mu_i^k(x_i^k)\right\}\\
\label{eqproduth1}&= \inf_{(\delta_i^k),(\alpha_i),(\beta_i)} ~ \sup_{\pi\in \Pi(\mu)}\left\{ \int_{\R_+^{nd}} \left( c(x)-\Phi_{(\delta_i^k),(\alpha_i),(\beta_i)} (x)  \right)\de \pi(x) \right\}
\end{align}
Indeed, Equation \eqref{eqproduth1} is a consequence of the Monge--Kantorovich duality (see \cite[Proposition 2.1]{beiglbock2013model} and \cite[Chapter 5]{villani2008optimal}) using that \(c-\Phi_{(\delta_i^k),(\alpha_i),(\beta_i)}\in U_{\operatorname{lin}}(\R_+^{nd})\) because \(c\) and \(\Phi_{(\delta_i^k),(\alpha_i),(\beta_i)}\) are both linearly bounded and $c$ is upper semicontinuous, whereas $\Phi_{(\delta_i^k),(\alpha_i),(\beta_i)}$ is lower semicontinuous.
Next, we apply the minimax theorem by Ky--Fan, see, e.g., \cite[Lemma 3.1]{Rueschendorf-2019}, to the compact convex set \(\Pi(\mu)\,,\) the convex set 
\begin{align*}\bigg\{ (\delta_i^k),(\alpha_i),(\beta_i) \in \left(C_b(\R_+^{d})\times \cdots \times C_b(\R_+^{(n-1)d})\right) \times \R^{|\cI^{\operatorname{eq }}|}\times \R_+^{|\cI^{\operatorname{ineq }}|} ~&\text{ s.t.}~\\ &\hspace{-8cm}\alpha_i=0,~\beta_j=0\text{ for all but finitely many } i\in \mathcal{I}^{\operatorname{eq}},~j \in \mathcal{I}^{\operatorname{ineq}} \bigg\},
\end{align*}
 and the function \(f\) given by
\begin{align*}
f\big(\pi,((\delta_i^k),(\alpha_i),(\beta_i))\big):=\int_{\R_+^d} \left(c(x)- \Phi_{(\delta_i^k),(\alpha_i),(\beta_i)}(x)\right) \de \pi(x).
\end{align*}
The compactness of \(\Pi(\mu)\) is meant w.r.t.\,the weak topology and can be obtained from \cite[Lemma 4.4]{villani2008optimal} and Prokhorov's Theorem.
Due to Lemma \ref{lemcont}, \(f\) is upper semicontinuous in \(\pi\) because \(c-\Phi_{(\delta_i^k),(\alpha_i),(\beta_i)}\) is upper semicontinuous and linearly bounded.
Further, \(f\) is linear, and thus concave in $\pi$, as well as linear w.r.t.\,\(((\delta_i^k),(\alpha_i),(\beta_i))$ and hence convex.
Thus, the minmax theorem is indeed applicable and we obtain with Equation \eqref{eqproduth1} that
\begin{align}
\label{eqproduth2} \underline{\cD}_{\cS} &= \sup_{\pi \in \Pi(\mu)} ~ \inf_{(\delta_i^k),(\alpha_i),(\beta_i)} \left\{ \int_{\R_+^{nd}} c(x) \de \pi(x) - \int_{\R_+^{nd}} \sum_{i=1}^{n-1}\sum_{k=1}^d \delta_i^k(x_1,\ldots,x_i)(x_{i+1}^k-x_i^k) \de \pi(x) \right. \\
\label{eqproduth3}&\left.~~~~~~~~~~~~~~~~~~~- \int_{\R_+^{nd}} \sum_{i \in \mathcal{I}^{\operatorname{eq}}} \alpha_i (f_i^{\operatorname{eq }}(x)-K_i^{\operatorname{eq }}) \de \pi(x) - \int_{\R_+^{nd}} \sum_{i \in \mathcal{I}^{\operatorname{ineq}}} \beta_i (f_i^{\operatorname{ineq }}(x)-K_i^{\operatorname{ineq }}) \de \pi(x)\right\}\\
\label{eqproduth4}&= \sup_{\Q\in \cM^{\operatorname{lin}}} \int_{\R_+^{nd}} c(x) \de \Q(x) = \overline{P}_{\cM^{\operatorname{lin}}}.
\end{align}

For the equality in \eqref{eqproduth4}, we note that \(\int c \de \pi\) is uniformly bounded in \(\pi\in \Pi(\mu)\) using that all \(\mu_i^k\) have finite first moments and that $c$ is linearly bounded. We observe that the second integral in \eqref{eqproduth2} vanishes whenever \(\pi\in \Pi(\mu)\) is a martingale. If \(\pi\in \Pi(\mu)\) fulfils the equality constraints \(\E_\pi [f_i^{\operatorname{eq }}]=K_i^{\operatorname{eq }}\) for all \(i\,,\) the first integral in \eqref{eqproduth3} is \(0\,,\) and if \(\pi\in \Pi(\mu)\) fulfils the inequality constraints \(\E_\pi [f_i^{\operatorname{ineq }}]\leq K_i^{\operatorname{ineq }}\) for all \(i,\) the second integral in \eqref{eqproduth3} is non-positive using that \(\beta_i\geq 0\,.\) Hence, for \(\pi\in \cM^{\operatorname{lin}}\,,\) the infimum of the expression in the curly brackets is given by $\int c(x)\de \pi(x)>-\infty$. If \(\pi\in \Pi(\mu)\) is not a martingale or does not fulfil one of the equality or inequality constraints, then there exist \(\delta_i^k\,,\) \(\alpha_i\,,\) and \(\beta_i\,,\) respectively, such that at least one of the corresponding integrals is positive. By scaling, we conclude that in this case the infimum over \((\delta_i^k),(\alpha_i),(\beta_i)\) is \(-\infty\,.\)

Next, we prove that the supremum is attained. By \cite[Proposition 2.4]{beiglbock2013model}, the set \(\cM(\mu)\) is compact in the weak topology. We show that \(\cM^{\operatorname{lin}}\) is a closed subset of \(\cM(\mu)\,.\) Let \((\pi_m)_{m\in \N}\subset \cM^{\operatorname{lin}}\) be a sequence that converges weakly to some $\pi \in \mathcal{M}^{\operatorname{lin}}$.
Then, Lemma \ref{lemcont} implies for all $i \in \mathcal{I}^{\operatorname{eq}}$ that \(K_i^{\operatorname{eq }}=\E_{\pi_n} [f_i^{\operatorname{eq }}]\to \E_\pi [f_i^{\operatorname{eq }}]\) as \(n\to \infty\,,\) and, thus, \(\E_\pi [f_i^{\operatorname{eq }}]=K_i^{\operatorname{eq }}\,,\) where we use that \(f_i^{\operatorname{eq }}\) is linearly bounded and continuous. For the inequality constraints, we obtain from Lemma \ref{lemcont} that for all $i \in \mathcal{I}^{\operatorname{ineq}}$ we have
$
\E_\pi [f_i^{\operatorname{ineq }}] \leq \liminf_{n\to \infty} \E_{\pi_n} [f_i^{\operatorname{ineq }}] \leq K_i^{\operatorname{ineq }}
$
using that \(f_i^{\operatorname{ineq }}\) is linearly bounded and lower semicontinuous.
 Hence, also \(\cM^{\operatorname{lin}}\) is compact. Now, let \((\Q_m)_m\subset \cM^{\operatorname{lin}}\) be a sequence such that \(\overline{P}_{\cM^{\operatorname{lin}}} \leq \E_{\Q_m}[c] + \frac 1 m\) for all \(m\in \N\,.\) Since \(\cM^{\operatorname{lin}}\) is compact, there exists a measure \(\Q^*\in \cM^{\operatorname{lin}}\) and a subsequence \((\Q_{m_k})_{k \in \N}\) which converges weakly to \( \Q^*\,.\) Since \(c\) is upper semicontinuous and linearly bounded, we obtain from Lemma \ref{lemcont} that \(\limsup_{k\to \infty} \E_{\Q_{m_k}} [c] \leq \E_{\Q^*}[c]\,.\) This implies that \(\overline{P}_{\cM^{\operatorname{lin}}}=\E_{\Q^*}[c]\,.\)\\

\end{proof}

For the proof of Theorem \ref{theqcub}, we apply the following lemmas.

\begin{lem}\label{lemsufu}
Let $m \in \N$, let $Q_2 \in \mathcal{Q}_2$. Then, it holds that \(Q^*\in \cQ_m\) defined in \eqref{deqcqs} is a quasi-copula with survival function \(\widehat{Q^*}\) given by
$
\widehat{Q^*}(u_1,\ldots,u_m)=\widehat{Q_2}\left(\max_{2\leq i\leq m} \{u_i\},u_1\right)$ for  $(u_1,\ldots,u_m)\in [0,1]^m\,.
$
\end{lem}

\begin{proof}[Proof of Lemma~\ref{lemsufu}]
The function \(Q^*\) fulfils the defining properties of a quasi-copula because \(Q_2\) is a quasi-copula.\\
Due to the definition of \(\widehat{Q^*}\,,\) we may consider for \((u_1,\ldots,u_m)\in [0,1]^m\) w.l.o.g.\, the case that \(u_2\geq \ldots\geq u_m\,.\)
Then, it holds true that
\begin{equation}
\label{eq_long_equation_1}
\begin{aligned}
\widehat{Q^*}(u_1,\ldots,u_m)&= \sum_{I\subseteq \{1,\ldots,m\}\atop v_i:=1 ~\forall i\in I\,,~v_i:=u_i~\forall i\notin I} (-1)^{m-|I|} Q^*(v_1,\ldots,v_m)\\
&=1-\sum_{J\subseteq \{1,\ldots,m\}, J\ne \emptyset\,,\atop v_i:=u_i ~\forall i\in J\,, ~v_i:=1 ~\forall i\notin J} (-1)^{|J|+1} Q^*(v_1,\ldots,v_m)\\
&=1 - u_1
- \sum_{ J\subseteq \{1,\ldots,m\}, \atop 1\in J, |J|\geq 2} (-1)^{|J|+1} Q_2\left(\min_{j\in J\setminus \{1\}} \{u_j\} , u_1\right) - \sum_{ J \subseteq \{2,\ldots,m\},\atop J \neq \emptyset} (-1)^{|J|+1} \min_{j\in J}\{ u_j\}.
\end{aligned}
\end{equation}
Indeed, to see that \eqref{eq_long_equation_1} holds, note that the first equality follows from the definition of a survival function in \eqref{defsurvfun}. For the second equality, we sum over \(J=\{1,\ldots,m\}\setminus I\) and use that $Q^*(1,\dots,1)=1$. The third equality follows from the definition of \(Q^*\) and the uniform marginal property of Definition \ref{def_quasicop}~(b) for quasi-copulas. With \eqref{eq_long_equation_1} we obtain 
\begin{equation}
\label{eq_long_equation_1_split_2}
\begin{aligned}
\widehat{Q^*}(u_1,\ldots,u_m)&= 1 - u_1 
- \sum_{k=2}^m \sum_{j=0}^{k-2} (-1)^{j+1} \binom{k-2}{j} Q_2(u_k,u_1)- \sum_{k=2}^m \sum_{j=0}^{k-2} (-1)^{j} \binom{k-2}{j} u_k\\
&= 1-u_1+ Q_2(u_2,u_1)-u_2\\
&= 1- u_1-\max_{2\leq i\leq m}\{u_i\} + Q_2(\max_{2\leq i\leq m}\{u_i\},u_1)= \widehat{Q_2}(\max_{2\leq i\leq m} u_i,u_1).
\end{aligned}
\end{equation}
Indeed, to see that \eqref{eq_long_equation_1} note that the first equality holds true for the following reason: in the first sum, we consider for every $k=2,\ldots,m$, the subsets $J\subseteq \{1,\ldots,k\}$ with $1,k\in J$. Then $k$ is the maximal element of $J$ and hence $\min_{j\in J\setminus \{1\}}\{u_j\}=u_k$. There are \(\binom{k-2}{j}\) subsets of \(\{2,\ldots,k-1\}\) with \(j\) elements and we have $|J|=j+2$. In the second sum, we consider subsets $J\subseteq \{2,\ldots,k\}$ with $k\in J$ for every $k=2,\ldots,m$. Here again there are \(\binom{k-2}{j}\) subsets of \(\{2,\ldots,k-1\}\) with \(j\) elements but now $|J|=j+1.$ The second equality follows from the symmetry of the binomial coefficients given by \(\sum_{i=0}^N \binom{N}{i} (-1)^i = \one_{\{N=0\}}\) for $N \in \N$ and the third equality by the definition of the survival function.
\end{proof}


Denote by \(\K_n^m:=\{\tfrac 1 {n+1},\tfrac 2{n+1},\ldots,\tfrac {n}{n+1}\}^m \subset [0,1]^m\) the canonical \(m\)-dimensional \(n\)-grid with edge length \(\tfrac 1 {n+1}\) contained in \([0,1]^m\,.\) Denote by \(\diag(\K_n^m):=\{(\tfrac 1 {n+1},\ldots,\tfrac 1 {n+1}),\ldots,(\frac n {n+1},\ldots,\frac n {n+1})\}\) the diagonal of \(\K_n^m\,.\) For a finite signed measure \(\sigma\) on \([0,1]^m\,,\) we define by
\begin{equation}\label{eq_def_g_sigma}
G_\sigma(u_1,\ldots,u_m):=\sigma([0,u_1],\ldots,[0,u_m])\,, \quad (u_1,\dots,u_m) \in [0,1]^m\,,
\end{equation}
its \emph{measure generating function}.

Conversely, by embedding $\K_n^m \subseteq [0,1]^m$ and identifying $f: \K_n^m \rightarrow \R$ with $\overline{f}:[0,1]^m \rightarrow \R$ defined by $ \overline{f}(x_1,\dots,x_m):= f \left(\tfrac{\lfloor x_1\cdot (n +1)\rfloor}{n+1} \wedge \tfrac{n}{n+1},\dots,\tfrac{\lfloor x_n\cdot (n +1)\rfloor}{n+1} \wedge \tfrac{n}{n+1}\right)$ where $\lfloor x \rfloor := \max \{ n \in \N_0: n \leq x \}$, $x \in \R_+$, we see that 
every function $f:\K_n^m \to \R$ has bounded Hardy--Krause variation as $\K_n^m \subseteq [0,1]^m$ is discrete, and hence induces a finite signed measure on $\K_n^m$. We refer to the discussion after \eqref{deffmarg}, see also, e.g. \cite[Theorem 3]{Aistleitner-2015}  or \cite[Theorem 3.29]{folland1999real}.
\begin{lem}\label{lemmadia}
Let $m,n \in \N$, and let \(\sigma\) be a finite signed measure on \(\K_n^m\,.\) Then the following statements hold true:
\begin{itemize}
\item[(a)] \label{lemmadia1} There exists a finite signed measure \(\mu\) on \(\K_n^1\) such that
\begin{align}\label{eqmadia}
G_\sigma(u_1,\ldots,u_m)=G_\mu(\min_{i=1,\ldots,m}\{u_i\})~~~ \text{for all }(u_1,\ldots,u_m)\in \K_n^m\,,
\end{align}
if and only if the mass of \(\sigma\) is concentrated on the diagonal of \(\K_n^m\,\), denoted by  \(\diag(\K_n^m)\),  i.e., \(\sigma(\{x\})=0\) for all \(x\in \K_n^m\setminus \diag(\K_n^m)\,.\)
In this case $\mu$ is defined by 
\begin{equation}\label{eq_def_mu_sigma}
\mu([0,u]):=\sigma\left([0,u]\times \cdots \times [0,u]\right),~u \in [0,1].
\end{equation}
\item[(b)] \label{lemmadia2} Assume that \(\sigma(\K_n^m)=1\,.\) In the case that \(\sigma\) fulfils \eqref{eqmadia} for some signed measure \(\mu\) on \(\K_n^1\,,\) it follows that
\begin{align}\label{eqmadia2}
\int_{[0,1]^m} f(u_1,\ldots,u_m) \de G_\sigma (u_1,\ldots,u_m) = \int_{[0,1]} f(v,\ldots,v) \de G_\mu(v)
\end{align}
for all \(\sigma\)-integrable functions \(f\colon [0,1]^m \to \R\,.\)
\end{itemize}
\end{lem}

\begin{proof}
To show (a), first assume  that \eqref{eqmadia} holds and let \(x=(x_1,\ldots,x_m)\in \K_n^m\setminus \diag(\K_n^m)\,.\) Then, 
\begin{align*}
\sigma(\{x\})&= \triangle_{1/n}^1 \cdots \triangle_{1/n}^m G_\sigma(x_1,\ldots,x_m)= \triangle_{1/n}^1\cdots \triangle_{1/n}^m G_\mu\left(\min_{i\in \{1,\ldots,m\}}\{x_i\}\right) = 0\,,
\end{align*}
because there exists \(j\in \{1,\ldots,m\}\) such that \(x_j>\min_{i\in \{1,\ldots,m\}}\{x_i\}\) and thus
\begin{align*}
\triangle_{1/n}^j G_\mu\left(\min_{i\in \{1,\ldots,m\}}\{x_i\}\right)= G_\mu\left(\min_{i\ne j}\{x_i\}\right)-G_\mu\left(\min_{i\ne j}\{x_i\}\right) = 0\,.
\end{align*}
For the reverse direction, assume that the mass associated with \(\sigma\) is concentrated on \(\diag(\K_n^m)\). Then \(\mu\) defined by \(\mu([0,u]):=\sigma([0,u]\times \cdots \times [0,u])\,,\) \(u\in \h_n^1\,,\) is a signed measure with the property that
\begin{align*}
G_\sigma(u_1,\ldots,u_m)&=\sigma([0,u_1]\times \cdots \times [0,u_m]) \\&= \sigma([0,\min_{i=1,\ldots,m} \{u_i\}]^m) = \mu([0,\min_{i=1,\ldots,m} \{u_i\}]) = G_\mu(\min_{i=1,\ldots,m} \{u_i\})\,,
\end{align*}
where the second equality holds true because \(\sigma(\{x\})=0\) for all \(x\in \h_n^m\setminus \diag(\h_n^m)\,.\)

To show (b), let us first consider the case where \(\sigma\) is a probability measure, i.e., all mass (which is by (a) distributed on \(\diag(\h_n^m)\)) is non-negative. 

Let $U_1,\dots,U_m$ be random variables on a probability space $(\Omega,\mathcal{A},\sigma)$ such that $(U_1,\dots,U_m) \sim \sigma$. Since $\sigma$ is concentrated on the diagonal, we have that $U_i \eqd U_j$ and that $U_i,U_j$ are comonotone for all $i,j\in \{ 1,\dots,m\}$. Hence, with $U: \eqd U_1$, we obtain $(U_1,\ldots,U_m)\eqd (U,\ldots,U)$.
This implies
\begin{align}\label{eqmadia3}
\begin{split}
\int_{[0,1]^m} f(u_1,\ldots,u_m)\de G_\sigma (u_1,\ldots,u_m) &= \int_{\Omega} f(U_1,\ldots,U_m)\de P\\
&= \int_{\Omega}f(U,\ldots,U)\de P = \int_{[0,1]} f(v,\ldots,v)\de G_\mu(v)\,,
\end{split}
\end{align}
which proves \eqref{eqmadia2} in the case where \(\sigma\) is a probability measure.

Now, consider the general case where \(\sigma\) is a finite signed measure satisfying $\sigma(\h_n^m)=1$. By \eqref{lemmadia1} all mass of \(\sigma\) is concentrated on \(\diag(\h_n^m)\) which is a finite set. So, there exists \(M\in \N\) such that \(\sigma(x)\geq -M\) for all \(x\in \diag(\h_n^m)\,.\) Denote by \(\sigma^u\) the uniform distribution on \(\diag(\h_n^m)\,,\) i.e., \(\sigma^u(\{x\})=\tfrac 1 n\) for all \(x\in \diag(\h_n^m)\) and \(\sigma^u(\{x\})=0\) for all \(x\in \h_n^m\setminus \diag(\h_n^m)\,.\) Then, since $\sigma(\h_n^m)=1$,
\begin{equation}\label{eq_definition_sigma^M}
\sigma^M:=\frac{nM \sigma^u+\sigma}{nM+1}
\end{equation}
defines a probability measure on \(\h_n^m\) with non-negative mass and which is concentrated on \(\diag(\h_n^m)\,.\) Then, by \eqref{eq_def_mu_sigma}, the measure \(\mu^M\) defined by
$
\mu^M([0,v]):= \sigma^M([0,v]\times \cdots\times [0,v])\,,
$
is related to \(\sigma^M\) by \(G_{\sigma^M}(u_1,\ldots,u_m)=G_{\mu^M}(\min_{i=1,\ldots,m}\{u_i\})\,,\) \((u_1,\ldots,u_m)\in [0,1]^m\,.\) Hence, we obtain by \eqref{eqmadia3} that
$
\int_{[0,1]^m} f(u_1,\ldots,u_m)\de \sigma^M(u_1,\ldots,u_m) = \int_{[0,1]} f(u,\ldots,u) \de \mu^M(u)\,.
$
For \(\mu^u\) defined by \(\mu^u([0,v]):=\sigma^u([0,v]\times \cdots \times [0,v])\,,\) \(v\in [0,1]\,,\) we obtain, by using \eqref{eq_definition_sigma^M}, the identity \(\mu=(nM+1) \mu^M-nM \mu^u\,\). This yields
\begin{align*}
\int_{[0,1]^m} f(u) \de \sigma(u) &= (nM+1) \int_{[0,1]^m} f(u)\de \sigma^M(u) - nM \int_{[0,1]^m} f(u)\de \sigma^u(u) \\
&= (nM+1) \int_{[0,1]} f(v,\cdots,v) \de \mu^M(v)-nM \int_{[0,1]} f(v,\cdots,v) \de \mu^u(v) \\
&= \int_{[0,1]} f(v,\cdots,v) \de \mu(v)\,,
\end{align*}
which proves \eqref{eqmadia2}.
\end{proof}

\begin{proof}[Proof of Theorem \ref{theqcub}.]\label{protheamr}
\underline{(f) \(\Longrightarrow\) (c)}: For any fixed \(u=(u_1,\ldots,u_m)\in (0,1)^m\,,\) let $f(x):= \one_{\{u<x\}}$, $\widetilde{f}(x):= \one_{\{u
\leq x\}}$, $x=(x_1,\ldots,x_m)\in [0,1)^m\,,$ and let \((\Phi_n)_{n \in \N}\) be a sequence of \(N(u,I_m/n)\)-distribution functions, i.e., \(\Phi_n\) is the distribution function of the \(m\)-variate normal distribution with mean vector \(u\) and covariance matrix \(I_m/n\,,\) where \(I_m\) denotes the \((m\times m)\)-unit matrix.
Then, \(\Phi_n\) is \(\Delta\)-monotone and, thus, supermodular and measure-inducing for all \(n \in \N\,.\) 
Note that \(\eta_{\Phi_n}\to \eta_{\widetilde{f}}=\delta_{\{u\}}\)  weakly as \(n\to \infty\,,\) where \(\delta_{\{u\}}\) denotes the one-point probability measure in \(u\).
Moreover, note that $\eta_{\widetilde{f}}=\eta_f$.
Further, \(\phi_f\) defined via \eqref{defphif} by
\begin{align}\label{eqonpomeas}
\phi_f(x_1,x_2)=f(x_2,x_1,\ldots,x_1)=\one_{\{u_1<x_2,\max_{2\leq i\leq m}\{u_i\} < x_1\}}
\end{align}
is componentwise left-continuous and induces the one-point probability measure \(\eta_{\phi_f} = \delta_{\{\max_{2\leq i\leq m}\{u_i\},u_1\}} \,.\)
Thus, for the survival function of the upper product, it follows by \eqref{defsurvfun}, and since $M^2\vee D^2\vee \cdots \vee D^m$ is a copula, that
\begin{equation}\label{eq_m_dddd_hat_proof_eq_1}
\begin{aligned}
\reallywidehat{M^2\vee D^2\vee \cdots \vee D^m}\,(u) &= \int_{[0,1]^m} \one_{\{u < v\}} \de (M^2\vee D^2\vee \cdots \vee D^m)(v) \\
&= \lim_{n\to \infty} \int_{[0,1]^m} \Phi_n(v) \de (M^2\vee D^2\vee \cdots \vee D^m)(v) \\
& = \lim_{n\to \infty} \psi_{\Phi_n}(M^2\vee D^2\vee \cdots \vee D^m) \leq \lim_{n\to \infty} \pi_{\phi_{\Phi_n}} (\widehat{Q_2})\end{aligned}
\end{equation}
Indeed, to see that \eqref{eq_m_dddd_hat_proof_eq_1} holds, note that the second equality follows from the dominated convergence theorem and by using that $M^2\vee D^2\vee \cdots \vee D^m$ is continuous. The third equality is due to \eqref{defexpop_0} using that the upper product is a copula and thus continuous and measure-inducing. The inequality holds by assumption using that \(u\to \Phi_n(u,\ldots,u)\) is Lebesgue-integrable and that \(\phi_{\Phi_n}\) is \(\Delta\)-monotone and thus measure-inducing. Now, \eqref{eq_m_dddd_hat_proof_eq_1} implies
\begin{equation}\label{eq_m_dddd_hat_proof_eq_1_split}
\begin{aligned}
\reallywidehat{M^2\vee D^2\vee \cdots \vee D^m}\,(u) &\leq  \lim_{n\to \infty} \sum_{I\subseteq \{1,2\}\atop I\ne \emptyset} \int_{[0,1]^{|I|}} (\widehat{Q_{2}})_I(v) \de \eta_{(\phi_{\Phi_n})_I}(v) + \phi_{\Phi_n}(0,0)\\
&= \int_{[0,1]^2} \widehat{Q_2}(u) \de \eta_{\phi_f}(u)= \widehat{Q_2}(\max_{2\leq i \leq m}\{u_i\},u_1) = \widehat{Q^*}(u_1,\ldots,u_m)\,,
\end{aligned}
\end{equation}

Indeed, the first line follows from \eqref{defquexpop}. For the first equality, we apply that \((\eta_{\phi_{\Phi_n}})_n\) converges weakly to \(\eta_{\phi_f}\,,\) and that the measures \(\eta_{(\phi_{\Phi_n})_{\{1\}}}\) and \(\eta_{(\phi_{\Phi_n})_{\{2\}}}\) induced by the marginals of \(\phi_{\Phi_n}\) converge weakly to the null-measure because, as $u \in (0,1)^m$, \((\phi_{\Phi_n})_{\{1\}}(x)=\phi_{\Phi_n}(x,0)\to 0=\phi_f(x,0)= (\phi_f)_{\{1\}}(x)\) for all \(x\in [0,1]^m\,,\) and similarly for \((\eta_{(\phi_{\Phi_n})_{\{2\}}})_n\,.\) Further, we use that \(\phi_{\Phi_n}(0,0)\to 0=\phi_f(0,0)\,.\) The second equality follows from \eqref{eqonpomeas}, and the last equality holds due to Lemma \ref{lemsufu}.
Since $u \mapsto \reallywidehat{M^2\vee D^2\vee \cdots \vee D^m}\,(u)$ and $u \mapsto \widehat{Q^*}(u)$ are both continuous on $[0,1]^m$, we obtain that \eqref{eq_m_dddd_hat_proof_eq_1} holds also for $u\in [0,1]^m$.\\
\underline{(c) \(\Longrightarrow\)(a)}:
For \(i\in\{2,\ldots,m\}\) let \(u=(u_1,\ldots,u_m)\in [0,1]^m\) with \(u_j=0\) for all \mbox{\(j\in \{2,\ldots,m\}\setminus \{i\}\,.\)} Then, the survival function of \(Q_2\) satisfies that
\begin{equation}
\label{eq_iii_i_proof_equalities_1}
\begin{aligned}
\widehat{Q_2}(u_i,u_1) &=\widehat{Q_2} \left(\max_{2\leq j\leq m}\{u_j\},u_1 \right)= \widehat{Q^*}(u_1,\ldots,u_m)\geq \reallywidehat{M^2 \vee D^2 \vee \cdots \vee D^m}\,(u) \\
&= 1-\int_0^1 \max\{\one_{\{u_1>t\}},\partial_2 D^2(u_2,t),\ldots,\partial_2 D^m(u_m,t)\} \de t\\
&= 1- u_1 - \int_{u_1}^1 \max_{2\leq j\leq m}\{\partial_2 D^j(u_j,t)\} \de t = 1- u_1 - u_i + D^i(u_i,u_1)\,.
\end{aligned}
\end{equation}
Indeed, to see that \eqref{eq_iii_i_proof_equalities_1} holds, note that the second equality follows with Lemma \ref{lemsufu}. The inequality holds by assumption (c). The third equality follows with \cite[Proposition 2.4~(viii)]{Ansari-Rueschendorf-2018} by using that \(\partial_2 M^2(u_1,t)=\one_{\{u_1>t\}}\) for all $t\in [0,1]$ with \(t\ne u_1\,.\) The fourth equality is a consequence of \(0\leq \partial_2 D^j(u_j,t)\leq 1\) for Lebesgue-almost all \(t\in [0,1]\) and for \(j=2,\ldots m\,,\) see \cite[Theorem 2.2.7]{Nelsen-2006}. The last equality holds true by \cite[Theorem 2.2.7]{Nelsen-2006} because \(\partial_2 D^i(u_i,t)\geq \partial_2 D^i(0,t)=\partial_2D^j(u_j,t)\) for Lebesgue-almost all \(t\in [0,1]\) and for all \(j\ne i\,,\) using that \(D^2,\ldots,D^m\) are copulas. Hence, it follows that
\begin{align*}
D^i(u_i,u_1)\leq \widehat{Q_2}(u_i,u_1)-1+u_1+u_i = Q_2(u_i,u_1)\,.
\end{align*}
\underline{(a) \(\Longrightarrow\) (b):}
For \(u=(u_1,\ldots,u_m)\in [0,1]^m\,,\) we have
\begin{equation}\label{eq_i_ii}
\begin{aligned}
M^2\vee D^2\vee \ldots\vee D^m ~(u)&=\int_0^{u_1}\min_{2\leq i \leq m}\{\partial_2 D^i(u_i,t)\} \de t \leq \min_{2\leq i \leq m}\{D^i(u_i,u_1)\} \\
&\leq \min_{2\leq i \leq m}\{Q_2(u_i,u_1)\}= Q_2(\min_{2\leq i \leq m}\{u_i\},u_1) = Q^*(u).
\end{aligned}
\end{equation}
Indeed, to see that \eqref{eq_i_ii} holds, note that the first equality follows from the definition of the upper product, from \(\partial_2 M^2(u_1,t)=\one_{\{u_1>t\}}\)  and by \(0 \leq \partial_2 D^i(u_i,t) \leq 1\) for all \(i=1,\dots,m\) and for Lebesgue-almost all \(t\in [0,1]\,.\) The first inequality is a consequence of Jensen's inequality, the fundamental theorem of calculus, and property \eqref{def_quasicop1} in Definition \ref{def_quasicop} of copulas. The second inequality holds by assumption (a). The second equality follows because \(Q_2\) is a quasi-copula and, thus, non-decreasing in each argument.\\
\underline{(b) \(\Longrightarrow\) (a):}
Let $i\in \{2,\ldots,m\}\,.\)  For \(u=(u_1,\ldots,u_m)\in [0,1]^m\) such that \(u_j=1\) for all \(j\in \{2,\ldots,m\}\setminus \{i\}\) it follows that
\begin{align*}
D^i(u_i,u_1) = M^2\vee D^2\vee \cdots \vee D^m \,(u) \leq Q^*(u) = Q_2(u_i,u_1)\,,
\end{align*}
where the first equality is given by \cite[Proposition 2.4~(iv),(vi)]{Ansari-Rueschendorf-2018} and the inequality holds by assumption~(b).\\
\underline{((b) and (c)) \(\Longleftrightarrow\) (d):} This holds by the definition of the concordance ordering.\\

\underline{(a)\(\Longrightarrow\)(e)}:
We extend the proof of the main result in \cite[Chapter 3]{Ansari-Rueschendorf-2020} to a quasi-copula \(Q_2\in \cQ_2\) instead of a copula \(E\in \cC_2\,,\) cf.\,Remark \ref{remmaithean1}~(a).
Analogously, we first prove the statement in a discretized version using that all discretized copulas and quasi-copulas induce (signed) measures with finite support.
Then, we show the statement by an approximation of the discretized version, which differs from the proof of \cite[Theorem 1]{Ansari-Rueschendorf-2020} because we need to apply the quasi-expectation operator w.r.t.\,a quasi-copula instead of the expectation w.r.t.\,a probability measure.

For the first step, we make use of the same ideas and concepts as in the first part of the proof of \cite[Theorem 1]{Ansari-Rueschendorf-2020}, namely applying mass transfer theory from \cite{Mueller-2013} which requires a discretization of the distributions to a finite grid as follows. For \(n\in \N\) and \(m\geq 1\) denote by
\begin{align*}
\G_{n}^m:&= \left\{(\tfrac {i_1} n,\ldots,\tfrac {i_m} n) ~\middle|~ i_k\in \{1,\ldots,n\} \text{ for all } k\in \{1,\ldots,m\}\right\}\,,\\
\G_{n,0}^m:&= \left\{(\tfrac {i_1} n,\ldots,\tfrac {i_m} n) ~\middle|~ i_k\in \{0,\ldots,n\} \text{ for all } k\in \{1,\ldots,m\}\right\}
\end{align*}
the (extended) uniform unit \emph{\(n\)-grid} of dimension \(m\) with edge length \(\tfrac 1 n\,.\)

For the discretization of copulas and quasi-copulas, we use the concept of a \emph{(signed) $n$-grid $m$-copula} \(D\colon [0,1]^m \to \R\) which is the measure-generating function (see \eqref{eq_def_g_sigma}) of a (signed) measure \(\mu\) on \(\G_{n}^m\) that satisfies for all \(u=(u_1,\dots,u_m)\in [0,1]^m\)
\begin{enumerate}[(i)]
\item \(D(u)=D\left(\frac{\lfloor n u_1 \rfloor}{n},\dots,\frac{\lfloor n u_m \rfloor}{n}\right)=\mu\left([0,\frac{\lfloor n u_1 \rfloor}{n}]\times \cdots \times [0,\frac{\lfloor n u_m \rfloor}{n}]\right)\),  
\item for all \(i=1,\ldots,m\,,\) it holds \(D(u)=\tfrac k n\) for all \(k=0,\ldots,n\,,\) if \(u_i=\tfrac k n\) and \(u_j=1\) for all \(j\ne i\,,\)
\end{enumerate}
where \(\lfloor\cdot\rfloor\) is the componentwise floor function. Denote by \(\cC_{m,n}\) (respectively \(\cC_{m,n}^s\)) the set of all (signed) \(n\)-grid \(m\)-copulas. Note that, as discussed after \eqref{eq_def_g_sigma}, a (signed) \(n\)-grid \(m\)-copula induces a (signed) measure with support on the finite grid \(\G_n^m\,.\)
Hence, for every $m$-variate quasi-copula \(Q\in \cQ_m\,,\) the \emph{canonical \(n\)-grid quasi-copula} \(\G_n(Q')\) defined by
$
\G_n(Q')(u):=Q'(\tfrac{\lfloor nu \rfloor}{n})\,,~~~u\in [0,1]^m\,,
$
induces a signed measure.

For a function \(g\colon [0,1]^m\to \R\,,\) denote the difference operator of length \(\tfrac 1 n\) w.r.t.\,the \(i\)-th variable by
$
\delta_n^i g(u):= g(u)-g(\max\{u-\tfrac 1 n e_i,0\})\,,
$
 where \(e_i\) is the \(i\)-th unit vector. Then, we define the upper product \(\bigvee \colon (\cC_{2,n})^m \to \cC_{m,n}\) for grid copulas \(D_n^1,\ldots,D_n^m\in \cC_{2,n}\) by
\begin{align}\label{defdisuppprod}
\bigvee_{i=1}^m D_n^i (u_1,\ldots,u_m):&=\sum_{k=1}^n \min_{1\leq i \leq m}\left\{\delta_n^2 D_n^i(u_i,\tfrac k n)\right\}=\frac 1 n\sum_{k=1}^n \min_{1\leq i \leq m}\left\{n\delta_n^2 D_n^i(u_i,\tfrac k n)\right\}
\end{align}
for \((u_1,\ldots,u_m)\in [0,1]^m\,.\)

The upper product  for signed grid copulas \(D^1,\ldots,D^m\) is defined analogously where \(\bigvee_{i=1}^m D_n^i\in \cC_{m,n}^s\) whenever \(\delta_n^2 D_n^i(\cdot,t)\leq \tfrac 1 n\) for all \(i\in \{1,\ldots,m\}\) and \(t\in [0,1]\,.\)

Let \(D_n^i:=\G_n(D^i)\), $M_n^2:=\G_n(M^2)$, \(Q_{2,n}:=\G_n(Q_2)\), and \(Q_n^*:=\G_n(Q^*)\) be the canonical \(n\)-grid (quasi-)copulas of \(D^i\,,\) \(i=2,\ldots,m\,,\) $M^2$, \(Q_2\), and \(Q^*\,.\) Then it holds that \(Q_n^*(u)=Q_{2,n}(\min_{2\leq i \leq m}\{u_i\},u_1)\) and \(D_n^i(u_1,u_i)\leq Q_{2,n}(u_1,u_i)\) for all \(u=(u_1,\ldots,u_m)\in [0,1]^m\) and \(i\in \{2,\ldots,m\}\,.\)
Similar to \cite[Proof of Theorem 1]{Ansari-Rueschendorf-2020}, there exists for every \(n\in \N\) a finite sequence \((Q_{n,k}^*)_{0\leq k \leq n}\) of signed \(n\)-grid quasi-copulas such that
\begin{align}\label{eq_appendix_proof1}
Q_{n,0}^*&=M_n^2\vee D_n^2\vee\cdots \vee D_n^m\,,\qquad 
Q_{n,n}^*= Q_n^*\,,\\
\label{eqppp} Q_{n,k-1}^*&\leq_{\operatorname{sm}} Q_{n,k}^*~~~\text{for all }1\leq k \leq n ~\text{and for all } n\in \N\,.
\end{align}
Note that the supermodular ordering can also be defined w.r.t.\,finite signed measures \(\nu_1\) and \(\nu_2\) with finite support because the inequality \(\int f \de \nu_1\leq \int f \de \nu_2\) depends only the difference \(\nu_2-\nu_1\) by  \(\int f \de (\nu_2-\nu_1)\geq 0\) for \(f\in \cF_{\operatorname{sm}}\,.\)
So, the comparison in \eqref{eqppp} is well-defined because the expressions on both sides are signed grid copulas which correspond to signed measures with finite support \(\G_n^m\,.\) Hence, each \(Q_{n,k}^*\) induces a signed measure which we can integrate against.

Due to the transitivity of the supermodular ordering, \eqref{eq_appendix_proof1} and \eqref{eqppp} imply that \(M_n^2\vee D_n^2\vee\cdots \vee D_n^m\leq_{\operatorname{sm}} Q_n^*\) for all \(n\in \N\,,\) i.e.,
\begin{align}\label{mrdisve}
\int_{[0,1]^m} f(u)\de (M_n^2\vee D_n^2\vee\cdots \vee D_n^m)(u) \leq \int_{[0,1]^m} f(u) \de Q_n^*(u)~~~\,,
\end{align}
$\text{for all }f\in \cF_{\operatorname{sm}}$ such that the integrals exist. This proves the statement in the discretized version for grid copulas.

For the second step, let \(f\in \cF_{\operatorname{mi}}^{\operatorname{c},\operatorname{l}}([0,1]^m)\) be a left-continuous and supermodular function.
Note that \(f\) is bounded because it is measure-inducing and defined on a compact domain.
In the first step, we chose for notational conveniences the grid \(\G_n^m=\{\tfrac 1 n,\ldots,\tfrac {n-1}n,1\}^m\) as support of the discretized copulas and quasi-copulas
$
C_n:= M_n^2\vee D_n^2 \vee\cdots \vee D_n^m = M^2\vee D^2\vee\cdots\vee D^m\circ(F_n,\ldots,F_n)$ and $Q_n^*:= Q^*\circ (F_n,\ldots,F_n)\,$,
respectively, where \(F_n:[0,1]\rightarrow [0,1]$ is now defined as
\begin{align}\label{dffn}
F_n(x)&= \begin{cases}
0 & \text{if } x<\tfrac 1 {n+1}\,,\\
\tfrac k n &\text{if } x \in \big[\tfrac k {n+1}, \tfrac {k+1}{n+1}\big)\,, k=1,\ldots,n\,,\\
1& \text{if } x\geq 1\,.
\end{cases}
\end{align}
Note that the range of \(F_n\) is also \(\{0,\tfrac 1 n,\ldots,\tfrac {n-1}n,1\}\,.\)
Then, \(C_n\) and \(Q_n^*\) are distributions with finite support on the grid \(\K_n^m:=\{\tfrac 1 {n+1},\ldots,\tfrac n {n+1}\}^m\,.\) Applying mass transfer theory analogously to the first step, we also obtain \eqref{mrdisve}, now for the discretization w.r.t.\, $F_n$, i.e., we have $ \int_{[0,1]^m} f(u) \de C_n(u)  \leq  \int_{[0,1]^m} f(u) \de Q_n^*(u)$. Then, it follows that
\begin{equation}\label{eq_pi_f}
\begin{aligned} 
\pi_f(M^2\vee D^2\vee \cdots \vee D^m) &= 
\sum_{I\subseteq \{1,\ldots,m\} \atop I\ne \emptyset} \int_{[0,1]^m} \lim_{n\to \infty} (\widehat{C_n})_I(u) \de \eta_{f_I}(u) + f(0,\ldots,0) \\
&= \lim_{n\to \infty} \sum_{I\subseteq \{1,\ldots,m\} \atop I\ne \emptyset} \int_{[0,1]^m} (\widehat{C_n})_I(u) \de \eta_{f_I}(u) + f(0,\ldots,0)\\
&=\lim_{n\to \infty} \int_{[0,1]^m} f(u) \de C_n(u)  \\
& \leq \lim_{n\to \infty} \int_{[0,1]^m} f(u) \de Q_n^*(u) = \lim_{n\to \infty} \psi_{Q^*\circ(F_n,\ldots,F_n)}(f)= \pi_f(\widehat{Q^*})
\end{aligned}
\end{equation}
Indeed, to see that \eqref{eq_pi_f} holds, note that for the first equality, we apply \eqref{qeopexopeq} using that \(f\in \cF_{\operatorname{mi}}^{\operatorname{c},\operatorname{l}}\) and thus \(f_I\) induce a finite signed measures, \(I\subseteq\{1,\ldots,m\}\,.\) Further, we apply that the grid approximation \(C_n\) converges weakly and, thus, pointwise to \eqref{defquexpop} see \cite[Proposition 2.12]{Ansari-Rueschendorf-2018}, using that \(M^2\vee D^2\vee \cdots \vee D^m\) is a copula and thus continuous. Moreover, we use that $\reallywidehat{M^2\vee D^2\vee \cdots \vee D^m}(0,\cdots,0)=1$.
The second equality holds due to the dominated convergence theorem applying again that \(f\) induces a finite signed measure.
The third and fourth equality follow from \eqref{eqqop} using that the discretized copula \(C_n\) and the discretized quasi-copula \(Q_n^* = Q^*\circ(F_n\ldots,F_n)\) are right-continuous, grounded, bounded, measure-inducing, and fulfil the continuity conditions \eqref{contboun1} and \eqref{contboun2}.
The last equality follows from the approximation due to \cite[Theorem 3.7]{Ansari-2021}.
The inequality is a consequence of the discretized supermodular ordering result \eqref{mrdisve} in the modified version discretizing w.r.t. the grid \(\K_n^m\,.\) 

\underline{(e) \(\Longrightarrow\) (f)}:
Let \(f\colon[0,1)^m\to \R\) be lower bounded by some $M^2\vee D^2\vee \cdots \vee D^m$-integrable function, left-continuous, supermodular, and componentwise increasing/componentwise decreasing
such that \((\phi_f)_I\) is Lebesgue integrable on \([0,1)^{|I|}\) for \(I\subseteq\{1,2\}\,,\) \(I\ne \emptyset\,,\)

For \(n\in \N\,,\) let \(F_n\) be the distribution function given by \eqref{dffn}. We first show that
\begin{equation}\label{azw}
\pi_{f\circ(F_n^{-1},\ldots,F_n^{-1})} (\widehat{Q^*}) = \pi_{\phi_f \circ (F_n^{-1},F_n^{-1})}(\widehat{Q_2})\,.
\end{equation}
Define \(Q_{(n)}(u_1,\ldots,u_m):=Q^*(F_n(u_1),\ldots,F_n(u_m))\) and \(Q_{2,(n)}(u_1,u_2):=Q_2(F_n(u_2),F_n(u_1))\) for \(u_1,\ldots,u_m\in [0,1]\,.\) Then \(Q_{(n)}\) and \(Q_{2,(n)}\) induce by \eqref{eqindmeasu} finite signed measures on \(\cB([0,1]^m)\) and \(\cB([0,1]^2)\) with mass concentrated on the \(n\)-grid
\(\K_n^m=\{\tfrac 1 {n+1},\ldots,\tfrac {n} {n+1}\}^m\) and \(\K_n^2=\{\tfrac 1 {n+1},\ldots,\tfrac {n} {n+1}\}^2,\) respectively.
For \(u_1\in \h_n^1=\{\tfrac 1 {n+1},\ldots,\tfrac {n} {n+1}\},\) consider the conditional measure generating functions
$
Q_{(n)}^{u_1}(u_2,\ldots,u_m) := n \cdot \left[Q_{(n)}(u_1+\tfrac 1 {n+1},u_2,\ldots,u_m)- Q_{(n)} (u_1,u_2,\ldots,u_m)\right]$, and $
Q_{2,(n)}^{u_1}(u_2) := n \cdot \left[Q_{2,(n)}(u_1+\tfrac 1 {n+1},u_2)- Q_{2,(n)}(u_1,u_2)\right]$.
Then we obtain that
\begin{equation}\label{eq_pi_f_Finverse}
\begin{aligned}
\pi_{f\circ (F_n^{-1},\ldots,F_n^{-1})}(\widehat{Q^*}) &=\pi_f(\widehat{Q^*}\circ (F_n,\ldots,F_n))= \int_{[0,1]^m} f(u) \de Q_{(n)}(u)\\
&= \int_{[0,1]}\int_{[0,1]^{m-1}} f(u_1,u_2,\ldots,u_m) \de Q_{(n)}^{u_1}(u_2,\ldots,u_m) \de F_n(u_1)\\
&= \int_{[0,1]}\int_{[0,1]^{m-1}} f(u_1,u_2,\ldots,u_m) \de Q_{2,(n)}^{u_1}(\min_{2\leq i \leq m}\{u_i\}) \de F_n(u_1)\\
&= \int_{[0,1]}\int_{[0,1]^{m-1}} f(u_1,v,\ldots,v) \de Q_{2,(n)}^{u_1}(v) \de F_n(u_1)\\
&= \int_{[0,1]^m} \phi_f(v,u_1)\de Q_{2,(n)}(v,u_1)= \pi_{\phi_f}(\widehat{Q_2}\circ (F_n,F_n))=\pi_{\phi_f\circ (F_n^{-1},F_n^{-1})}(\widehat{Q_2})\,.
\end{aligned}
\end{equation}
Indeed, to see that \eqref{eq_pi_f_Finverse} holds, note that the first and last equality follow from the marginal transformation formula \eqref{margtrafgs}. Since \(Q_{(n)}\) and \(Q_{2,(n)}\) are defined on \([0,1)^m\) and \([0,1)^2\,,\) respectively, they are grounded and satisfy \(\widehat{Q^*}\circ (F_n,\ldots,F_n)=\widehat{Q_{(n)}}\) and \(\widehat{Q_2}\circ (F_n,F_n)=\widehat{Q_{2,(n)}}\). Hence, the second and seventh equality follows with \eqref{eqqop} using that \(Q_{(n)}\) and \(Q_{2,(n)}\) are right-continuous and measure-inducing. The third and sixth equality hold true by the disintegration theorem applied on the positive part and the negative part of the Hahn-Jordan decomposition of the signed measures induced by \(Q_{(n)}\) and \(Q_{2,(n)}\,,\) respectively. The fourth equality follows from \(Q^*(u_1,u_2,\ldots,u_m)=Q_2(\min_{2\leq i \leq m}\{u_i\},u_1)\,.\) The fifth equality holds by Lemma \ref{lemmadia} using that \(Q_{(n)}^{u_1}\) and \(Q_{2,(n)}^{u_1}\) are measure generating functions of signed measures with \(Q_{(n)}^{u_1}(u_2,\ldots,u_m)=Q_{2,(n)}^{u_1}(\min_{i=\{2,\ldots,m\}}\{u_i\})\) for all \(u_1\in \h_n^1\) and for all \((u_2,\ldots,u_m)\in \h_n^{m-1}\,.\) As a consequence of \eqref{azw}, we now obtain that
\begin{equation}
\label{appslbc}
\begin{aligned}
\begin{split}
\psi_f(M^2\vee D^2\vee \cdots \vee D^m) &\leq \liminf_{n\to \infty} \psi_{f\circ (F_n^{-1},\ldots,F_n^{-1})} (M^2\vee D^2\vee \cdots \vee D^m) \\
&= \liminf_{n\to \infty} \pi_{f\circ (F_n^{-1},\ldots,F_n^{-1})} (\reallywidehat{M^2\vee D^2\vee \cdots \vee D^m}) \\
& \leq \liminf_{n\to \infty} \pi_{f\circ (F_n^{-1},\ldots,F_n^{-1})}(\widehat{Q^*})= \liminf_{n\to \infty} \pi_{\phi_f\circ (F_n^{-1},F_n^{-1})}(\widehat{Q_2}) \\
&= \liminf_{n\to \infty} \pi_{\phi_f}(\widehat{Q_2}\circ(F_n,F_n)) = \liminf_{n\to \infty} \psi_{\phi_f}({Q_2\circ (F_n,F_n)})= \pi_{\phi_f}(\widehat{Q_2})\,,
\end{split}
\end{aligned}
\end{equation}
Indeed, to see that \eqref{appslbc} holds observe that the first inequality follows by an application of Fatou's Lemma using that \(f\) is lower bounded by some integrable function.
The first equality is a consequence of \eqref{qeopexopeq}. The second inequality holds true by assumption using that \(f\circ (F_n^{-1},\ldots,F_n^{-1})\) is left-continuous, supermodular, and measure-inducing. The third equality follows from the marginal transformation formula \eqref{margtrafgs}. The fourth equality follows from \eqref{eqqop} noting that \(Q_2\circ (F_n,F_n)\) is grounded. The last equality is a consequence of \cite[Corollary 3.12]{Ansari-2021}
using that \(\phi_f\in \cF_{\operatorname{mi}}^{\operatorname{c},\operatorname{l}}([0,1)^2)\) is left-continuous and \((\phi_f)_I\) is Lebesgue integrable for \(I\subseteq \{1,2\}\,,\) \(I\ne \emptyset\,,\)  as well as \(F_n(x)\to x\) for all \(x\in [0,1]\,.\)
This proves (f).

\end{proof}

\begin{proof}[Proof of Lemma~\ref{lem_indicator_constraints}]
We only prove the assertion of Lemma~\ref{lem_indicator_constraints}~(a). The assertion of Lemma~\ref{lem_indicator_constraints}~(b) follows analogously.

We observe that for each sequence $(x^{(N)})_{N \in \N} \subset  \mathds{Q}_+^{nd}$ with $x^{(N)} \downarrow x\in  \mathds{Q}_+^{nd}$ for $N \rightarrow \infty$ we have for all $\Q \in \cM(\mu)$ that 
\begin{equation}\label{eq_ineq_gxfx_1}
\lim_{N \rightarrow \infty} \E_\Q\left[g_{x^{(N)}}(S)\right]=\lim_{N \rightarrow \infty}\Q\left(S<x^{(N)}\right)=\Q \left(S\leq x\right)=\E_\Q[f_x(S)]
\end{equation}
and
\begin{equation}\label{eq_ineq_gxfx_2}
\lim_{N \rightarrow \infty} \overline{Q}\left(F_1^1\left({x_1^1}^{(N)}\right),\ldots,F_n^d\left({x_n^d}^{(N)}\right)\right) =\overline{Q}\left(F_1^1({x_1^1}),\ldots,F_n^d({x_n^d})\right).
\end{equation}
Thus, if we have 
\begin{equation}\label{eq_ineq_1}
\E_{\Q} \left[g_x(S)\right]\leq \phantom{-}\overline{Q}(F_1^1(x_1^1),\ldots,F_n^d(x_n^d)) \text{ for all } x=\left(x_1^1,\ldots,x_n^d\right)\in \mathds{Q}_+^{nd},
\end{equation}
then we may choose for each $x\in \mathds{Q}_+^{nd}$ a sequence $(x^{(N)})_{N \in \N} \subset  \mathds{Q}_+^{nd}$ with $x^{(N)} \downarrow x\in  \mathds{Q}_+^{nd}$ for $N \rightarrow \infty$, and  it follows with \eqref{eq_ineq_gxfx_1} and \eqref{eq_ineq_gxfx_2} that
\begin{equation}\label{eq_ineq_2}
\E_{\Q} \left[f_x(S)\right]\leq \phantom{-}\overline{Q}(F_1^1(x_1^1),\ldots,F_n^d(x_n^d))\text{ for all } x=\left(x_1^1,\ldots,x_n^d\right)\in \mathds{Q}_+^{nd}.
\end{equation}
Moreover, \eqref{eq_ineq_2} implies \eqref{eq_ineq_1} by definition of the respective indicator functions, thus  \eqref{eq_ineq_1} and \eqref{eq_ineq_2} are equivalent.
The assertion follows, since \(\underline{Q} \leq_{\operatorname{lo}} C_{\Q}\leq_{\operatorname{lo}} \overline{Q}\) is, by definition of the lower orthant order, equivalent to
$
\phantom{-}\E_{\Q} \left[f_x(S)\right]\leq \phantom{-}\overline{Q}(F_1^1(x_1^1),\ldots,F_n^d(x_n^d))$ and $
\E_{\Q} \left[-f_x(S)\right]\leq -\underline{Q}(F_1^1(x_1^1),\ldots,F_n^d(x_n^d))
$
for all $x=\left(x_1^1,\ldots,x_n^d\right)\in \mathds{Q}_+^{nd}$.
\end{proof}

\begin{proof}[Proof of Theorem~\ref{theloob}]
We only prove part (a), part~(b) follows analogously.\
Equation \eqref{theloob1a} is a consequence of Proposition \ref{theaddcon} and Lemma~\ref{lem_indicator_constraints}.
Inequality \eqref{eqlosb} follows from
\begin{equation}\label{eq_ineq_thm42_proof}
\overline{P}_{{\cM}^{\operatorname{lo}}_{\underline{Q},\overline{Q}}} = \sup_{\Q\in {\cM}^{\operatorname{lo}}_{\underline{Q},\overline{Q}}} \E_\Q[c(S)] \leq  \sup_{Q\in \cQ_{nd}\atop \underline{Q}\leq_{\operatorname{lo}} Q\leq_{\operatorname{lo}} \overline{Q}} \pi_c^{\mu}(\widehat{Q})=\pi_c^{\mu}(\widehat{\overline{Q}}),
\end{equation}
where we neglect the martingale property and the requirement that $\Q$ needs to be a probability measure for the inequality in \eqref{eq_ineq_thm42_proof}. The last equality is a consequence of the characterization of the lower orthant order for quasi-copulas in \eqref{charortord} noting that, by Remark \ref{propconncomdom}\,\eqref{propconncomdom2}, \(\pi_c^{\mu}(\widehat{\overline{Q}})\) exists because
\begin{align*}
\int_0^1 |c_I((F_{i}^{j_1})^{-1}(u),\ldots,(F_{i}^{j_k})^{-1}(u)) |\de u &\leq \alpha \int_0^1 \bigg(1+ \sum_{j\in I} | (F_i^j)^{-1}(u)|\bigg) \de u 
=\alpha \bigg(1+\sum_{j \in I} \E_{\mu_i^j}[|S_{t_i}^j|]\bigg)
\end{align*}
is finite for all \(I=\{j_1,\ldots,j_k\}\subseteq \{1,\ldots,nd\}\,,\) \(I\ne \emptyset\,,\) and for some \(\alpha>0\) using that \(c\in U_{\operatorname{lin}}(\R_+^{nd})\) and using that the first moments of \(\mu\) exist.
\end{proof}

For the proof of Theorem \ref{corsmifm}, we formulate an auxiliary lemma based, for some fixed time \(t_i\) and a quasi-copula \(Q\in \cQ_d\,,\) on the class
$
\overline{\cM}_{Q,t_i}^{\operatorname{lo}}(\mu):=\{\Q\in \cM(\mu)\, |\, C_{\Q_i}\leq_{\operatorname{lo}} Q\}
$
of probability measures \(\Q \in \cM(\mu)\) such that the copula \(C_{\Q_i}\in \mathcal{C}_d\) at time \(t_i\) defined by
$
C_{\Q_i}(F_i^1(x_1),\ldots,F_i^d(x_d))=\Q(S_{t_i}^1\leq x_1, \dots,S_{t_i}^d\leq x_d)$,~~~$x=(x_1,\ldots,x_d)\in \R^d\,,
$
is upper bounded by \(Q\in \cQ_d\) w.r.t.\,the lower orthant ordering, and where $\Q_i =  \Q \circ S_{t_i}^{-1}$. 
Note that we have \(\overline{\cM}^{\operatorname{lo}}_{M^d,t_i}(\mu)=\cM(\mu)\) in the case that no additional dependence restriction is included.

\begin{lem}\label{lemifm}
For \(Q_2\in \cQ_2\,,\) let \(Q^*\in \cQ_d\) be the \(d\)-variate quasi-copula given by \eqref{deqcqs}. Then it holds for all $i =1,\dots,n$ that
$
\cM_{Q_2,t_i}^{\operatorname{CCD}}(\mu) = \overline{\cM}_{Q^*,t_i}^{\operatorname{lo}}(\mu)\,.
$
\end{lem}

\begin{proof}[Proof of Lemma~\ref{lemifm}]
For $i \in \{1,\dots,n\}$ let \(\Q\in \cM_{Q_2,t_i}^{\operatorname{CCD}}(\mu)\) and \(D^k=C_{\Q_i^{1k}}\in \mathcal{C}_2\) be the copula associated with the bivariate \((1,k)\)-marginal 
\(\Q_i^{1k}\in \mathcal{P}(\R_+^2)\) of \(\Q_i=\Q\circ S_{t_i}^{-1}\,,\) \(2\leq k\leq d\,.\) 
Then, \cite[Proposition 2.4 (i)]{Ansari-Rueschendorf-2018} and Theorem \ref{theqcub} imply \(C_{\Q_i} \leq_{\operatorname{lo}} M^2\vee D^2\vee \cdots \vee D^d\leq_{\operatorname{lo}} Q^*\,,\) 
which means that \(\Q\in \overline{\cM}_{Q^*,t_i}^{\operatorname{lo}}(\mu)\,.\)\\
For the reverse inclusion, let \(\widetilde{\Q}\in \overline{\cM}_{Q^*,t_i}^{\operatorname{lo}}(\mu)\,.\) From the closure of the lower orthant ordering under marginalization\footnote{The lower orthant ordering is closed under marginalization in the sense that \(Q\leq_{\operatorname{lo}}Q'\,,\) \(Q,Q'\in \cQ_d\,,\) implies \(Q^I(u_1,\ldots,u_k)\leq_{\operatorname{lo}} Q'^I(u_1,\ldots,u_k)\) for all \(k=1,\ldots,d\,,\) \(I=\{i_1,\ldots,i_k\}\subseteq \{1,\ldots,d\}\,,\) and \(u_1,\ldots,u_k\in [0,1]^k\,,\) where the \(Q^I\,,\) and analogously \(Q'^I\,,\) is defined by \(Q^I(u_{i_1},\ldots,u_{i_k}):=Q(u_1,\ldots,u_d)\) for all \(u_1,\ldots,u_d)\in [0,1]^d\) with \(u_j=1\) whenever \(j\ne I\,.\)} we obtain for the copula of the bivariate \((1,k)\)-marginal distribution of \(\Q_i\) that \(C_{\Q_{i}^{1k}}\leq_{\operatorname{lo}} Q_2\,,\) \(2\leq k\leq d\,,\) which means that \(\widetilde{\Q}\in \cM_{Q_2}^{\operatorname{CCD}}(\mu)\,.\)
\end{proof}


\begin{proof}[Proof of Theorem~\ref{corsmifm}]
The statement \eqref{eqcorsmifm1} in Theorem~\ref{corsmifm}~(a) follows from Proposition \ref{theaddcon} and Lemma~\ref{lem_indicator_constraints} with the inequality constraints 
\begin{align*}
\E_\Q [g_{x,y}(S_{t_i}^1,S_{t_i}^k)]\leq Q_2(F_{i}^1(x),F_{i}^k(y))\,,~~~g_{x,y}:=\one_{\{\cdot < (x,y)\}}\,, (x,y)\in \mathds{Q}_+^2\,, 2\leq k\leq d\,.
\end{align*}
To prove \eqref{eqcorsmifm2}, first note that the definition of super-modularity \(\widetilde{c}\in \cF_{\operatorname{sm}} \cap C_{\operatorname{lin}}(\R_+^{nd})\) and that $\widetilde{c}$ is componentwise increasing/componentwise decreasing implies that $c=\left(\widetilde{c} \circ \operatorname{proj}_i^1,\cdots, \widetilde{c} \circ \operatorname{proj}_i^d\right) \in \cF_{\operatorname{sm}}\cap C_{\operatorname{lin}}(\R_+^{d})$ and that $c$ is componentwise increasing/componentwise decreasing. Moreover, we obtain by Lemma \ref{lemifm} that
\begin{equation}
\label{eqprcorsmifm}
\begin{aligned}%
\overline{P}_{\cM_{Q_2}^{\operatorname{\operatorname{CCD}}}} &= \sup_{\Q\in \cM_{Q_2}^{\operatorname{\operatorname{CCD}}}(\mu)} \E_\Q \left[c(S_{t_i}^1,\ldots,S_{t_i}^d)\right]
= \sup_{\Q\in \overline{\cM}_{Q^*,t_i}^{\operatorname{lo}}(\mu)} \E_\Q \left[c(S_{t_i}^1,\ldots,S_{t_i}^d)\right]\\
&\leq \sup_{C \leq_{\operatorname{lo}} Q^*} \psi_c^{(F_i^1,\ldots,F_i^d)}(C)
\leq \sup_{C=M^2\vee D^2\vee \cdots \vee D^d\,,\atop D^k\leq_{\operatorname{lo}} Q_2\,, k=2,\ldots,d} \psi_c^{(F_i^1,\ldots,F_i^d)}(C)\leq \pi_{\phi_{c\,\circ\left((F_i^1)^{-1},\ldots,(F_i^d)^{-1}\right)}} (\widehat{Q_2})\,.
\end{aligned}
\end{equation}
Indeed, to see that \eqref{eqprcorsmifm} holds, observe that we neglect the martingale property for the first inequality.
For the second inequality, let \(D^k\in \cC_2\,,\) \(2\leq k\leq d\,,\) such that the transposed copula \((D^k)'\in \mathcal{C}_2\) (defined by \((D^k)'(u,v):=D^k(v,u)\) for $u,v \in [0,1]$) is the bivariate \((1,k)\)-marginal copula of \(C\,.\) Then, \(C\leq_{\operatorname{lo}} Q^*\) implies \(D^k\leq_{\operatorname{lo}} Q_2\,.\) Since the upper product \(M^2\vee D^2\vee \cdots \vee D^d\) is the greatest element w.r.t.\,\(\leq_{\operatorname{sm}}\) in the class of copulas with bivariate \((1,k)\)-marginal specifications \(D^k\,,\) \(2\leq k\leq d\,,\) see \cite[Proposition 2.4]{Ansari-Rueschendorf-2018}, it follows that \(C\leq_{\operatorname{sm}} M^2\vee D^2\vee \cdots \vee D^d\,.\) This implies the second inequality using that \(c\) is supermodular. The last inequality is a consequence of Theorem \ref{theqcub}~(f) using that \(u\mapsto c_I\,\circ\left((F_i^{j_1})^{-1},\ldots,(F_i^{j_k})^{-1}\right) (u,\ldots,u)\) is Lebesgue-integrable for all \(I=\{j_1,\ldots,j_k\}\subseteq\{1,\ldots,d\}\,,\) \(I\ne \emptyset\,,\) since \(c\in C_{\operatorname{lin}}(\R_+^{d})\) and the first moments of \(F_i^1,\ldots,F_i^d\) exist, see \eqref{eq_ineq_thm42_proof}. Moreover, we use that $c\in C_{\operatorname{lin}}(\R_+^{d})$ which implies that $c\circ ((F_i^1)^{-1},\dots,(F_i^d)^{-1})$ is lower bounded by a function of the form $(x_1,\dots,x_d)\mapsto K\bigg(1+\sum_{j=1}^d\big|(F_i^j)^{-1}(x_j)\big|\bigg)$ for $K \in \R$, which is $M^2\vee D^2\vee \cdots \vee D^d$-integrable due to the existing first moments of the marginals.
\end{proof}

\begin{proof}[Proof of Lemma \ref{lembasopt}:]
(a): 
If \(d\leq 2\,,\) then \(\mathfrak{C}\) and \(\mathfrak{P}\) induce signed measures, see \cite[Table 1]{Tankov-2011}. If \(d\geq 3\,,\) then \(\mathfrak{P}\) and \(\mathfrak{C}\) do not induce signed measures because they can be linearly transformed into the lower Fr\'{e}chet bound \(W^d\) which is a quasi-copula that does not induce a signed measure, see \cite[Theorem 2.4]{Nelsen-2010}.

(b): We show the statement for the payoff function \(\mathfrak{C}\,.\) The proof for \(\mathfrak{P}\) follows analogously.
By definition of \(\phi\) and \(G^{-1}\,,\) we have for all $x_1,x_2 \in \R$ that 
\begin{align*}
\phi_{\mathfrak{C}\circ(F_1^{-1},\ldots,F_d^{-1})}(x_1,x_2)&= \left(\alpha_1F_1^{-1}(x_2)+\sum_{i=2}^d \alpha_i F_{i}^{-1}(x_1)-K\right)_+\\
&=\left(\alpha_1F_1^{-1}(x_2)+\sum_{i=2}^d \alpha_i G^{-1}(x_1)-K\right)_+=\phi_{\mathfrak{C}}\circ(G^{-1},F_1^{-1})(x_1,x_2)\,.
\end{align*}
First, consider the special case where for all $i=1,\dots,d$ the generalized inverse distribution functions \(F_i^{-1}\)  are continuous on the range of $F_i$, which is equivalent to \(F_i\) being strictly increasing. Then, we obtain that
\begin{equation}
\label{eq_psi_pi_proof}
\begin{aligned}
\psi_{\mathfrak{C}}^{(F_1,\ldots,F_d)} (M^2\vee D^2\vee \cdots \vee D^d) &= \psi_{\mathfrak{C}\circ(F_1^{-1},\ldots,F_d^{-1})}(M^2\vee D^2\vee \cdots \vee D^d)\\
&\leq
 \pi_{\phi_{\mathfrak{C}}\circ(G^{-1},F_1^{-1})}(\widehat{Q_2})= \pi_{\phi_{\mathfrak{C}}}^{(G,F_1)}(\widehat{Q_2})\,,
\end{aligned}
\end{equation}
where the first equality is given by \eqref{defexpop}. The inequality follows from Theorem \ref{theqcub}~(f) using that \(\mathfrak{C}\circ (F_1^{-1},\ldots,F_d^{-1})\) is continuous and increasing supermodular applying that it is an increasing transformations of the increasing supermodular function \(\mathfrak{C}\,.\) Note that \(\int_0^1 \mathfrak{C}\circ (F_1^{-1},\ldots,F_d^{-1})(u,\ldots,u) \de u\) exists because \(\mathfrak{C}\in C_{\operatorname{lin}}(\R_+^m)\) and the first moments of \(F_1,\ldots,F_d\) exist. Further, we use that \(\phi_{\mathfrak{C}}\) and thus \(\phi_{\mathfrak{C}}\circ (G^{-1},F_1^{-1})\) are continuous and hence, by (a), measure-inducing. The last equality follows from \eqref{defquexpop}.

In the case that for $i \in \N$ \(F_i\) is continuous (but not \(F_i^{-1}\)), approximate \(F_i\) pointwise by a sequence \((F_{i,n})_{n\in \N}\) of strictly increasing distribution functions supported on $\R_+$ with finite first moments such that $F_{i,n} \to F_i$, \(F_{i,n}\geq F_i\) pointwise and \(\int_{\R_+} x \de F_{i,n}(x)\to \int_{\R_+} x \de F_i(x)\) as \(n\to \infty\,.\) Note that for all \(i\in \{1,\ldots,d\}\,,\) the first moment of \(F_i\) exists by assumption.
We approximate $G$ pointwise by a sequence $(G_n)_{n \in \N}$ with \(G_n(x)\geq G(x)\) for all \(n\) and such that $G_n \rightarrow G$ pointwise for $n \rightarrow \infty$. Consider for \((U_1,\ldots,U_d)\sim M^2 \vee D^2 \vee \cdots \vee D^d\,,\) the random variables \(X_{n,i}:=F_{n,i}^{-1}(U_i)\) and \(X_i:=F_i^{-1}(U_i)\,.\) Then, by Scheffé's lemma,  we have that $X_{n,i}$ converge in $L^1$ to $X_i$ for all $i$. This implies also \(\sum_{i=1}^d X_{n,i} \to \sum_{i=1}^d X_i\) in \(L^1\) (see, e.g., \cite[Theorem 6.25]{Klenke-2020}) and thus for \(X_n:=(X_{n,1},\ldots,X_{n,d})\) and \(X:=(X_1,\ldots,X_d)\,,\) we obtain \(\mathfrak{C}(X_n)\to \mathfrak{C}(X)\) in \(L^1\,.\)

Since by Sklar's Theorem \(X_n\sim M^2\vee D^2\vee \cdots \vee D^d(F_{1,n},\ldots,F_{d,n})\) and \(X\sim M^2\vee D^2\vee \cdots \vee D^d(F_{1},\ldots,F_{d})\,,\) we then obtain that\begin{equation}
\label{eq_C_P_proof_1}
\begin{aligned}
\psi_{\mathfrak{C}}^{(F_1,\ldots,F_d)} &(M^2\vee D^2\vee \cdots \vee D^d)\\
&=\int \mathfrak{C}(x) \de (M^2\vee D^2\vee \cdots \vee D^d)(F_1(x_1),\ldots,F_d(x_d)) \\
&= \lim_{n\to \infty} \int \mathfrak{C}(x) \de (M^2\vee D^2\vee \cdots \vee D^d)(F_{1,n}(x_1),\ldots,F_{n,d}(x_d))
\end{aligned}
\end{equation}
Indeed, to see that \eqref{eq_C_P_proof_1} holds, note that the third equality is given by \eqref{defexpop}. Now, \eqref{eq_C_P_proof_1} implies
\begin{equation}
\label{eq_C_P_proof_1_split}
\begin{aligned}
\psi_{\mathfrak{C}}^{(F_1,\ldots,F_d)} (M^2\vee D^2\vee \cdots \vee D^d)&= \lim_{n\to \infty} \psi_{\mathfrak{C}}^{F_{1,n},\ldots,F_{d,n}} (M^2\vee D^2\vee \cdots \vee D^d) \leq \lim_{n\to \infty} \pi_{\phi_{\mathfrak{C}}}^{(G_n,F_{1,n})} (\widehat{Q_2})\\
&= \lim_{n\to \infty} \pi_{\phi_{\mathfrak{C}}} (\widehat{Q_2}\circ (G_n,F_{1,n}))= \pi_{\phi_{\mathfrak{C}}} (\widehat{Q_2}\circ (G,F_{1}))= \pi_{\phi_{\mathfrak{C}}}^{(G,F_{1})} (\widehat{Q_2}).
\end{aligned}
\end{equation}
Indeed, the inequality in \eqref{eq_C_P_proof_1_split} follows from the special case \eqref{eq_psi_pi_proof} where \(F_{i,n}^{-1}\) is continuous and integrable. The fourth and the last equality hold due to the notation in \eqref{defquexpop}, and the fifth equality is a consequence of the dominated convergence theorem using that \(\phi_{\mathfrak{C}}\) induces a positive measure and that \(F_{i,n}(x)\geq F_i(x)\) and \(G_n(x)\geq G(x)\) for all \(n\) and \(x\in \R\) as well as \(F_{i,n}(x)\to F_i(x),~G_n(x) \to G(x)\) for all \(x\in \R\) using the continuity of \(F_i\,,\) \(i\in \{1,\ldots,d\}\,\) and of $G$.
\end{proof}

\section*{Acknowledgments}
\noindent

We thank the editor, the associate editor and two anonymous reviewers for various constructive remarks that helped to significantly improve our paper. \\
Ariel Neufeld gratefully acknowledges  the financial support by the Nanyang Assistant Professorship Grant (NAP Grant) \emph{Machine Learning based Algorithms in Finance and Insurance}. 

\bibliography{literature}
\bibliographystyle{plain}

\appendix

\end{document}